%% file: main.tex
\documentclass[oneside,reqno]{amsart}

\usepackage[usenames,dvipsnames]{xcolor}

\usepackage[english]{babel}
\usepackage[utf8]{inputenc}
\usepackage[T1]{fontenc}

\usepackage{palatino}
\usepackage{mathpazo}
\usepackage{microtype}

\usepackage{tikz}
\usetikzlibrary{cd}
\usetikzlibrary{decorations.pathmorphing}

\usepackage{stmaryrd}
\usepackage{amsmath}
\usepackage{amsthm}
\usepackage{amssymb}
\usepackage{mathtools}
\usepackage{mathpartir}
\usepackage{environ}
\usepackage{enumerate}
\usepackage{bm}
\usepackage{bbm}
\usepackage{hyperref}
\usepackage{xspace}
\usepackage{stackengine}

\hypersetup{
  breaklinks=true,
  colorlinks=true
  }

\usepackage[nameinlink]{cleveref}

\usepackage[style=authoryear]{biblatex}
\calclayout{}

\theoremstyle{plain}

\newtheorem{thm}{Theorem}[section]
\newtheorem*{thm*}{Theorem}
\newtheorem{cor}[thm]{Corollary}
\newtheorem{lem}[thm]{Lemma}

\newtheorem{prop}[thm]{Proposition}

\newtheorem{innercustomasm}{Assumption}
\newenvironment{customasm}[1]
  {\renewcommand\theinnercustomasm{#1}\innercustomasm}
  {\endinnercustomasm}

\theoremstyle{definition}

\newtheorem{rem}[thm]{Remark}

\newtheorem{exa}[thm]{Example}

\newtheorem*{exas*}{Examples}
\newtheorem{defi}[thm]{Definition}

\newtheorem{conj}[thm]{Conjecture}

\newtheorem{clm}[thm]{Claim}

\newtheorem{constr}[thm]{Construction}

\input{macros.tex}
\input{lmacros.tex}

\author{Rafaël Bocquet}
\email{bocquet@inf.elte.hu}
\address{Department of Programming Languages and Compilers, Eötvös Loránd University, Budapest, Hungary}
\date{November 14, 2022}
\title{External univalence for second-order generalized algebraic theories}

\begin{document}

\maketitle

\begin{abstract}
  Voevodsky's univalence axiom is often motivated as a realization of the equivalence principle; the idea that equivalent mathematical structures satisfy the same properties.
  Indeed, in Homotopy Type Theory, properties and structures can be transported over type equivalences.
  However, we may wish to explain the equivalence principle without relying on the univalence axiom.
  For example, all type formers preserve equivalences in most type theories; thus it should be possible to transport structures over type equivalences even in non-univalent type theories.

  We define \emph{external univalence}, a property of type theories (and more general second-order generalized algebraic theories) that captures the preservation of equivalences (or other homotopy relations).
  This property is defined syntactically, as the existence of identity types on the (syntactically defined) coclassifying $(\reppre)$-CwF (also called generic model or walking model) of the theory.
  Semantically, it corresponds to the existence of some left semi-model structure on the category of models of the theory.
  We give syntactic conditions that can be used to check that a theory satisfies external univalence.
  We prove external univalence for some theories, such as the first-order generalized algebraic theory of categories, and dependent type theory with any standard choice of type formers and axioms, including identity types, $\Sigma$-types, $\Pi$-types, universes à la Tarski, the univalence axiom, the Uniqueness of Identity Proofs axiom, \etc.
\end{abstract}

\input{introduction.tex}
\input{background.tex}
\input{sogats.tex}
\input{homotopy.tex}
\input{reflexive_equivalences.tex}
\input{inverse_diagram.tex}
\input{univalence_proof.tex}
\input{applications.tex}
\input{conclusion.tex}

\appendix

\printbibliography

\end{document}

%% file: lmacros.tex
\newcommand{\reptype}{\mathsf{type}_{\rep}}
\newcommand{\repty}{\mathsf{ty}_{\rep}}

\newcommand{\RepTy}{\mathsf{RepTy}}

\newcommand{\BTy}{\mathsf{BTy}}
\newcommand{\BRepTy}{\mathsf{BRepTy}}
\newcommand{\MonoTy}{\mathsf{MonoTy}}
\newcommand{\PolyTy}{\mathsf{PolyTy}}

\newcommand{\GenTy}{\mathsf{GenTy}}
\newcommand{\GenRepTy}{\mathsf{GenRepTy}}

\newcommand{\lexpre}{\Sigma}
\newcommand{\reppre}{\Sigma,\Pi_{\rep}}

\newcommand{\repinfty}{\Sigma,\Pi_{\rep},{\Id_{ws}}}

\newcommand{\inner}[1]{{\color{RedOrange}{\mathbf{#1}}}}
\newcommand{\innerSym}[1]{{\color{RedOrange}{\bm{#1}}}}

\newcommand{\iS}{\inner{S}}
\newcommand{\iT}{\inner{T}}

\newcommand{\ity}{\inner{ity}}
\newcommand{\itm}{\inner{itm}}
\newcommand{\iTy}{\inner{iTy}}
\newcommand{\iTm}{\inner{iTm}}

\newcommand{\iid}{\inner{iid}}

\newcommand{\iId}{\inner{iId}}
\newcommand{\irefl}{\inner{irefl}}
\newcommand{\iJ}{\inner{iJ}}
\newcommand{\iJb}{\inner{iJ_{\mathnormal{\beta}}}}

\newcommand{\iEquiv}{\mathrm{iEquiv}}

\newcommand{\iUnit}{\innerSym{\Unit}}
\newcommand{\iSigma}{\innerSym{\Sigma}}
\newcommand{\ipi}{\innerSym{\pi}}
\newcommand{\iPi}{\innerSym{\Pi}}
\newcommand{\iapp}{\inner{iapp}}

\newcommand{\iob}{\inner{ob}}
\newcommand{\ihom}{\inner{hom}}
\newcommand{\ieqhom}{\inner{eqhom}}
\newcommand{\iOb}{\inner{Ob}}
\newcommand{\iHom}{\inner{Hom}}
\newcommand{\iEqHom}{\inner{EqHom}}
\renewcommand{\iid}{\inner{id}}
\newcommand{\icomp}{\inner{comp}}
\newcommand{\icirc}{\mathbin{\innerSym{\circ}}}

\newcommand{\ReflGraph}{\mathsf{ReflGraph}}
\newcommand{\PreReflGraph}{\mathsf{PreReflGraph}}

\newcommand{\CPreReflGraph}{\mathbf{PreReflGraph}}

\newcommand{\CPreReflEqv}{\mathbf{PreReflEqv}}

\newcommand{\CReflEqv}{\mathbf{ReflEqv}}

\newcommand{\Th}{\CT}

\newcommand{\Free}{\mathsf{Free}}
\newcommand{\Clos}{\mathsf{Clos}}
\newcommand{\isRefl}{\mathrm{isRefl}}
\newcommand{\isRepContr}{\isContr^{\rep}}
\newcommand{\isBContr}{\isContr^{b}}
\newcommand{\isBRepContr}{\isContr^{\rep,b}}

\newcommand{\contr}{\mathsf{contr}}
\newcommand{\ccenter}{\mathsf{center}}
\newcommand{\chpath}{\mathsf{hpath}}
\newcommand{\cpath}{\mathsf{path}}
\newcommand{\DId}{\mathsf{DId}}

\newcommand{\Colon}{\mathrel{::}}

\newcommand{\funl}{\overleftarrow{\mathsf{fun}}}
\newcommand{\funr}{\overrightarrow{\mathsf{fun}}}

\newcommand{\El}{\mathsf{El}}


%% file: introduction.tex
\section{Introduction}\label{sec:introduction}

The principle of equivalence, also called principle of isomorphism, equivalence-invariance, \etc{}, is the idea that all constructions (in some language or theory) should respect equivalences (for some notion of equivalence associated to the theory).
Structures and properties should be transportable over equivalences.
Voevodsky's univalence axiom can be seen as an internalization of this principle in the language of type theory.
However, univalence is a non-conservative extension of type theory, and incompatible with other useful type theoretic principles, such as Uniqueness of Identity Proofs (UIP).
We also wish to achieve transport over equivalent structures in non-univalent type theories.

In some ways, univalence is similar to parametricity.
Parametricity captures the preservation of $n$-ary relations, whereas univalence is related to the preservation of equivalences (which can be seen as binary relations that are functional in both directions).
While some theories satisfy internal parametricity, many others only satisfy parametricity externally.
External parametricity is a provable metatheoretic property of these theories.
In this paper, we introduce \emph{external univalence} a metatheoretic property of theories which captures the preservation of equivalences (or other homotopy relations).

The name ``external univalence'' corresponds to two ideas.
First, as already mentioned, the link between external and internal univalence is somewhat similar to the relationship between external and internal parametricity.
Secondly, external univalence is directly related to the other established use of the word ``univalence'', as found in the notion of univalent category~\parencite{UnivalentCategories}.
Indeed, a theory will satisfy external univalence when its generic model is univalent, for some suitable definition of univalent model of the given theory.

The original motivation for this paper is the author's work~\parencite{CoherenceStrictEqualities} on the conservativity of extensions of type theories by additional definitional equalities.
These conservativity results are proven by replacing definitional equalities by transports over equivalences and identifications (elements of the identity type).
It is then important to know that equivalences and identifications are preserved by everything in the theory.

We formulate external univalence for any second-order generalized algebraic theory (SOGAT) equipped with the data of homotopy relations on every sort.
SOGATs correspond to a class of type theories studied by Uemura~\parencite{UemuraFramework,UemuraThesis}.
The syntax and semantics of most type theories (including type theories with unusual contextual structure, such as cubical type theories and two-level type theories) can all be described using SOGATs.
All first-order generalized algebraic theories (GATs, such as the theory of categories) can also be seen as SOGATs.
The homotopy relations specify a notion of weak equality on every sort of the theory.

The GAT of categories has three sorts: objects, morphisms and equality between morphisms.
The homotopy relation on objects is given by isomorphisms, the homotopy relation on morphisms is given by equality of morphisms, and the homotopy relation on equality of morphisms is trivial.
Type theories are usually SOGATs with two sorts: types and terms.
The homotopy relations on types and terms can be given by respectively type equivalences and identifications.

We study properties of a SOGAT $\Th$ by focusing on its \emph{coclassifying} $(\reppre)$-CwF (which contains the \emph{generic model} or \emph{walking model} of $\Th$).
The coclassifying $(\reppre)$-CwF is also written $\Th$ and is often identified with the theory.
The underlying category of $\Th$ is equivalent to the category of finitely generated models of $\Th$.
However, we use a more syntactic definition of $\Th$, as the initial model of some two-level type theory.
The syntax (\ie{} the initial model) of $\Th$ embeds faithfully into its coclassifying $(\reppre)$-CwF; any property of $\Th$ has direct consequences on the syntax of the theory.
For many theories, the initial model is however too trivial to be interesting; for example, the initial category is empty.

In this setting, we say that a SOGAT $\Th$ satisfies external univalence if its coclassifying $(\reppre)$-CwF $\Th$ can be equipped with (weakly stable and weakly computational) identity types that are compatible with the specified homotopy relations.
The elimination principle of these identity types then gives transport over elements that are related by the homotopy relations.

More semantically, we will show in a future article that external univalence is equivalent to the existence of a left semi-model structure on the category of models of $\Th$, where the classes of cofibrations, fibrations and weak equivalences are determined by the theory $\Th$ and its homotopy relations.
In the case of the theory of categories, this semi-model structure is the canonical (or ``folk'') model structure on $\CCat$, while in the case of type theories with identity types, this semi-model structure is the one constructed by~\textcite{HomotopyTheoryTTs}.

Our main theorem states that external univalence can be proven for a theory by checking some syntactic conditions.
Checking these conditions requires providing witnesses of preservation of the homotopy relations by every operation of the theory, together with $1$- and $2$- dimensional cubical composition and filling operations for the homotopy relations.
In the restricted case of theories without equations (\eg type theories without computation rules), these conditions are actually necessary conditions.

Using this theorem, we show external univalence for:
\begin{itemize}
  \item the theory of categories, as a minimal example application of the method;
  \item type theory with identity types and any standard choice of additional type-theoretic structures, such as $\Pi$-types, $\Sigma$-types, universes à la Tarski (without a coding function), booleans, univalence, UIP, \etc.
\end{itemize}
Depending on the precise algebraic definition of universes, it is not always possible to prove external univalence in the absence of (internal) univalence.
A universe comes with a decoding function $\El$, that sends terms of the universe type to types.
A coding function is an inverse of the decoding function $\El$, universes with a coding function are also called Coquand universes~\parencite{CoquandPresheafNote,CoquandNormalizationDTT}.
In the absence of a coding function, it is always possible to prove external univalence.
\Textcite{MarriageUnivalenceParametricity} give some counter-examples to the preservation of equivalences in the absence of univalence, but they all rely on the use of universes à la Russell, which identify types with terms of the universe.

It should be possible to use our methods to show that other theories (such as the first-order generalized algebraic theory of $2$-categories, cubical type theories with or without Glue-types, \etc{}) also satisfy external univalence.

\subsection*{Example: the theory of categories}

We look in more details at external univalence in the setting of the generalized algebraic theory $\Th_{\CCat}$ of categories, which is perhaps the simplest theory with non-trivial homotopical content.
Since $\Th_{\CCat}$ is a \emph{first-order} generalized algebraic theory, its coclassifying $(\reppre)$-CwF $\Th_{\CCat}$ is in fact a coclassifying $\lexpre$-CwF.
It admits multiple equivalent definitions:
\begin{itemize}
  \item $\Th_{\CCat}$ is the initial category with families equipped with $\Unit$-types, $\Sigma$-types and with an ``internal category'';
  \item $\Th_{\CCat}$ is the initial model of a type theory with:
        \begin{itemize}
          \item A type $(\iob\ \type)$ of objects;
          \item A dependent type $((x,y : \iob) \vdash \ihom(x,y)\ \type)$ of morphisms;
          \item A dependent type $((x,y : \iob)\ (f,g : \ihom(x,y)) \vdash \ieqhom(f,g)\ \type)$ of equalities between morphisms;
          \item Such that types are closed under $\Unit$-types and $\Sigma$-types;
          \item A dependent term $(x : \iob) \vdash \iid(x) : \ihom(x,x)$;
          \item A dependent term $(f : \ihom(x,y))\ (g : \ihom(y,z)) \vdash \icomp(f,g) : \ihom(x,z)$;
          \item Such that the categorical laws are satisfied:
            \begin{alignat*}{1}
              & \icomp(f,\iid) = f, \\
              & \icomp(\iid,f) = f, \\
              & \icomp(\icomp(f,g),h) = \icomp(f,\icomp(g,h));
            \end{alignat*}
          \item And such that the type $\ieqhom(f,g)$ is propositional, and inhabited if and only if $f = g$;
        \end{itemize}
        The types and terms of this initial model are respectively the \emph{sorts} and \emph{elements} of the theory of categories.
  \item The objects of $\Th_{\CCat}$ are the categories that are finitely generated by a finite set of objects, a finite collection of morphisms between these objects, and a finite collection of equalities between compositions of these morphisms.
        The category $\Th_{\CCat}$ is a full subcategory of the $1$-category $\CCat$.

        The types of $\Th_{\CCat}$ over a finitely generated category $\Gamma$ are the ``diagram shapes'' over $\Gamma$; extensions of $\Gamma$ by a finite collection of new generating objects, morphisms and equalities.
        Equivalently, these are the functors into $\Gamma$ that have a finitely generated domain and are injective-on-objects (\ie{} that are cofibrations in the canonical model structure on $\CCat$).

        The terms of a diagram shape $A$ are the actual diagrams of that shape in the category that is finitely generated by $\Gamma$.
\end{itemize}

For example, the context (that can also be seen as a closed record type)
\[ \mathsf{Sect} = (x : \iob, y : \iob, r : \ihom(x,y), s : \ihom(y,x), p : \ieqhom(\icomp(s,r),\iid(y))) \]
is an object of $\Th_{\CCat}$.
The corresponding finitely generated category is the ``walking section''
\[
  \begin{tikzcd}
    y
    \ar[r, "s"']
    \ar[rr, "\iid(y)", bend left]
    &
    x
    \ar[r, "r"']
    &
    y \rlap{\ .}
  \end{tikzcd}
\]

We claim that $\Th_{\CCat}$ satisfies external univalence, meaning that the coclassifying $\lexpre$-CwF $\Th_{\CCat}$ can be equipped with identity types.
For any type $A$, \ie{} a diagram shape over $\Gamma$, the dependent type $\Id_{A}(x,y)$ is the diagram shape of isomorphisms between the two copies ($x$ and $y$) of the diagram $A$.

For example, the diagram shape $\Id_{\mathsf{Sect}}((x_{0},y_{0},r_{0},s_{0},p_{0}),(x_{1},y_{1},r_{1},s_{1},p_{1}))$ consists of the vertical isomorphisms in the following commutative diagram.
\[ \begin{tikzcd}
    y_{0}
    \ar[r, "s_{0}"']
    \ar[rr, "\iid(y_{0})", bend left]
    \ar[d, "\cong", "y_{2}"']
    &
    x_{0}
    \ar[r, "r_{0}"']
    \ar[d, "\cong", "x_{2}"']
    &
    y_{0}
    \ar[d, "\cong", "y_{2}"']
    \\
    y_{1}
    \ar[r, "s_{1}"]
    \ar[rr, "\iid(y_{1})"', bend right]
    &
    x_{1}
    \ar[r, "r_{1}"]
    &
    y_{1}\rlap{\ .}
  \end{tikzcd} \]

The elimination principle of the identity types then tells us that we can transport any diagram extension along such diagram isomorphisms.
For instance, if we know that $s_{0}$ and $r_{0}$ are actually inverses in the above setting, we can transport this fact over the diagram isomorphism to obtain that $s_{1}$ and $r_{1}$ are also inverses.
Concretely, we can form a dependent type $P$ over $\mathsf{Sect}$, with $P(x,y,r,s) = \ieqhom(\icomp(r,s), \iid(x))$.
If we have any term of type $P(x_{0},y_{0},r_{0},s_{0},p_{0})$, we obtain an element of type $P(x_{1},y_{1},r_{1},s_{1},p_{1})$ by transport.

Note that we do not include any type of ``equality between objects''.
Indeed, equalities between objects cannot be transported over diagram isomorphisms.

\subsection*{Example: Dependent type theories}

We also describe what external univalence entails for a dependent type theory $\Th$ with identity types, universes à la Tarski and any choice of standard type formers ($\Pi$-types, $\Sigma$-types, inductive types, \etc{}).

The coclassifying $(\reppre)$-CwF of $\Th$ can be described as the initial model of a two-level type theory with:
\begin{itemize}
  \item For every universe level $i$, we have:
        \begin{itemize}
          \item An outer type $(\ity_{i}\ \type)$ of inner types of level $i$.
          \item A dependent outer type $(A:\ity_{i} \vdash \itm_{i}(A)\ \type)$ of inner terms.
          \item An inner type $(\inner{\bm{U}}_{i} : \ity_{i+1})$ for the universe of $i$-small inner types and a dependent inner type $(A : \itm_{i+1}(\inner{\bm{U}}_{i}) \vdash \inner{El}(A) : \iTy_{i})$ for its decoding function.
        \end{itemize}
  \item The inner types and terms are closed under the operations of the dependent type theory $\Th$, including identity types $\iId$, $\irefl$, \etc.
  \item The outer types are closed under $\Unit$- and $\Sigma$-types.
  \item The outer types are closed under $\Pi$-types with arities in inner terms.
        This means that we have a type forming operation
        \begin{mathpar}
          \inferrule
          {A : \ity_{i} \\
            a : \itm_{i}(A) \vdash B(a)\ \type}
          {\Pi(A,B)\ \type}
        \end{mathpar}
        such that terms of type $\Pi(A,B)$ correspond bijectively to dependent terms $(a:\itm_{i}(A) \vdash b : B(a))$.
\end{itemize}
The underlying category of $\Th$ is equivalent to the category of all finitely generated contextual models of $\Th$.

The model $\Th$ does not coincide with the initial model $\Init_\Th$ of $\Th$; but there is a faithful embedding $\Init_\Th \to \Th$, so that anything constructed in $\Th$ is also valid in the syntax $\Init_\Th$.

In that setting, external univalence for $\Th$ says that the $(\reppre)$-CwF $\Th$ is equipped with (weakly stable) identity types with $\Id_{\ity_{i}}(A,B) \simeq \iEquiv(A,B)$ and $\Id_{\itm_{i}(A)}(x,y) \simeq \itm_{i}(\iId_{A}(x,y))$, where $\iEquiv(A,B)$ is the outer type of inner equivalences between the inner types $A$ and $B$.

A closed dependent inner type $(A:\ity_{i} \vdash P(A) : \ity_{i})$ in the $(\reppre)$-CwF $\Th$ is exactly a type expression that depends on a type variable $A$.
If $\Th$ satisfies external univalence, we know that any such $P$ perserves equivalences.
Indeed, $P$ has an action on paths:
\[ (A,B:\ity_{i}), E:\Id_{\ity_{i}}(A,B) \vdash \ap(P,E) : \Id_{\ity_{i}}(P(A),P(B)). \]
By external univalence, this is equivalent to an action of $P$ on equivalences:
\[ (A,B:\ity_{i}), E:\iEquiv(A,B) \vdash \ap(P,E) : \iEquiv(P(A),P(B)). \]

Furthermore, this action of $P$ on equivalences preserves composition of equivalences and the whole $\infty$-groupoid structure of types.

As an example, we can show how to transport the commutativity of addition from a type $\Nat_{\mathsf{un}}$ of unary natural numbers to a type $\Nat_{\mathsf{bin}}$ of binary natural numbers.
We consider the following dependent type:
\begin{alignat*}{3}
  & P && :{ }
  && (N : \ity_0) \times (\mathsf{plus} : \itm(N) \to \itm(N) \to \itm(N)) \to \ity_0, \\
  & P(N,\mathsf{plus}) && \triangleq{ }
  && \forall n\ m \to \iId_{N}(\mathsf{plus}(n,m),\mathsf{plus}(m,n)).
\end{alignat*}
We have an identification $E$ between $(\Nat_{\mathsf{un}},+_{\mathsf{un}})$ and $(\Nat_{\mathsf{bin}},+_{\mathsf{bin}})$ in the outer type $(N : \ity_0) \times (\mathsf{plus} : \itm(N) \to \itm(N) \to \itm(N))$.
By external univalence and function extensionality in $\Th$, this identification consists of an equivalence between $\Nat_{\mathsf{un}}$ and $\Nat_{\mathsf{bin}}$ along with a proof that it is compatible with $+_{\mathsf{un}}$ and $+_{\mathsf{bin}}$.
We can then use the action on paths of $P$ to obtain an identification $\ap(P,E)$ between $P(\Nat_{\mathsf{un}},+_{\mathsf{un}})$ and $P(\Nat_{\mathsf{bin}},+_{\mathsf{bin}})$.
By external univalence, we can also see $\ap(P,E)$ as an equivalence between $P(\Nat_{\mathsf{un}},+_{\mathsf{un}})$ and $P(\Nat_{\mathsf{bin}},+_{\mathsf{bin}})$.
Now given any term of type $P(\Nat_{\mathsf{un}},+_{\mathsf{un}})$, we can apply the equivalence to obtain a term of type $P(\Nat_{\mathsf{bin}},+_{\mathsf{bin}})$, \ie{} a proof of commutativity for the binary natural numbers.

Note that if $\Th$ has Coquand universes instead, given by inverses $(A:\ity_{i} \vdash \inner{c}(A) : \itm_{i+1}(\inner{\bm{U}}_i))$ of the coding functions $\inner{El}$, then external univalence implies internal univalence.
Indeed the action of $\inner{c}$ on paths is
\[ (A,B:\ity_{i}), E:\Id_{\ity_{i}}(A,B) \vdash \ap(\inner{c},E) : \Id_{\itm_{i+1}(\inner{\bm{U}}_i)}(\inner{c}(A),\inner{c}(B)). \]
By external univalence, this is equivalent to
\[ (A,B:\ity_{i}), E:\iEquiv(A,B) \vdash \ap(\inner{c},E) : \itm_{i+1}(\iId_{\inner{\bm{U}}_i}(\inner{c}(A),\inner{c}(B))), \]
\ie to the fact that any equivalence can be turned into an identification between elements of the universe $\inner{\bm{U}}_i$.

\subsection*{Related work}

\subsubsection*{Relational parametricity and the Identity Extension Lemma}

Reynolds' relational parametricity~\parencite{ReynoldsParametricity} provides an interpretation of the types of System F as binary relations for any given mapping of the type variables to relations.
A crucial property of Reynolds' model is the Identity Extension Lemma, which states that whenever all type variables are mapped to identity relations, the interpretation of any type is also the identity relation.

For dependent type theories, constructing models that satisfy the Identity Extension Lemma is generally challenging.
\Textcite{RelationallyParametricModelDTT} show that the Identity Extension Lemma can be motivated by the use of reflexive graphs in the construction of relationally parametric models.
They also construct a relationally parametric model of dependent types in reflexive graphs.
In that model, the universe of small types is interpreted as a universe of discrete and proof-irrelevant reflexive graphs.

Although the general setting differs, external univalence seems to be related to the Identity Extension Lemma, as our goal is to interpret every type as a (type-valued) relation that is also an identity type.
We also almost use reflexive graphs in our constructions, except that we have to replace the diagram shape of reflexive graphs by an inverse diagram shape (see~\cref{sec:mgraph_model}).

\subsubsection*{Univalent Parametricity}
\Textcite{MarriageUnivalenceParametricity} give a univalent parametricity translation for a type theory with the univalence axiom.
The univalence axiom is needed in their translation of the universes.
This translation allows for the transport of proofs and structures over equivalences.
In many instances, the transport is effective, meaning that the output term does not actually rely on the univalence axiom.
In these cases, the translated terms can be used even in non-univalent type theories.

Tabareau \etal implemented this univalent parametricity using the typeclass mechanism of Coq.
\Textcite{OrnamentsForProofReuse} also implemented a related transformation as a Coq plugin.

In our work, we show that the transport of structures over equivalences can be achieved even for non-univalent type theories, if their universes do not have a coding function.
Our constructions involve a homotopical inverse diagram model that is closely related to the univalent parametricity translation of Tabareau \etal.

We do not provide any algorithmic implementation of our results.
However, we work in a constructive metatheory, it is in principle possible to extract an algorithm from our proofs.
Furthermore most of our constructions involve syntactic manipulations that should be directly implementable.

\subsubsection*{Semi-model structures on categories of models of type theories}
\Textcite{HomotopyTheoryTTs} construct left semi-model structures on the categories of models of type theories with identity types, $\Sigma$-types, and (optionally) $\Pi$-types with function extensionality.
The properties of these semi-model structures can be used to transport structures over equivalences.
They prove the existence of the semi-model structures using several homotopical inverse diagram models.
For this purpose, \textcite{HomotopicalInverseDiagrams} have constructed homotopical inverse diagram models over arbitrary homotopical inverse categories.
Note that closely related homotopical gluing models had been constructed before by \textcite{ShulmanUnivalenceInverseDiagrams}.
\Textcite{IsaevModelStructuresOnModels} has also constructed some model structures on categories of models of type theories with an interval.

We will show in another paper that a SOGAT equipped with homotopy relations satisfies external univalence if and only if its category of models is a left semi-model category, for classes of trivial cofibrations, cofibrations and weak equivalences that are derived from the SOGAT and the chosen homotopy relations.
Thus our results will yield an alternative proof of the results of Kapulkin and Lumsdaine.
Since we prove external univalence for a large class of type theories, we will also obtain left semi-model structures for a large class of type theories.

\subsubsection*{Computing with univalence}
Proving external univalence for a theory essentially involves providing a computational explanation of univalence in a very restricted setting: the outer layer of the coclassifying $(\reppre)$-CwF of the theory.
It has only $\Sigma$-types, some $\Pi$-types, and some base types.
In particular, there is no universe classifying the outer types, so univalence cannot be iterated.
As a consequence, giving a computational explanation of external univalence is much simpler than for internal univalence.

Nevertheless, there are some similarities between our setting and computation with internal univalence.
Some of our constructions are reminiscent of the cubical type theory without an interval of \textcite{CubicalTypeTheoryWithoutInterval}, which was an early attempt at providing a computation interpretation of internal univalence.

\subsubsection*{Principle of equivalence}
Makkai's Principle of Isomorphism~\parencite{MakkaiTowardsCategoricalFoundation} is the idea that ``All grammatically correct properties of objects of a fixed category are to be invariant under isomorphism.''
These ideas were formally developed in the framework of First-Order Logic with Dependent Sorts~\parencite{MakkaiFOLDS}.
There was also prior work by \textcite{FreydPropertiesInvariant} and \textcite{BlancEquivalenceNaturelle}, showing that first-order categorical statements can be transported over equivalences of categories, as long as they do not mention equalities between objects.
In a recent talk, \textcite{HenryLanguageModelCat} has explained the relationship between these ideas and homotopy theory.

\textcite{HigherSIP} have revisited FOLDS in a univalent setting, and give a generic definition of ``indiscernability'' for any FOLDS-signature.
FOLDS-signatures can be identified with first-order generalized algebraic theories without operations.
It would be interesting to investigate whether indiscernabilities are homotopy relations that always satisfy external univalence in our setting.

\subsubsection*{\texorpdfstring{$\infty$}{infinity}-type theories}
We expressed external univalence using the structure of identity types on the coclassifying $(\reppre)$-CwF of a SOGAT $\Th$.
This coclassifying $(\reppre)$-CwF then has the structure of a model of type theory with $\Sigma$-types, (weakly stable) identity types and some $\Pi$-types.
In line with internal language conjectures~\parencite{HomotopyTheoryTTs,InternalLanguageLexInfty}, which assert that models of type theories with identity types and other structures are the internal languages of structured $\infty$-categories, the coclassifying $(\reppre)$-CwF ought to be the internal language of some $\infty$-category with representable maps.

\textcite{InftyTypeTheories} have used a precise definition of $\infty$-categories with representable maps as a notion of $\infty$-type theory.
Such an $\infty$-type theory has an $\infty$-category of models; in a model all substitution laws and computation rules only hold up to homotopy.
They have also established some coherence theorems that compare some $\infty$-type theories with some $1$-type theories.

Our results provide a way to work with objects that are morally $\infty$-type theories, without relying on any simplicial presentation of $\infty$-categories.
Instead we morally use a type-theoretic definition of (structured) $\infty$-categories, originally inspired by Brunerie's type-theoretic definition of $\infty$-groupoids~\parencite[Appendix~B]{BrunerieThesis}.


%% file: background.tex
\section{Background}\label{sec:background}
We work in a constructive metatheory.

\subsection{Notations}

We use different relation symbols for the different notions of identifications that occur in this paper.
We reserve the use of $(\sim)$ for homotopy relations associated to a theory (see~\cref{def:homotopy_relations}).
The symbol $(\simeq)$ is used for equivalences between types and identifications (terms of an identity type).
Isomorphisms are denoted by the symbol $(\cong)$.

\subsection{Factorization systems}

We recall some basic results on (both weak and orthogonal) factorization systems over locally finitely presentable categories.
We omit all proofs.
Details on locally presentable categories can be found in the standard reference book by~\textcite{LocallyPresentableAndAccessibleCategories}.
A general introduction to factorization systems can be found in notes by~\textcite{RiehlFactorizationSystems}.

We fix a locally finitely presentable category $\BC$.

\begin{defi}
  Let $l : A \to B$ and $r : X \to Y$ be two maps in $\BC$.
  We say that $l$ has the \defemph{left lifting property} with respect to $r$, or that $r$ has the \defemph{right lifting property} with respect to $l$ if for any square (lifting problem) of the form
  \[ \begin{tikzcd}
      A \ar[d, "l"] \ar[r, "f"] & X \ar[d, "r"] \\
      B \ar[r, "g"] & Y \rlap{\ ,}
    \end{tikzcd} \]
  there exists a diagonal map $h : B \to X$ such that $h \circ l = f$ and $g = r \circ h$.
  In that case we write $l \boxslash r$.

  We say that $l$ has the \defemph{unique left lifting property} with respect to $r$, when the diagonal filler $h$ is unique.
  This is also denoted by $l \perp r$
  \defiEnd{}
\end{defi}

\begin{prop}
  Given $l : A \to B$ and $r : X \to Y$, we have $l \perp r$ if and only if $l \boxslash r$ and $\nabla_{f} \boxslash r$ where $\nabla_{f} : B +_{A} B \to B$ is the codiagonal of $f$.
  \qed{}
\end{prop}

\begin{defi}
  A \defemph{weak factorization system} consists of two classes $\CL$ and $\CR$ of maps of $\BC$, such that
  \begin{alignat*}{1}
    & \CL = \{ l : A \to B \mid \forall r : X \to Y, l \boxslash r \}, \\
    & \CR = \{ r : X \to Y \mid \forall l : A \to B, l \boxslash r \},
  \end{alignat*}
  and such that every map can be factored as a map in $\CL$ followed by a map in $\CR$.
  \defiEnd{}
\end{defi}

\begin{defi}
  An \defemph{orthogonal factorization system} consists of two classes $\CL$ and $\CR$ of maps of $\BC$, such that
  \begin{alignat*}{1}
    & \CL = \{ l : A \to B \mid \forall r : X \to Y, l \perp r \}, \\
    & \CR = \{ r : X \to Y \mid \forall l : A \to B, l \perp r \},
  \end{alignat*}
  and such that every map can be factored as a map in $\CL$ followed by a map in $\CR$.
  \defiEnd{}
\end{defi}

\begin{prop}
  Any orthogonal factorization system is also a weak factorization system.
  Conversely, a weak factorization system is an orthogonal factorization system if and only if for every $f \in \CL$, $\nabla_{f} \in \CL$.
  \qed{}
\end{prop}

We fix a set $\CX$ of maps in $\BC$.
\begin{defi}
  An \defemph{$\CX$-cellular map} is a sequential composition of pushouts of coproducts of maps in $\CX$.
  A \defemph{$\CX$-cellular complex} is an object $A$ of $\BC$ such that the unique map $\Init_{\BC} \to A$ is an $\CX$-cellular map.

  A \defemph{finite $\CX$-cellular map} is a finite composition of pushouts of maps in $\CX$.
  A \defemph{finite $\CX$-cellular complex} is an object $A$ of $\BC$ such that the unique map $\Init_{\BC} \to A$ is a finite $\CX$-cellular map.
  \defiEnd{}
\end{defi}
We see $\CX$-cellularity as additional structure on the maps of $\BC$.
The cellular maps are usually defined as arbitrary transfinite compositions of pushouts of coproducts of maps in $\CX$; but since $\BC$ is locally \emph{finitely} presentable, it suffices to consider sequential compositions.

\begin{lem}[Small object argument]
  There is a weak factorization system on $\BC$, said to be cofibrantly generated by $\CX$.
  The right class of maps consists of maps with the right lifting property with respect to every map in $\CX$.
  The maps in the left class are the retracts of $\CX$-cellular maps.
  Furthermore, every map in $\BC$ factors as a $\CX$-cellular map followed by a map in the right class.
  \qed{}
\end{lem}

\begin{lem}[Small object argument for orthogonal factoriation systems]
  There is an orthogonal factorization system on $\BC$, generated by $\CX$.
  The maps in the right class are the maps with the unique right lifting property with respect to every map in $\CX$.
  As a weak factorization system, it is cofibrantly generated by
  \[ \CX \cup \{ \nabla_{f} \mid f \in \CX \}.
    \tag*{\qed{}}
  \]
\end{lem}

\subsection{Internal language of presheaf categories}\label{ssec:internal_language_psh}

We frequently use the type-theoretic internal languages of presheaf categories throughout this paper.

We use $\CPsh(\CC)$ to refer to the presheaf topos over $\CC$; it is a model of extensional type theory with a hierarchy of universes closed under many type-theoretic structures, including dependent products, dependent sums, extensional equality types, quotient types, \etc.

We use $\psh{\CC}$ to refer to the presheaf category over $\CC$.
It could be the underlying category of the topos $\CPsh(\CC)$, but we typically assume that $\psh{\CC}$ lives in a smaller universe than $\CPsh(\CC)$.

The types of $\CPsh(\CC)$ are the dependent presheaves; a dependent presheaf $Y$ over a presheaf $X$ is equivalently a presheaf over the category of elements $\int_{\CC} X$.

The universes of $\CPsh(\CC)$ are the Hofmann-Streicher universes; they classify the ($i$-small) dependent presheaves.
We denote them by $\UPsh_{i}$, or just $\UPsh$.

We often need to reason externally with objects that were defined in the internal language.
In that case, we borrow the following notations from crisp type theory~\parencite{ShulmanBrouwersFT}.
If $X$ is a presheaf over $\CC$, \ie{} a type of $\CPsh(\CC)$ over the empty context of $\CPsh(\CC)$, we write $x \Colon X$ to indicate that $x$ is a \emph{global} element of $X$.
When the category $\CC$ has a terminal object $\diamond$, this means that $x$ is an element of $X_{\diamond}$.
In particular, if $x \Colon \yo(\Gamma) \to X$, then $x$ is a global element of the exponential presheaf $(\yo(\Gamma) \to X)$, or equivalently an element of $X_{\Gamma}$ by the Yoneda lemma.
We leave implicit such uses of the Yoneda lemma.
Conversely, whenever we have a global element $x$ of $X$, we can use it in the internal language of $\CPsh(\CC)$ wherever an element of $X$ would be expected.

\subsubsection{Local representability}

We recall the notion of locally representable dependent presheaves, which is used to model context extensions.

\begin{defi}
  A dependent presheaf $Y$ over a presheaf $X$ is \defemph{locally representable} when for every element $x \Colon \yo(\Gamma) \to X$, the restricted presheaf
  \begin{alignat*}{3}
    & Y_{\mid x} && :{ } && {(\CC/\Gamma)}^{\op} \to \CSet, \\
    & Y_{\mid x}(\rho : \Delta \to \Gamma) && \triangleq{ } && Y_{\Delta}(x[\rho]).
  \end{alignat*}
  is representable.

  Its representing object consists of an extended context $\Gamma.Y_{\mid x}$ along with an isomorphism
  \[ \angles{\bm{p},\bm{q}} : \yo(\Gamma.Y_{\mid x}) \cong (\gamma:\yo(\Gamma)) \times Y(x(\gamma)). \]

  We will often denote the extended context by $(\gamma:\Gamma).Y(x(\gamma))$ and implicitly coerce through the isomorphism above.
  \defiEnd{}
\end{defi}

A dependent presheaf $Y$ is locally representable if and only if the corresponding total natural transformation $\Sigma_{X} Y \to X$ is a representable natural transformation~\parencite{AwodeyNaturalModels}.

There is a universe $\URepPsh$ classifying the locally representable dependent presheaves in $\CPsh(\CC)$, see for instance~\parencite{StreicherRepresentability} for a construction.

\subsection{Type-theoretic structures over internal families}
We now work internally to a presheaf topos $\CPsh(\CC)$.

\subsubsection{Internal families}

\begin{defi}
  A \defemph{family} is a pair $(\Ty,\Tm)$, where $\Ty : \UPsh$ and $\Tm : \Ty \to \UPsh$.
  It is said to have \defemph{representable elements} when $\Tm(A)$ is locally representable for any type $A$, \ie{} when $\Tm : \Ty \to \URepPsh$.
  \defiEnd{}
\end{defi}
The elements of $\Ty$ are often called types, and the elements of $\Tm$ are called terms.

\begin{defi}
  A \defemph{restriction} $\Ty' \to \Ty$ of a family $(\Ty,\Tm)$ consists of a presheaf $\Ty' : \UPsh$ along with a map ${\iota : \Ty' \to \Ty}$.
  It induces a \defemph{restricted family} $(\Ty',\Tm')$, with $\Tm'(A) = \Tm(\iota(A))$.

  A \defemph{subfamily} $\Ty' \hra \Ty$ is a restriction that is also a monomorphism.
  \defiEnd{}
\end{defi}
We will often leave $\iota$ implicit, especially when it is a monomorphism.

\subsubsection{Basic type-theoretic structures}\label{sssec:basic_tt_structures}

\begin{defi}
  A $\Unit$-type structure over a family $(\Ty,\Tm)$ consists of a type
  \begin{alignat*}{3}
    & \Unit && :{ } && \Ty
  \end{alignat*}
  along with an isomorphism
  \begin{alignat*}{1}
    & \Tm(\Unit) \cong \{\tt\}.
    \tag*{\defiEnd{}}
  \end{alignat*}
\end{defi}

\begin{defi}
  A $\Sigma$-type structure over a family $(\Ty,\Tm)$ consists of an operation
  \begin{alignat*}{3}
    & \Sigma && :{ } && (A : \Ty) (B : \Tm(A) \to \Ty) \to \Ty
  \end{alignat*}
  along with an isomorphism
  \begin{alignat*}{1}
    & \Tm(\Sigma(A,B)) \cong ((a : \Tm(A)) \times (b : \Tm(B(a)))).
    \tag*{\defiEnd{}}
  \end{alignat*}
\end{defi}

\begin{defi}
  A $\Pi$-type structure over a family $(\Ty,\Tm)$ consists of an operation
  \begin{alignat*}{3}
    & \Pi && :{ } && (A : \Ty) (B : \Tm(A) \to \Ty) \to \Ty
  \end{alignat*}
  along with an isomorphism
  \begin{alignat*}{1}
    & \Tm(\Pi(A,B)) \cong ((a : \Tm(A)) \to \Tm(B(a))).
    \tag*{\defiEnd{}}
  \end{alignat*}
\end{defi}

We will implicitly coerce through these isomorphisms.

\subsubsection{First-order $\Pi$-types}\label{sssec:first_order_pi_types}
We also need to consider a restriction of $\Pi$-types that will be used to describe the binders of type theories.

\begin{defi}
  The structure of \defemph{$\Pi$-types} in a family $(\Ty,\Tm)$ \defemph{with arities} in a family $(\Ty',\Tm')$ consists of an operation
  \begin{alignat*}{3}
    & \Pi && :{ } && (A : \Ty') (B : \Tm'(A) \to \Ty) \to \Ty
  \end{alignat*}
  along with an isomorphism
  \begin{alignat*}{1}
    & \Tm(\Pi(A,B)) \cong ((a : \Tm'(A)) \to \Tm(B(a))).
    \tag*{\defiEnd{}}
  \end{alignat*}
\end{defi}

\begin{defi}
  The structure of \defemph{first-order $\Pi$-types} in a family $(\Ty,\Tm)$ consists of a restricted family $\RepTy \to \Ty$, along with $\Pi$-types in $(\Ty,\Tm)$ with arities in $\RepTy$.
  \defiEnd{}
\end{defi}

The intuition here is that $\Ty$ is the family of first-order types, while $\RepTy$ is its restricted family of zeroth-order types.
The domain of a first-order $\Pi$-type has to be a zeroth-order $\Pi$-type.
Elements of $\RepTy$ will also be called \emph{representable types}, since they will typically be interpreted as locally representable dependent presheaves.
We sometimes use $\Pi_{\rep}$ to refer to the first-order $\Pi$-types.

We use these first-order $\Pi$-types in the definition of second-order generalized algebraic theories; in a second-order theory, the domain of an operation can be any first-order type.

\begin{exa}\label{exa:psh_reppi}
  A presheaf topos $\CPsh(\CC)$ is equipped with first-order $\Pi$-types, where the representable types are the locally representable dependent presheaves.
  The first-order $\Pi$-types could be defined to be the usual $\Pi$-types of the presheaf topos, but there is also an alternative definition that relies on the local representability of the domain.
  Indeed, if $X$ is a presheaf, $Y$ is a dependent presheaf over $X$ and $Z$ is a dependent presheaf over $\Sigma_{X}Y$, we can pose
  \begin{alignat*}{3}
    & {\Pi(Y,Z)}_{\Gamma}(x) && \triangleq{ } && Z_{(\gamma:\Gamma).Y(x(\gamma))}(\lambda (\gamma,-) \mapsto x(\gamma)).
  \end{alignat*}
  In the simply-typed case, this was first observed by \textcite{SemanticalAnalysisHOAS}.

  Because the two definitions satisfy the same universal property, they are interchangeable.
  However the alternative definition gives a first-order algebraic presentation of the categories of models of algebraic theories with binders, ensuring that the category of models is locally finitely presentable and the existence of initial models.
  \defiEnd{}
\end{exa}

\subsubsection{Telescopes}\label{sssec:cwfs_telescopes}

Given any family $(\Ty,\Tm)$, we can consider the family $(\Ty^{\star},\Tm^{\star})$ of \emph{telescopes};
the notation $\Ty^{\star}$ is inspired from the notation $A^{\star}$ for the set of lists of elements of a set $A$.
The elements of $\Ty^{\star}$ are finite dependent sequences
\[ (A_{1} : \Ty, A_{2} : \Tm(A_{1}) \to \Ty, A_{3} : (a_{1} : \Tm(A_{1})) \to (a_{2} : \Tm(A_{2}(a_{1}))) \to \Ty, \dotsc) \]
of types, and elements of $\Tm^{\star}(A)$ are sequences
\[ (a_{1} : \Tm(A_{1}), a_{2} : \Tm(A_{2}(a_{1})), a_{3} : \Tm(A_{3}(a_{1},a_{2})), \dotsc) \]
of terms of the types of the sequence $A$.

\begin{defi}
  The family $(\Ty^{\star},\Tm^{\star})$ is defined by induction-recursion as follows:
  \begin{alignat*}{3}
    & \Ty^{\star} && :{ } && \UPsh, \\
    & \Tm^{\star} && :{ } && \Ty^{\star} \to \UPsh, \\
    & \diamond && :{ } && \Ty^{\star}, \\
    & \Tm^{\star}(\diamond) && \triangleq{ } && \top, \\
    & \_{}. \_{} && :{ } && (A : \Ty^{\star}) \to (\Tm^{\star}(A) \to \Ty^{\star}) \to \Ty^{\star}, \\
    & \Tm^{\star}(A. B) && \triangleq{ } && (a : \Tm^{\star}(A)) \times \Tm(B(a)).
    \tag*{\defiEnd{}}
  \end{alignat*}
\end{defi}

The family of telescopes can be equipped with (strictly associative and unital) $\Sigma$-types, given by concatenation of the sequences of types.
When the base family $(\Ty,\Tm)$ has $\Sigma$-types, there is a family morphism $\Ty^{\star} \to \Ty$ that interprets telescopes as (either left-nested or right-nested) iterated $\Sigma$-types.

\subsection{Categories with Families}

We now return to an external setting.
The internal notions of type-theoretic structures yield external notions of models equipped with these type-theoretic structures.
More precisely, these models are categories with families (CwFs,\cite{InternalTypeTheory,CwFsUSD}).

Note that for most of this paper, CwFs are not directly used as the notion of model of type theory, but rather as worlds in which the notion of model of type theory can be interpreted.
In other words, they do not correspond to the object theories we are interested in, but rather to logical frameworks in which the object theories can be specified and interpreted.
Accordingly, while we study arbitrary object theories, the CwFs will only be equipped with a handful of structures ($\Sigma$-types, (first-order) $\Pi$-types, and some identity types).
These correspond~\parencite{BiequivalenceLCCC} to well-known classes of structured categories, such as clans, finitely complete categories, representable map categories, locally cartesian closed categories, \etc{}

\begin{defi}
  A \defemph{category with families} (CwF) $\CC$ is a category, equipped with a terminal object, along with a global family $(\Ty_{\CC}, \Tm_{\CC})$ with representable elements in $\CPsh(\CC)$.
  \defiEnd{}
\end{defi}
We have a locally finitely presentable $1$-category $\CCwf$ of CwFs and strict CwF morphisms.

\begin{defi}
  A $\lexpre$-CwF is a CwF whose family is equipped with $\Unit$- and $\Sigma$- types.
  \defiEnd{}
\end{defi}
We write $\CCwf_{\lexpre}$ for the $1$-category of $\lexpre$-CwFs.

\begin{defi}
  A $(\reppre)$-CwF is a CwF equipped with:
  \begin{itemize}
    \item A restriction $\RepTy \to \Ty$ inducing a family of \defemph{representable types} (or first-order types).
    \item First-order $\Pi$-types with respect to $\RepTy \to \Ty$.
    \item Along with $\Unit$- and $\Sigma$- type structures over the families $\Ty$ and $\RepTy$.
          They do not have to be strictly preserved by $\RepTy \to \Ty$ (but they are automatically preserved up to isomorphism).
          \defiEnd{}
  \end{itemize}
\end{defi}
We write $\CCwf_{\reppre}$ for the $1$-category of $(\reppre)$-CwFs and strict morphisms.

\begin{exa}
  Any presheaf category $\psh{\CC}$ is equipped with the structure of a $(\reppre)$-CwF where:
  \begin{itemize}
  \item The types are the dependent presheaves.
  \item The representable types are the locally representable dependent presheaves.
  \item The first-order $\Pi$-types are defined as in~\cref{exa:psh_reppi}.
    \defiEnd{}
  \end{itemize}
\end{exa}




\subsection{Identity types}

We now recall the definitions of some classes of identity types.
We only use identity types with an elimination rule \emph{à la Paulin-Mohring}, also called based path induction.
We only work with \emph{weak} identity types, whose computation rule only holds up to a path.

We use both strictly stable and weakly stable variants of the identity type.
In presence of either variant, we have well-behaved notions of contractibility, equivalence, transport, \etc that we don't explicitely introduce.

\subsubsection{Weak identity types}

\begin{defi}[Weak identity types]\label{def:weak_identity_types}
  The structure of \defemph{weak identity types} over an internal family $(\Ty,\Tm)$ consists of four components $\Id$, $\refl$, $\J$, $\Jbeta$ with the following signature:
  \begin{alignat*}{3}
    & \Id && :{ } && \forall (A : \Ty)\ (x,y : \Tm(A)) \to \Ty, \\
    & \refl && :{ } && \forall (A : \Ty)\ (x : \Tm(A)) \to \Tm(\Id(A,x,x)), \\
    & \J && :{ } && \forall (A : \Ty)\ (x : \Tm(A)) \\
    &&&&& \phantom{\forall} (P : (y:\Tm(A))(p:\Tm(\Id(A,x,y))) \to \Ty)\ (d : \Tm(P(x,\refl(x)))) \\
    &&&&& \to \forall y\ p \to \Tm(P(y,p)), \\
    & \Jbeta && :{ } && \forall A\ x\ P\ d \to \Tm(\Id(P(x,\refl(x)), \J(A,x,P,d,x,\refl(x)), d)).
    \tag*{\defiEnd{}}
  \end{alignat*}
\end{defi}

Once a CwF has weak identity types, many notions can be derived, such as composition of paths, the action on paths of a function, the notion of contractibility, \etc.
They can be defined mostly in the same way as in the HoTT book~\parencite{hottbook}, although some additional effort is needed to deal with the absence of the strict $\beta$-rule for $\J$ and with the lack of $\Sigma$- and $\Pi$- types.

\subsubsection{Weakly stable identity types}

We will only consider weakly stable identity types with a weak computation rule.
We fix a base CwF $\CC$.

\begin{defi}[Weakly stable identity types]
 A \defemph{$\Id$-introduction context} is a triple $(\Gamma,A,x)$, where
  \begin{alignat*}{3}
    & \Gamma && :{ }
    && \Ob_{\CC}, \\
    & A && \Colon{ }
    && \yo(\Gamma) \to \Ty, \\
    & x && \Colon{ }
    && (\gamma : \yo(\Gamma)) \to \Tm(A(\gamma)).
  \end{alignat*}
  Here $\Gamma$ is an object of $\CC$, and $A$ and $x$ are types and terms that only depend on $\Gamma$.

  A \defemph{weakly stable identity type introduction structure} consists, for every $\Id$-introduction context $(\Gamma,A,x)$, of operations
  \begin{alignat*}{3}
    & \Id_{(\Gamma,A,x)} && \Colon{ }
    && \forall (\gamma : \yo(\Gamma)) (y : \Tm(A(\gamma))) \to \Ty, \\
    & \refl_{(\Gamma,A,x)} && \Colon{ }
    && \forall (\gamma : \yo(\Gamma)) \to \Tm(\Id_{(\Gamma,A,x)}(\gamma,x)).
  \end{alignat*}

  A \defemph{$\Id$-elimination context} over an $\Id$-introduction context $(\Gamma,A,x)$ is a tuple $(\Delta,\gamma,P,d)$, where
  \begin{alignat*}{3}
    & \Delta && :{ } && \Ob_{\CC}, \\
    & \gamma && :{ } && \Delta \to \Gamma, \\
    & P && \Colon{ } && \forall (\delta : \yo(\Delta)) (y : \Tm(A(\gamma(\delta)))) (p : \Tm(\Id_{(\Gamma,A,x)}(\gamma(\delta), y))) \to \Ty, \\
    & d && \Colon{ } && \forall (\delta : \yo(\Delta)) \to \Tm(P(\delta, x(\gamma(\delta)), \refl_{(\Gamma,A,x)}(\gamma(\delta), x(\gamma(\delta))))).
  \end{alignat*}

  A \defemph{weakly stable identity type elimination structure} consists, for every $\Id$-elimination context $(\Delta,\gamma,P,d)$ over $(\Gamma,A,x)$, of operations
  \begin{alignat*}{3}
    & \J_{(\Gamma,A,x,\Delta,\gamma,P,d)} && \Colon{ }
    && \forall (\delta : \yo(\Delta)) (y : \Tm(A(\gamma(\delta)))) (p : \Tm(\Id_{(\Gamma,A,x)}(\gamma(\delta), y))) \to \Tm(P(\delta,y,p)), \\
    & \Jbeta_{(\Gamma,A,x,\Delta,\gamma,P,d)} && \Colon{ }
    && \forall (\delta : \yo(\Delta)) \to
       \Tm(\Id_{(\Delta,P',d)}
       (\delta, \J_{(\Gamma,A,x,\Delta,\gamma,P,d)}(\delta, x(\gamma(\delta)), \refl_{(\Gamma,A,x)}(\gamma(\delta))))), \\
    & P'(\delta') && \triangleq{ }
    && P(\delta', x(\delta'), \refl_{\Gamma}(\gamma(\delta'), x(\delta'))).
  \end{alignat*}

  A \defemph{weakly stable identity type structure} consists of introduction and elimination structures.
  \defiEnd{}
\end{defi}

\begin{defi}\label{def:weakly_stable_contr}
  Let $\CC$ be a CwF that is equipped with weakly stable identity types.
  Given a type $A \Colon \yo(\Gamma) \to \Ty_\CC$, the set $\isContr(A)$ of witnesses of contractibility of $A$ is defined as
  \[ \isContr(A) \triangleq (\forall \gamma \to \Tm_\CC(A)) \times (\forall \gamma\ (x,y : \Tm_\CC(A)) \to \Tm_\CC(\Id_{(\Gamma.A, A)}((\gamma,x),y))).
    \tag*{\defiEnd{}}
  \]
\end{defi}

\begin{defi}
  Let $\CC$ be a $(\reppre)$-CwF that is also equipped with weakly stable identity types.
  We say that $\CC$ satisfies \defemph{function extensionality} if for every $\Gamma : \Ob_{\CC}$, $A \Colon \yo(\Gamma) \to \RepTy_{\CC}$, $B \Colon (\gamma : \yo(\Gamma)) \to \Tm_{\CC}(A(\gamma)) \to \Ty_{\CC}$ and $f \Colon (\gamma : \yo(\Gamma)) \to \Tm_{\CC}((a : A(\gamma)) \to B(\gamma,a))$, the type
  \[ (g : (a : A(\gamma)) \to B(\gamma,a)) \times ((a : A(\gamma)) \to \Id_{((\gamma':\Gamma).(a':A(\gamma')), B(\gamma',a'), f(\gamma',a'))}((\gamma,a),g(a))) \]
  is contractible (over $(\gamma:\Gamma)$).
  \defiEnd{}
\end{defi}

We write $\CCwf_{\repinfty}$ for the category of $(\reppre)$-CwFs equipped with weakly stable identity types that satisfy function extensionality.
A $(\repinfty)$-CwF can be thought of as an $\infty$-category with representable maps.


%% file: sogats.tex
\section{Second-order generalized algebraic theories}\label{sec:sogats}

We introduce our definition of second-order generalized algebraic theory (SOGAT), which are algebraic theories with dependent sorts and bindings.
It is closely related to Uemura's general definition of type theory with functorial semantics in representable map categories~\parencite{UemuraFramework}; a large part of the material presented in this section can be found in Uemura's work, with a different presentation.
We call these theories SOGATs rather than type theories to emphasize that we also consider theories that are not usually seen as type theories, such as the (first-order) generalized algebraic theory of categories.
We note that Uemura uses the term SOGAT to refer to syntactic presentations of representable map categories in his thesis~\parencite{UemuraThesis}.

Our definition differs from Uemura's definition in the following ways:
\begin{itemize}
  \item Uemura's representable map categories have all finite limits.
        This means that they generalize essentially algebraic theories (EATs) rather than generalized algebraic theories (GATs).
        Essentially algebraic theories do not have dependent sorts, but allow for partial operations instead.
        Any generalized algebraic theory induces an essentially algebraic theory with an equivalent category of models, but this translation loses information about the sort dependencies.
        This information is important; for example it equips the category of models of a generalized algebraic theory with notions of cofibrations and trivial fibrations (see~\cref{ssec:trivial_fibrations_sogats}).

  \item We use $(\reppre)$-CwFs instead of representable map categories.
        This is partially a matter of preference, as $(\reppre)$-CwFs ought to be equivalent to categories with classes of representable maps and display maps (``representable map clans'').
        One advantage of our approach is that freely generated $(\reppre)$-CwFs are perhaps easier to understand syntactically, since they are themselves the initial models of some type theories.
        Furthermore, we may embed $(\reppre)$-CwFs into CwFs with additional structure.
        In particular we will consider $(\repinfty)$-CwFs, which should correspond to some notion of representable map $\infty$-categories.
        It seems possible to observe both homotopical and computational properties of the theories using $(\repinfty)$-CwFs, while computational properties are not always easily observable with $\infty$-categories (depending on the chosen model of $\infty$-categories).

  \item We prefer to work with the $1$-category of $(\reppre)$-CwFs and strict $(\reppre)$-CwF morphisms, instead of the $(2,1)$-category of $(\reppre)$-CwFs and pseudo-morphisms.
        Similarly, we prefer to work with its $1$-category of models and strict morphisms, rather than the $(2,1)$-category of models and weak morphisms.
        One of the reason is that we consider factorization systems and semi model structures on these categories, which are easier to understand in the $1$-categorical setting.
        This does not play an important role in this paper, as we work almost exclusively with the coclassifying $(\reppre)$-CwF of the theory, without considering morphisms between other models.
\end{itemize}

\subsection{Definition and functorial semantics}\label{ssec:semantics_sogats}

\begin{defi}\label{def:sogat}
  A \defemph{second-order generalized algebraic theory} (SOGAT) is an $\{I^{\ty},I^{\repty},I^{\tm},E^{\tm}\}$-cellular $(\reppre)$-CwF $\Th$, where the maps $\{I^{\ty},I^{\repty},I^{\tm},E^{\tm}\}$ are the generic extensions of $(\reppre)$-CwFs by a type, representable type, term or term equality:
  \begin{alignat*}{3}
    & I^{\ty} && :{ }
    && \Free_{\reppre}(\bm{\Gamma} \vdash) \to \Free_{\reppre}(\bm{\Gamma} \vdash \bm{A}\ \type),
    \\
    & I^{\repty} && :{ }
    && \Free_{\reppre}(\bm{\Gamma} \vdash) \to \Free_{\reppre}(\bm{\Gamma} \vdash \bm{A}\ \reptype),
    \\
    & I^{\tm} && :{ }
    && \Free_{\reppre}(\bm{\Gamma} \vdash \bm{A}\ \type) \to \Free_{\reppre}(\bm{\Gamma} \vdash \bm{a} : \bm{A}),
    \\
    & E^{\tm} && :{ }
    && \Free_{\reppre}(\bm{\Gamma} \vdash \bm{x},\bm{y} : \bm{A}) \to \Free_{\reppre}(\bm{\Gamma} \vdash \bm{x} = \bm{y}).
       \tag*{\defiEnd{}}
  \end{alignat*}
\end{defi}
In other words, a SOGAT is a presentation of a $(\reppre)$-CwF by collections of generating types, generating representable types, generating terms and generating equations between terms.
We will write these generators using a $\inner{red,bold}$ font.
In practice, a SOGAT is given by a signature, and the $(\reppre)$-CwF $\Th$ is reconstructed from the signature.
We keep the notion of signature informal in this paper; a formal definition of signature can be given by modifying the definition of QIIT-signature of \textcite{ConstructingQIITs}.
For every generating type, term or equation in a signature, the $(\reppre)$-CwF is extended by pushout against a map in $\{I^{\ty},I^{\repty},I^{\tm},E^{\tm}\}$.

For example, the signature of a pointed dependent type
\begingroup{}\allowdisplaybreaks{}
\begin{alignat*}{3}
  & A && :{ } && \Ty, \\
  & B && :{ } && (a : \Tm(A)) \to \Ty, \\
  & b && :{ } && (a : \Tm(A)) \to \Tm(B(a)).
\end{alignat*}\endgroup{}
gets translated to the following iterated pushout:
\[
\begin{tikzcd}[column sep = 60pt]
  \Free_{\reppre}(\bm{\Gamma} \vdash)
  \ar[d]
  \ar[r, "{\angles{\diamond}}"']
  \ar[rd, phantom, very near end, "\ulcorner"]
  & \Init_{\reppre}
  \ar[d]
  &
  \\ \Free_{\reppre}(\bm{\Gamma} \vdash \bm{A}\ \type)
  \ar[r, "{\angles{\diamond, A}}"']
  & \Init_{\reppre}[A]
  \ar[d]
  & \Free_{\reppre}(\bm{\Gamma} \vdash)
  \ar[d]
  \ar[l, "{\angles{(a:A)}}"]
  \ar[ld, phantom, very near end, "\urcorner"]
  \\ \Free_{\reppre}(\bm{\Gamma} \vdash \bm{A}\ \type)
  \ar[d]
  \ar[r, "{\angles{(a:A), B(a)}}"']
  & \Init_{\reppre}[A,B]
  \ar[d]
  & \Free_{\reppre}(\bm{\Gamma} \vdash \bm{A}\ \type)
  \ar[l, "{\angles{(a:A), B(a)}}"]
  \\ \Free_{\reppre}(\bm{\Gamma} \vdash \bm{a} : \bm{A}\ \type)
  \ar[r, "{\angles{(a:A), B(a), b(a)}}"']
  & \Init_{\reppre}[A,B,b] \rlap{\ .}
  &
\end{tikzcd}
\]

Our running examples will be the first-order generalized algebraic theory $\Th_{\CCat}$ of categories and the type theory $\Th_{\Id}$ of weak identity types.
\begin{exa}
  The (first-order) generalized algebraic theory of categories $\Th_{\CCat}$ is given by the following signature:
  \begingroup{}\allowdisplaybreaks{}
  \begin{alignat*}{3}
    & \iob && :{ } && \Ty \\
    & \iOb && \triangleq{ } && \Tm(\iob) \\
    & \ihom && :{ } && \iOb \to \iOb \to \Ty \\
    & \iHom(x,y) && \triangleq{ } && \Tm(\ihom(x,y)) \\
    & \ieqhom && :{ } && \forall x\ y \to \iHom(x,y) \to \iHom(x,y) \to \Ty \\
    & \iEqHom(f,g) && \triangleq{ } && \Tm(\ieqhom(f,g)) \\
    & \iid && :{ } && \forall x \to \iHom(x,x) \\
    & \icomp && :{ } && \forall x\ y\ z \to \iHom(x,y) \to \iHom(y,z) \to \iHom(x,z) \\
    & f \icirc g && \triangleq{ } && \icomp(g,f) \\
    & \irefl && :{ } && \forall x\ y\ \to (f : \iHom(x,y)) \to \iEqHom(f,f) \\
    &&&&& \iid \icirc f = f \\
    &&&&& f \icirc \iid = f \\
    &&&&& (f \icirc g) \icirc h = f \icirc (g \icirc h) \\
    &&&&& \forall x\ y\ f\ g \to (p,q : \iEqHom(f,g)) \to p = q \\
    &&&&& \iEqHom(f,g) \to f = g
  \end{alignat*}\endgroup{}
  We use the capitalized $\iOb$, $\iHom$, $\iEqHom$ to denote the elements of the sorts $\iob$, $\ihom$ and $\ieqhom$.
  We also use $\_{} \icirc \_{}$ as an infix notation for composition.

  Note that including the sort $\ieqhom$ of equalities between morphisms does not change the categories of models of $\Th_{\CCat}$.
  However it has to be included in order to determine the correct ``language of categories''.
  In our setting, including this sort is needed to equip $\Th_{\CCat}$ with homotopy relations in~\cref{ssec:homotopy_relations}; isomorphisms cannot be defined without mentioning equality of morphisms.
  \defiEnd{}
\end{exa}

\begin{exa}
  The SOGAT $\Th_{\CCwf}$ of a family with representable elements is given by the following signature:
  \begin{alignat*}{3}
    & \ity && :{ } && \Ty  \\
    & \iTy && \triangleq{ } && \Tm(\ity) \\
    & \itm && :{ } && \iTy \to \RepTy \\
    & \iTm(A) && \triangleq{ } && \Tm(\itm(A))
  \end{alignat*}

  Type-theoretic structures ($\Sigma$, $\Pi$, \etc{}) can be specified by extensions of this signature by new operations and equations.
  In particular, the theory $\Th_{\Id}$ of weak identity types is the extension of $\Th_{\CCwf}$ by the new operations $(\iId,\irefl,\iJ,\iJb)$ with the signature given in~\cref{def:weak_identity_types}.
  \defiEnd{}
\end{exa}

For the remainder of this section, we fix an arbitrary SOGAT $\Th$.

We now briefly recall the main definitions of the functorial semantics of $\Th$; we refer the reader to~\cite{UemuraFramework} for further details.
The main results of this paper only involve the syntax of $\Th$; but are motivated by the semantics.

\begin{defi}
  An \defemph{internal model} of $\Th$ in a $(\reppre)$-CwF $\CC$ is a $(\reppre)$-CwF morphism
  \[ \MC : \Th \to \CC.
    \tag*{\defiEnd{}}
  \]
\end{defi}
When unambiguous, we will write $X_{\CC}$, $A_{\CC}$, $a_{\CC}$, \etc{} instead of $\MC(X)$, $\MC(A)$, $\MC(a)$, \etc{} for the application of the $(\reppre)$-CwF morphism $\MC$ on objects, morphisms, types and terms.

By the universal property of $\Th$, an internal model in $\CC$ is uniquely determined by the image of the generators of $\Th$, that is by an interpretation of the signature $\Th$ in $\CC$.

We have a locally finitely presentable $1$-category $(\CCwf_{\reppre} \backslash \Th)$ of $(\reppre)$-CwFs equipped with an internal model of $\Th$.
The identity morphism $\id : \Th \to \Th$ equips $\Th$ with the structure of an internal model, called the \emph{generic model} of $\Th$.
It is also the initial object of $(\CCwf_{\reppre} \backslash \Th)$).

\begin{defi}
  A \defemph{model} of $\Th$ consists of a category $\CC$ with a terminal object, along with an internal model of $\Th$ in the $(\reppre)$-CwF $\widehat{\CC}$, that is a $(\reppre)$-CwF morphism $\MC : \Th \to \widehat{\CC}$.
  \defiEnd{}
\end{defi}

\begin{defi}
  A \defemph{weak morphism} $F$ of models of $\Th$ consists of a functor $F : \CC \to \CD$ such that:
  \begin{itemize}
    \item The functor $F$ weakly preserves terminal objects.
    \item For every object $X : \Th$, we have a transformation
          \[ F^{X} : (\Gamma : \CC^{\op}) \to (\yo(\Gamma) \to X_{\CC}) \to (\yo(F(\Gamma)) \to X_{\CD}), \]
          contravariantly natural in $\Gamma$.
    \item For every morphism $\alpha : X \to Y$, the following square commutes
          \[ \begin{tikzcd}
              (\yo(\Gamma) \to X_{\CC})
              \ar[d, "(\alpha_{\CC} \circ -)"]
              \ar[r, "F^{X}"] &
              (\yo(F(\Gamma)) \to X_{\CD})
              \ar[d, "(\alpha_{\CD} \circ -)"] \\
              (\yo(\Gamma) \to Y_{\CC})
              \ar[r, "F^{Y}"] &
              (\yo(F(\Gamma)) \to Y_{\CD})
          \end{tikzcd} \]
    \item Remark that we obtain, for every object $X : \Th$ and type $A \Colon \yo(X) \to \Ty_{\Th}$, a natural transformation
          \begin{alignat*}{3}
            & F^{A} && :{ } && (\Gamma : \CC^{\op}) \to (x : \yo(\Gamma) \to X_{\CC}) \to (a : (\gamma : \yo(\Gamma)) \to A_{\CC}(x(\gamma))) \\
            &&&&& \to ((\gamma : \yo(F(\Gamma))) \to A_{\CD}(F^{X}(x)(\gamma)))
          \end{alignat*}
          such that $F^{X.A}(x,a) = (F^{X}(x),F^{A}(x,a))$.
    \item Context extensions are weakly preserved: for every object $X : \Th$, representable type $A \Colon \yo(X) \to \RepTy_{\Th}$, object $\Gamma : \CC$ and element $x \Colon \yo(\Gamma) \to X_{\CC}$, the comparison map
          \[ \angles{F(\bm{p}),F^{A}(\bm{q})} : F((\gamma:\Gamma).A_{\CC}(x(\gamma))) \to (\gamma:F(\Gamma)).A_{\CD}(F^{X}(x)(\gamma)) \]
          is an isomorphism.
  \end{itemize}

  A morphism is \defemph{strict} if the terminal object and context extensions are strictly preserved.
  \defiEnd{}
\end{defi}

\begin{defi}
  A $2$-cell between two weak morphisms $F,G : \CC \to \CD$ of models of $\Th$ consists of a natural isomorphism $\alpha : F \cong G$, such that:
  \begin{itemize}
    \item For every object $X : \Th$, context $\Gamma : \CC$, elements $x \Colon \yo(\Gamma) \to X_{\CC}$ and $\gamma : \yo(F(\Gamma))$, we have $F^{X}(x,\gamma) = G^{X}(x,\alpha_{\Gamma}(\gamma))$.
          \defiEnd{}
  \end{itemize}
\end{defi}

We have a $(2,1)$-category of models, weak morphisms and $2$-cells, and a $1$-category $\CMod_{\Th}$ of models and strict morphisms.
We will mainly work with the $1$-category $\CMod_{\Th}$.
The category $\CMod_{\Th}$ is locally finitely presentable; in particular we have an initial model $\Init_{\Th}$ and more general freely generated models.

\subsection{Structure of the types of a SOGAT}

We write $\GenTy_{\Th}$ for the set of generating types of $\Th$; it can be obtained from the presentation of $\Th$ as an $\{I^{\ty},I^{\repty},I^{\tm},E^{\tm}\}$-cellular $(\reppre)$-CwF.
For every $\iS : \GenTy_{\Th}$, we have an object $\partial \iS : \Th$ and a type $\iS : \yo(\partial \iS) \to \Ty_{\Th}$.

We also have a subset $\GenRepTy_{\Th} \subseteq \GenTy_{\Th}$ of generating representable types of $\Th$.

For example, $\GenTy_{\Th_{\CCat}} = \{\iob,\ihom\}$ with $\partial \iob = \Unit$ and $\partial \ihom = \iob \times \iob$.
For $\Th_\Id$, we have $\GenTy_{\Th_{\Id}} = \{\ity,\itm\}$ and $\GenRepTy_{\Th_{\Id}} = \{\itm\}$, with $\partial \ity = \Unit$ and $\partial \itm = \ity$.

Because a SOGAT cannot contain any equations between sorts, the types of $\Th$ can all be reconstructed by closing the generating types under substitution and the type-formers $\Sigma$, $\Unit$ and $\Pi$.
We can consider the same closure in arbitrary internal models of $\Th$.
Furthermore we stratify these types into basic types (obtained by closing the generating types under substitution), the monomial types (obtained by closing the basic types under dependent products with arities in representable types) and the polynomial types (``sums of products'', obtained by closing the monomial types under dependent sums).

\begin{defi}\label{def:sogat_named_families}
  Let $\CC$ be a $(\reppre)$-CwF equipped with an internal model of $\Th$.
  We define families $\BTy_{\CC}$, $\MonoTy_{\CC}$ and $\PolyTy_{\CC}$ that are restrictions of $\Ty_{\CC}$, and a restricted family $\BRepTy_{\CC} \to \RepTy_{\CC}$.
  We work internally to $\CPsh(\CC)$.

  \begin{itemize}
    \item A \defemph{basic type} $\iS(\sigma) : \BTy_{\CC}$ consists of $\iS : \GenTy_{\Th}$ and $\sigma : \Tm_{\CC}(\partial \iS_{\CC})$.

          The corresponding type in $\Ty_{\CC}$ is $\iS_{\CC}(\sigma)$.
    \item A \defemph{basic representable type} $\iS(\sigma) : \BRepTy_{\CC}$ consists of $\iS : \GenRepTy_{\Th}$ and $\sigma : \Tm_{\CC}(\partial \iS)$.
    \item A \defemph{monomial type} $[\Delta \vdash A] : \MonoTy_{\CC}$ consists of a telescope $\Delta : \BRepTy^{\star}_{\CC}$ of basic representable types, along with a dependent basic type $A : \Tm^{\star}_{\CC}(\Delta) \to \BTy_{\CC}$.
          The corresponding type in $\Ty_{\CC}$ is an iterated first-order $\Pi$-type.
    \item A \defemph{polynomial type} is a telescope of monomial types: $\PolyTy_{\CC} \triangleq \MonoTy_{\CC}^{\star}$.
          The corresponding type in $\Ty_{\CC}$ is obtained as an iterated $\Sigma$-type.
          \defiEnd{}
  \end{itemize}
\end{defi}

We also define the closure $\Clos_{\reppre}(\BTy_{\CC})$ of basic types under $\Unit$-, $\Sigma$- and first-order $\Pi$- types.
\begin{defi}
  We define restricted families $\Clos_{\reppre}(\BTy_{\CC}) \to \Ty_{\CC}$ and $\Clos_{\lexpre}(\BRepTy_{\CC}) \to \RepTy_{\CC}$ by induction-recursion (internally to $\CPsh(\CC)$).
  \begingroup{}\allowdisplaybreaks{}
  \begin{alignat*}{3}
    & \Clos_{\lexpre}(\BRepTy_{\CC}) && :{ } && \UU, \\
    & \iota_{\rep} && :{ } && \Clos_{\lexpre}(\BRepTy_{\CC}) \to \RepTy_{\CC}, \\
    & \Clos_{\reppre}(\BTy_{\CC}) && :{ } && \UPsh, \\
    & \iota && :{ } && \Clos_{\reppre}(\BTy_{\CC}) \to \Ty_{\CC}.
  \end{alignat*}\endgroup{}
  The family $\Clos_{\lexpre}(\BRepTy_{\CC})$ has constructors $\tau_{\rep} : \BRepTy_{\CC} \to \Clos_{\lexpre}(\BRepTy_{\CC})$, $\Unit$ and $\Sigma$, with $\iota_{\rep}(\tau_{\rep}(A)) = A$, $\iota_{\rep}(\Unit) = \Unit$ and $\iota_{\rep}(\Sigma(A,B)) = \Sigma(\iota_{\rep}(A), \lambda a \mapsto \iota_{\rep}(B(a)))$.
  Similarly, the family $\Clos_{\reppre}(\BTy_{\CC})$ has constructors $\tau : \BTy_{\CC} \to \Clos_{\reppre}(\BTy_{\CC})$, $- : \RepTy_\CC \to \Clos_{\reppre}(\BTy_{\CC})$, $\Unit$, $\Sigma$ and $\Pi_{\rep}$ that are preserved by $\iota$.
  \defiEnd{}
\end{defi}

\begin{prop}\label{prop:type_iso_polytype}
  The canonical maps $\PolyTy_{\CC} \to \Clos_{\reppre}(\BTy_{\CC})$ and $\BRepTy^{\star}_{\CC} \to \Clos_{\lexpre}(\BRepTy_{\CC})$ are essentially surjective: for every $A : \Clos_{\reppre}(\BTy_{\CC})$, there is some $A_{0} : \PolyTy_{\CC}$ such that $\Tm_{\CC}(A_{0}) \simeq \Tm_{\CC}(A)$; and for every $A : \Clos_{\lexpre}(\BRepTy_{\CC})$, there is some $A_{0} : \BRepTy^{\star}_{\CC}$ such that $\Tm_{\CC}(A_{0}) \simeq \Tm_{\CC}(A)$.
\end{prop}
\begin{proof}
  This follows from the facts that $\Sigma$-types are essentially associative and that (first-order) $\Pi$-types essentially distribute over $\Sigma$-types.
\end{proof}

Since the presentation of $\Th$ does not include any type equation, the types of $\Th$ are exactly the closure of the basic types under $\Unit$, $\Sigma$ and first-order $\Pi$ -types.
\begin{prop}\label{prop:description_ty_coclassifying_model}
  The canonical maps
  \[ \Clos_{\reppre}(\BTy_{\Th}) \to \Ty_{\Th} \]
  and
  \[ \Clos_{\lexpre}(\BRepTy_{\Th}) \to \RepTy_{\Th} \]
  are isomorphisms.
  \qed{}
\end{prop}
We omit the proof; it follows from a standard normalization argument.
This result allows us to use induction over the structure of types of $\Th$.

\begin{cor}
  The canonical maps $\PolyTy_{\Th} \to \Ty_{\Th}$ and $\BRepTy^{\star}_{\Th} \to \RepTy_{\Th}$ are essentially surjective: for every $A : \Ty_{\Th}$, there is some $A_{0} : \PolyTy_{\Th}$ such that $A_{0} \simeq A$; and for every $A : \RepTy_{\Th}$, there is some $A_{0} : \BRepTy^{\star}_{\Th}$ such that $A_{0} \simeq A$.
\end{cor}
\begin{proof}
  By~\cref{prop:type_iso_polytype} and~\cref{prop:description_ty_coclassifying_model}.
\end{proof}

\begin{prop}
  The family restriction $\RepTy_\Th \to \Ty_\Th$ is a monomorphism.
\end{prop}
\begin{proof}
  This follows from the isomorphism $\Ty_\Th \cong \Clos_{\reppre}(\BTy_\Th)$.
  Indeed $\RepTy_\Th \to \Clos_{\reppre}(\BTy_\Th)$ is a constructor of $\Clos_{\reppre}(\BTy_\Th)$, and is therefore injective.
\end{proof}

\subsection{Contextual models}\label{ssec:contextual_models_sogats}

We can generalize the notions of contextuality from CwFs to the category of models of an arbitrary SOGAT.

\begin{defi}\label{def:contextual_isomorphism}
  A morphism $F : \CC \to \CD$ in $\CMod_{\Th}$ is a \defemph{contextual isomorphism} if it is bijective on every sort: for every generating type $\iS : \GenTy_{\Th}$, object $\Gamma : \CC$, boundary $\sigma \Colon \yo(\Gamma) \to \partial\iS_{\CC}$ and element $a \Colon (\gamma : \yo(F(\Gamma))) \to \iS_{\CD}(F(\sigma)(\gamma))$, there is a unique element $a_{0} \Colon (\gamma : \yo(\Gamma)) \to \iS_{\CC}(\sigma(\gamma))$ such that $F(a_{0}) = a$.
  \defiEnd{}
\end{defi}

The contextual isomorphisms are the right class of maps of an orthogonal factorization system generated by a set $I_{\Th}$ of maps in $\CMod_{\Th}$.
\begin{alignat*}{3}
  & I_{\Th} && \triangleq{ } && \{ I^{\iS} \mid \iS : \GenTy_{\Th} \} \\
  & I^{\iS} && :{ } && \Free_{\Th}(\bm{\Gamma} \vdash \bm{\sigma} : \partial \iS) \to \Free_{\Th}(\bm{\Gamma} \vdash \bm{x} : \iS(\bm{\sigma}))
\end{alignat*}
The maps in the corresponding left class are called \defemph{left contextual} maps.

\begin{defi}\label{def:contextual_core}
  The \defemph{contextual core} $\cxl(\CC)$ is obtained from the factorization of the unique map $\Init_{\Th} \to \CC$ as a left contextual map $\Init_{\Th} \to \cxl(\CC)$ followed by a contextual isomorphism $\cxl(\CC) \to \CC$.
  \defiEnd{}
\end{defi}

\begin{defi}\label{def:contextual}
  A model $\CC : \CMod_{\Th}$ is \defemph{contextual} if $\cxl(\CC) \to \CC$ is an isomorphism.
  \defiEnd{}
\end{defi}

The $1$-category $\CMod_{\Th}^{\cxl}$ of contextual models forms a coreflective subcategory of $\CMod_{\Th}$; the functor $\cxl : \CMod_{\Th} \to \CMod_{\Th}^{\cxl}$ is right adjoint to the subcategory inclusion $\CMod_{\Th}^{\cxl} \to \CMod_{\Th}$.

\subsection{Trivial fibrations}\label{ssec:trivial_fibrations_sogats}

\begin{defi}\label{def:trivial_fibration}
  A morphism $F : \CC \to \CD$ in $\CMod_{\Th}$ is a \defemph{trivial fibration} if it is surjective on every sort: for every generating type $\iS : \GenTy_{\Th}$, object $\Gamma : \CC$, boundary $\sigma : \yo(\Gamma) \to \partial\iS_{\CC}$ and element $a : (\gamma : \yo(F(\Gamma))) \to \iS_{\CD}(F(\sigma)(\gamma))$, there exists an element $a_{0} : (\gamma : \yo(\Gamma)) \to \iS_{\CC}(\sigma(\gamma))$ such that $F(a_{0}) = a$.
  \defiEnd{}
\end{defi}

The trivial fibrations are the right class of maps of the weak factorization system that is cofibrantly generated by the same set $I_{\Th}$ of maps that we used to define contextual isomorphisms.
The maps in the left class are called \defemph{cofibrations}.

In the case of the GAT $\Th_\CCat$, the trivial fibrations are the trivial fibrations of the canonical model structure on $\CCat$, that is functors that are surjective on objects and fully faithful.

For the SOGAT $\Th_\CCwf$, the (cofibrations, trivial fibrations) weak factorization system on $\CCwf$ coincides with the one defined by \textcite{HomotopyTheoryTTs}.


%% file: homotopy.tex
\section{Theories with homotopy relations}\label{sec:homotopy}

\subsection{Homotopy relations}\label{ssec:homotopy_relations}

We now consider SOGATs that are equipped with an additional piece of data: a choice of a homotopy relation for every generating sort of the theory.

From the point of view of model categories, this roughly corresponds to the choice of a relative cylinder object for every generating cofibration.

\begin{defi}\label{def:homotopy_relations}
  The data of \defemph{homotopy relations} on a SOGAT $\Th$ consists, for every generating type $\iS : \GenTy_{\Th}$, of a reflexive type-valued binary relation on its terms:
  \begin{alignat*}{3}
    & \_{} \sim_{\iS(\_{})} \_{} && \Colon{ } && (\sigma : \partial \iS) (x,y : \iS(\sigma)) \to \Ty_{\Th}, \\
    & \refl_{\iS(\_{})} && \Colon{ } && (\sigma : \partial \iS) (x : \iS(\sigma)) \to x \sim_{\iS(\sigma)} x.
    \tag*{\defiEnd{}}
  \end{alignat*}
\end{defi}
Since these homotopy relations are specified in the $(\reppre)$-CwF $\Th$, they are automatically available in any other model of $\Th$.

\begin{exa}\label{exa:homotopy_relations_th_cat}
  Homotopy relations are defined over the theory of categories $\Th_{\CCat}$ as follows:
  \begin{alignat*}{3}
    & x \sim_{\iob} y && \triangleq{ } && (x \cong y), \\
    & f \sim_{\ihom(x,y)} g && \triangleq{ } && \ieqhom(f,g), \\
    & \_{} \sim_{\ieqhom(f,g)} \_{} && \triangleq{ } && \Unit,
  \end{alignat*}
  where $(x \cong y)$ is the type of isomorphisms between $x$ and $y$, \ie{}
  \begin{alignat*}{3}
    & (x \cong y) && \triangleq{ } && (f : \ihom(x,y)) \times (g : \ihom(y,x)) \\
    &&&&& \quad\times \ieqhom(g \icirc f,\iid) \times \ieqhom(f \icirc g,\iid).
  \end{alignat*}
  Reflexivities are given by the identity isomorphisms on objects, by $\irefl$ on morphisms, and by $\tt$ on equalities between morphisms.
  \defiEnd{}
\end{exa}

\begin{exa}\label{exa:homotopy_relations_th_id}
  Homotopy relations are defined over the type theory $\Th_{\Id}$ of weak identity types as follows:
  \begin{alignat*}{3}
    & A \sim_{\ity} B && \triangleq{ } && \iEquiv(A,B) \\
    & x \sim_{\itm(A)} y && \triangleq{ } && \itm(\iId(A,x,y))
  \end{alignat*}
  where $\iEquiv(A,B)$ is the type of relational equivalences between $A$ and $B$.
  Note that even though $\iEquiv(A,B)$ is not classified by an inner type in $\Th_{\Id}$, it can be written as an outer type in $\Th_{\Id}$.
  \begin{alignat*}{3}
    & \Equiv(A,B) && \triangleq{ } && (R : \itm(A) \to \itm(B) \to \ity) \\
    &&&&& \times ((a : \itm(A)) \to \isContr((b : B) \times R(a,b))) \\
    &&&&& \times ((b : \itm(B)) \to \isContr((a : A) \times R(a,b))) \\
    & \isContr(X) && \triangleq{ } && (x : \itm(X)) \times (\forall (x,y : \itm(X)) \to \itm(\iId(X,x,y)))
  \end{alignat*}
  Reflexivities are given by the identity equivalence $\iId_{-}$ on types, and by $\irefl$ on terms.
  \defiEnd{}
\end{exa}

We fix a SOGAT $\Th$ equipped with homotopy relations for the remainder of this section.

\subsection{Classes of maps}\label{ssec:homotopy_classes_of_maps}

The homotopy relations induce notions of weak equivalences and of fibrations over the category $\CMod_{\Th}$.
In the case of the theory $\Th_{\Id}$ of weak identity types, we recover the classes of weak equivalences and fibrations on $\CCwf_{\Id}$ that were introduced by \textcite{HomotopyTheoryTTs}.

\begin{defi}\label{def:weak_equivalence}
  A morphism $F : \CC \to \CD$ in $\CMod_{\Th}$ is a \defemph{weak equivalence} if it is essentially surjective on every sort: for every generating type $\iS : \GenTy_{\Th}$, object $\Gamma : \CC$, boundary $\sigma \Colon \yo(\Gamma) \to \partial\iS_{\CC}$, and element $x \Colon (\gamma : \yo(F(\Gamma))) \to \iS_{\CD}(F(\sigma)(\gamma))$, there exists a lifted element $x_{0} \Colon (\gamma : \yo(\Gamma)) \to \iS_{\CC}(\sigma(\gamma))$ along with a homotopy
  \[ p \Colon (\gamma : \yo(F(\Gamma))) \to F(x_{0})(\gamma) \sim_{\iS_{\CD}(F(\sigma)(\gamma))} x(\gamma).
    \tag*{\defiEnd{}}
  \]
\end{defi}

\begin{defi}\label{def:fibration}
  A morphism $F : \CC \to \CD$ in $\CMod_{\Th}$ is a \defemph{fibration} if it satisfies a lifting condition for homotopies with a fixed left endpoint.
  \begin{description}
    \item[homotopy lifting] For every generating type $\iS : \GenTy_{\Th}$, object $\Gamma : \CC$, boundary $\sigma \Colon \yo(\Gamma) \to \partial\iS_{\CC}$, element $x \Colon (\gamma : \yo(\Gamma)) \to \iS_{\CC}(\sigma(\gamma))$ and homotopy
          \[ p \Colon (\gamma: \yo(F(\Gamma))) \to F(x)(\gamma) \sim_{\iS_{\CD}(F(\sigma)(\gamma))} y(\gamma), \]
          there exists a homotopy
          \[ p_{0} \Colon (\gamma : \yo(\Gamma)) \to x(\gamma) \sim_{\iS_{\CC}(\sigma(\gamma))} y_{0}(\gamma) \]
          such that $F(y_{0}) = y$ and $F(p_{0}) = p$.
          \defiEnd{}
  \end{description}
\end{defi}

\subsection{Univalent internal models}

Recall that a $(\repinfty)$-CwF is a $(\reppre)$-CwF equipped with weakly stable identity types satisfying function extensionality.
Consider a $(\repinfty)$-CwF $\CC$ equipped with an internal model of $\Th$.
Internally to $\CC$, we have two notions of ``weak equality'' between elements of the model of $\Th$, given by the homotopy relations $(\sim)$ and by the (outer) identity types $(\simeq)$.
There is always a comparison map that sends elements of the outer identity types $(\simeq)$ to homotopies $(\sim)$, defined by sending the outer reflexivity to the inner reflexivity.
It is then natural to ask for this map to be an equivalence (with respect to the outer identity types).
We express this as a contractibility condition.

\begin{defi}\label{def:univalent_internal_model}
  Let $\CC$ be a $(\repinfty)$-CwF equipped with an internal model of $\Th$.
  We say that the internal model is \defemph{univalent}, or that the identity types are \defemph{saturated} (with respect to the homotopy relations) if for every generating type $\iS : \GenTy_{\Th}$, the dependent type
  \[ (y : \iS_{\CC}(\sigma)) \times (p : x \sim_{\iS_{\CC}(\sigma)} y) \]
  is contractible over $(\sigma : \partial \iS_{\CC}, x : \iS_{\CC}(\sigma))$, for the notion of contractibility induced by the outer identity types $(\simeq)$.
  \defiEnd{}
\end{defi}
In the case of the theory $\Th_{\CCat}$, an internal category is univalent in the sense of~\cref{def:univalent_internal_model} when it is univalent in the sense of HoTT~\parencite{UnivalentCategories}.

\subsection{External univalence}

We can finally define the main notion of this paper.

\begin{defi}
  We say that a SOGAT $\Th$ equipped with homotopy relations satisfies \defemph{external univalence} when the $(\reppre)$-CwF $\Th$ can be equipped with weakly stable identity types satisfying function extensionality and saturation with respect to the homotopy relations.
  \defiEnd{}
\end{defi}

The following claim will be proven in a future paper.
\begin{clm}\label{clm:external_univalence_equiv_modelcat}
  Let $\Th$ be a SOGAT equipped with homotopy relations.
  It satisfies external univalence if and only if the category $\CMod_\Th^\cxl$ of contextual models of $\Th$, equipped with the classes of trivial fibrations, fibrations and weak equivalences defined in~\cref{def:trivial_fibration}, \cref{def:fibration} and~\cref{def:weak_equivalence}, is a left semi-model category.
\end{clm}



%% file: reflexive_equivalences.tex
\section{Contractibility data and reflexive equivalences}\label{sec:refl_eqv}
We show that weakly stable identity types can be reconstructed from the data of reflexive relational equivalences (also called one-to-one relations or one-to-one correspondences).
Similar ideas are used in the cubical type theory without an interval of \textcite{CubicalTypeTheoryWithoutInterval} and in the higher observational type theory of Altenkirch, Kaposi and Shulman~\parencite{ShulmanHigherObservationalTT, TowardsHigherObservationalTT}.

We fix a CwF $\CC$ equipped with $\Sigma$-types.
\begin{defi}[Internally to $\CPsh(\CC)$]\label{def:contractibility_data}
  \defemph{Contractibility data} over $\CC$ consists of a dependent presheaf
  \[ \isContr : \Ty_{\CC} \to \UPsh.
    \tag*{\defiEnd{}}
  \]
\end{defi}
Note that contractibility is not propositional data, even though a witness of contractibility should be unique up to homotopy.
Whenever we say that some type $A$ is contractible, we really mean that we have an element of $\isContr(A)$.

We now assume that $\CC$ is equipped with global contractibility data.

\begin{defi}[Internally to $\CPsh(\CC)$]\label{def:relational_equivalence}
  An \defemph{equivalence} $E : A \simeq B$ between two types $A, B : \Ty_{\CC}$ consists of a binary relation
  \[ E : \Tm_{\CC}(A) \to \Tm_{\CC}(B) \to \Ty_{\CC} \]
  that is functional in both directions, as witnessed by the following contractibility conditions
  \begin{alignat*}{1}
    & E .\funr : \forall (a : \Tm_{\CC}(A)) \to \isContr((b : B) \times E(a,b)), \\
    & E .\funl : \forall (b : \Tm_{\CC}(B)) \to \isContr((a : A) \times E(a,b)).
    \tag*{\defiEnd{}}
  \end{alignat*}
\end{defi}

\begin{defi}[Internally to $\CPsh(\CC)$]\label{def:reflexive_equivalence}
  A \defemph{reflexive equivalence} is an equivalence $E : A \simeq A$ that is additionally equipped with a reflexivity map
  \[ E.\refl : (a : \Tm_{\CC}(A)) \to \Tm_{\CC}(E(a,a)).
    \tag*{\defiEnd{}}
  \]
\end{defi}

\begin{defi}[Internally to $\CPsh(\CC)$]\label{def:dependent_equivalence}
  A \defemph{dependent equivalence} for a dependent type $B : \Tm_{\CC}(A) \to \Ty_{\CC}$ over an equivalence $E : A \simeq A$ consists of a family of equivalences
  \[ T_{B} : \forall (x, y : \Tm_{\CC}(A)) \to \Tm_{\CC}(E(x,y)) \to B(x) \simeq B(y).
    \tag*{\defiEnd{}}
  \]
\end{defi}

\begin{defi}\label{def:contractibility_reflexive_equivalences}
  We say that $\CC$ is equipped with \defemph{reflexive equivalences} when for every type $A \Colon \yo(\Gamma) \to \Ty_{\CC}$, there is a reflexive equivalence
  \begin{alignat*}{3}
    & \Id_{A} && \Colon{ } && \forall \gamma \to A(\gamma) \simeq A(\gamma).
  \end{alignat*}
  We denote its reflexivity map by $\refl_{A} \Colon \forall \gamma\ (x : \Tm_{\CC}(A(\gamma))) \to \Id_{A}(\gamma,x,x)$.
  \defiEnd{}
\end{defi}

\begin{defi}\label{def:contractibility_dependent_equivalences}
  We say that $\CC$ is equipped with \defemph{dependent equivalences} when for every dependent type $B \Colon (\gamma : \yo(\Gamma)) \to \Tm_{\CC}(A(\gamma)) \to \Ty_{\CC}$, there is a dependent equivalence
  \begin{alignat*}{3}
    & \DId_{A.B} && \Colon{ } && \forall \gamma\ (x, y : \Tm_{\CC}(A(\gamma))) \to \Tm_{\CC}(\Id_{A}(\gamma,x,y)) \to B(\gamma,x) \simeq B(\gamma,y).
    \tag*{\defiEnd{}}
  \end{alignat*}
\end{defi}
Note that we do not assume that $\DId_{A.B}$ is reflexive; the reason is that $\DId_{A.B}$ can be replaced by reflexive dependent equivalences by considering the composition
\[ \DId_{A.B}(\gamma,x,y,p) \circ {\DId_{A.B}(\gamma,x,x,\refl_{A}(x))}^{-1} \]
when it is defined.

\begin{defi}\label{def:contractibility_center}
  We say that a type $A \Colon \forall \gamma \to \Ty_\CC$ has a \defemph{center} (of contraction) if we have an element
  \[ \ccenter_A \Colon{ } (\gamma : \yo(\Gamma)) \to \Tm_{\CC}(A(\gamma)).
  \]
\end{defi}

\begin{defi}\label{def:contractibility_has_all_paths}
  We say that a type $A \Colon \forall \gamma \to \Ty_\CC$ \defemph{has all paths}, or a \defemph{homogeneous all-paths operation}, if we have an element
  \begin{alignat*}{3}
    & \chpath_A && \Colon{ } && \forall (\gamma : \yo(\Gamma))\ (x, y : \Tm_{\CC}(A(\gamma))) \to \Tm_{\CC}(\Id_{A}(\gamma,x,y)).
  \end{alignat*}
  In other words, the type $A$ has all paths if it is a homotopy proposition, with respect to the identity type $\Id_{A}$.
  \defiEnd{}
\end{defi}
We will also need to consider an analogous heterogeneous structure for dependent contractible types, similarly to homogeneous and heterogeneous compositions structures in cubical type theories (see \eg{}~\textcite{SyntaxModelsCartCTT}).

\begin{defi}\label{defi:heterogeneous_all_paths}
  We say that a dependent type
  \[ B \Colon \forall (\gamma : \yo(\Gamma))\ (a : \Tm_\CC(A(\gamma))) \to \Ty_\CC \]
  has a \defemph{heterogeneous all-paths operation} if we have
  \begin{alignat*}{3}
    & \cpath_B && :{ }
    && \forall \gamma\ (a_{l},a_{r} : \Tm_\CC(A(\gamma)))\ (a_{e} : \Tm_\CC(\Id_{A}(\gamma,a_{l},a_{r}))) \\
    &&&&& \phantom{\forall} (b_{l} : \Tm_\CC(B(\gamma,a_{l})))\ (b_{r} : \Tm_\CC(B(\gamma,a_r))) \to \DId_{A.B}(\gamma,a_{e},b_{l},b_{r})).
          \tag*{\defiEnd{}}
  \end{alignat*}
\end{defi}

We say that a contractibility witness
\[ c \Colon \forall (\gamma : \yo(\Gamma)) \to \isContr(A(\gamma)) \]
has a center or a homogeneous all-paths operations if the type $A$ has a center or a homogeneous all-paths operations.
In that case, it is written $c.\ccenter$ or $c.\chpath$.

Similarly, we say that a dependent contractibility witness
\[ c \Colon \forall (\gamma : \yo(\Gamma))\ (a : \Tm_\CC(A(\gamma))) \to \isContr(B(\gamma,a)) \]
has a heterogeneous all-paths operations if the dependent type $B$ has one.
In that case, it is written $c.\cpath$.

When $c \Colon \forall \gamma \to \isContr((a : A(\gamma)) \times B(\gamma,a))$, the type $B(\gamma,a)$ is typically of the form $\Id_{A}(\gamma,a,a_{0})$ or $\Id_{A}(\gamma,a_{0},a)$ for some $a_{0} : A(\gamma)$.
In that case, we can think of the center and all-paths operations as specific cubical composition and filling operations, as described in the following diagrams (where we write $-.1$ and $-.2$ for the first and second projections out of a $\Sigma$-type):

\[ \begin{tikzcd}[column sep=70]
    a_{0}
    \ar[r, dashed, no head, "c.\ccenter.2"]
    & c.\ccenter.1
\end{tikzcd} \]

\[ \begin{tikzcd}[column sep=70, row sep=70]
    {a_{0}}
    \ar[r, equal]
    \ar[d, no head, "{b_{l}}"']
    \ar[rd, phantom, "(c.\chpath.2)"]
    & {a_{0}}
    \ar[d, no head, "{b_{r}}"]
    \\
    {a_{l}}
    \ar[r, dashed, no head, "{c.\chpath.1}"']
    &
    {a_{r}}
\end{tikzcd} \]

Indeed, the first and second projections of the operations $\ccenter$ and $\chpath$ correspond approximately to the operations $\mathsf{coe}$, $\mathsf{coh}$, $\mathsf{uncoe}$ and $\mathsf{uncoh}$ of the cubical type theory without an interval investigated by \textcite{CubicalTypeTheoryWithoutInterval}.

\begin{thm}\label{thm:identity_types_from_refleqv}
  Assume that $\CC$ is equipped with the following data:
   \begin{itemize}
    \item A family $\isContr$ of contractibility data (\cref{def:contractibility_data});
    \item Along with reflexive equivalences $(\Id_{-},\refl_{-})$ (\cref{def:contractibility_reflexive_equivalences});
    \item Together with dependent equivalences $(\DId_{-})$ (\cref{def:contractibility_dependent_equivalences});
    \item Such that every contractible type has a center (\cref{def:contractibility_center}) and all paths (\cref{def:contractibility_has_all_paths}).
  \end{itemize}

  Then the identity type introduction structure $(\Id_{-},\refl_{-})$ can be equipped with a weakly stable elimination structure.
\end{thm}
\begin{proof}
  Let $A \Colon \yo(\Gamma) \to \Ty_{\CC}$ be a type of $\CC$, along with a point $x \Colon \forall \gamma \to \Tm_{\CC}(A(\gamma))$.

  Take parameters $(\Delta,\gamma,P,d)$ for the weakly stable elimination structure, consisting of:
  \begin{alignat*}{3}
    & \Delta && :{ } && \CC, \\
    & \gamma && :{ } && \Delta \to \Gamma, \\
    & P && \Colon{ } && \forall (\delta : \yo(\Delta)) (y : \Tm_{\CC}(A(\gamma(\delta)))) (p : \Tm_{\CC}(\Id_{A}(\gamma(\delta),x(\gamma(\delta)),y))) \to \Ty_{\CC}, \\
    & d && \Colon{ } && \forall (\delta : \yo(\Delta)) \to \Tm_{\CC}(P(x(\gamma(\delta)), \refl_{A}(\gamma(\delta), x(\gamma(\delta))))).
  \end{alignat*}

  We have to construct
  \begin{alignat*}{3}
    & j && \Colon{ } && \forall (\delta : \yo(\Delta)) (y : \Tm_{\CC}(A(\gamma(\delta)))) (p : \Tm_{\CC}(\Id_{A}(\gamma(\delta),x(\gamma(\delta)),y))) \to \Tm_{\CC}(P(\delta,y,p)), \\
    & {j\beta} && \Colon{ } && \forall (\delta : \yo(\Delta)) \to \Tm_{\CC}(\Id_{\delta.P(\delta,x(\gamma(\delta)),\refl_{A}(\gamma(\delta),x(\gamma(\delta))))}(\delta,j(\delta,x(\gamma(\delta)),\refl_{A}(\gamma(\delta),x(\gamma(\delta)))),d(\delta))).
  \end{alignat*}

  We pose $S(\delta) \triangleq (y : A(\gamma(\delta))) \times \Id_{A}(\gamma(\delta),x(\gamma(\delta)),y)$; it is the type of the dependency of the motive $P$.
  Since $\Id_{A}$ is a reflexive equivalence, $S$ is a family of contractible types, \ie{} we have an element ${(\lambda \delta \mapsto \Id_{A}.\funr(\gamma(\delta)))}$ of $\forall \delta \to \isContr(S(\delta))$.

  We now see $P$ as a dependent type $P \Colon \forall \delta \to \Tm_{\CC}(S(\delta)) \to \Ty_{\CC}$.
  We consider the dependent equivalence
  \[ \DId_{S.P} \Colon \forall \delta\ (a,b : \Tm_{\CC}(S(\delta)))\ (p : \Tm_{\CC}(\Id_{S}(\delta,a,b))) \to P(\delta,a) \simeq P(\delta,b). \]

  Since $S$ is contractible and contractible types have all paths, we can specialize $\DId_{S.P}$ to
  \[ T_{S.P} \Colon \forall \delta\ (a,b : \Tm_{\CC}(S(\delta))) \to P(\delta,a) \simeq P(\delta,b), \]
  providing a way to transport between different fibers of $P$.

  We can now define $j$; we pose
  \begingroup{}\allowdisplaybreaks{}
  \begin{alignat*}{3}
    & a && \Colon{ } &&
    \forall \delta \to \Tm_{\CC}(S(\delta)), \\
    & a(\delta) && \triangleq{ } &&
    (x(\gamma(\delta)), \refl_{A}(\gamma(\delta),x(\gamma(\delta)))), \\
    & \overleftarrow{T_{S.P}} && :{ } &&
    \forall \delta \to \isContr((d' : P(\delta,a(\delta))) \times T_{S.P}(\delta,a(\delta),a(\delta),d',d(\delta))), \\
    & \overleftarrow{T_{S.P}}(\delta) && \triangleq{ } &&
    T_{S.P}(\delta,a(\delta),a(\delta)).\funl(d(\delta)), \\
    & d' && :{ } &&
    \forall \delta \to \Tm_{\CC}(P(\delta,a(\delta))), \\
    & d' && \triangleq{ } &&
    \overleftarrow{T_{S.P}}.\ccenter(\delta).1, \\
    & \overrightarrow{T_{S.P}} && :{ } &&
    \forall \delta\ (b : \Tm_{\CC}(S(\delta))) \to \isContr((j : P(\delta,b(\delta))) \times T_{S.P}(\delta,a(\delta),b,d'(\delta),j)), \\
    & \overrightarrow{T_{S.P}}(\delta) &&
    \triangleq{ } && T_{S.P}(\delta,a(\delta),b(\delta)).\funr(d'(\delta)), \\
    & j && :{ } && \forall \delta \to (b : \Tm_{\CC}(S(\delta))) \to \Tm_{\CC}(P(\delta,b(\delta))), \\
    & j && \triangleq{ } && \overrightarrow{T_{S.P}}.\ccenter(\delta,b).1.
  \end{alignat*}\endgroup{}

  Defining $j$ by transporting twice deals with the lack of reflexivity for the dependent equivalences.

  It remains to construct an element
  \begin{alignat*}{3}
    & {j\beta} && \Colon{ } &&
    \forall \delta \to \Id_{\delta.P(\delta,a(\delta))}(\delta,j(\delta,a(\delta)),d(\delta)).
  \end{alignat*}

  Note that we have elements $\widetilde{d'} : T_{S.P}(\delta,a(\delta),a(\delta),d'(\delta),d(\delta))$ and $\widetilde{j} : T_{S.P}(\delta,a(\delta),a(\delta),d'(\delta),j(\delta,a(\delta)))$.

  We consider the dependent type
  \begin{alignat*}{3}
    & Q(\delta,(y,q)) && :{ } && \forall \delta\ (y : \Tm_{\CC}(P(\delta,a(\delta))))\ (q : \Tm_{\CC}(T_{S.P}(\delta,a(\delta),a(\delta),d'(\delta),y))) \to \Ty_{\CC}, \\
    & Q(\delta,(y,q)) && \Colon{ } &&
    \Id_{\delta.P(\delta,a(\delta))}(\delta,y,d(\delta)).
  \end{alignat*}

  The type of such pairs $(y,q)$ is contractible, because $T_{S.P}$ is a dependent equivalence.
  Therefore, it has all paths and we obtain a family of equivalences
  \[ \forall \delta \to Q(\delta,(d(\delta),\widetilde{d'})) \simeq Q(\delta,(j(\delta,a(\delta)),\widetilde{j})). \]

  The element ${j\beta}(\delta)$ is obtained by transporting $\refl_{\delta.P(\delta,a(\delta))}(\delta,d(\delta))$ over that equivalence.
\end{proof}

We also show that, conversely, any weakly stable identity type structures satisfies the assumptions of~\cref{thm:identity_types_from_refleqv}.
For this we need to construct contractibility data that is stable under substitution, that is we need to strictify the standard definition (\cref{def:weakly_stable_contr}) of contractible types.
\begin{constr}\label{constr:stable_contractibility_data}
  Assume that $\CC$ is equipped with weakly stable identity types.
  Then we construct a family $\isContr' : \Ty_\CC \to \UPsh$ such that for every $A \Colon \yo(\Gamma) \to \Ty_\CC$, there is a logical equivalence
  \[ (\forall \gamma \to \isContr'(A(\gamma))) \leftrightarrow \isContr(A). \]
\end{constr}
\begin{proof}[Contruction]
  We construct $\isContr'$ by cofreely adding naturality to $\isContr$.

  We define $\isContr'$ as a dependent presheaf over $\Ty_\CC$.
  For any object $\Gamma : \CC$ and element $A \Colon \yo(\Gamma) \to \Ty_\CC$, we let $\isContr'_\Gamma(A)$ be the set of functions $c$ that send every morphism $\rho : \Delta \to \Gamma$ to a witness $c(\rho) : \isContr(A[\rho])$ of the contractibility of $A[\rho]$.
  For any morphism $f : \Omega \to \Gamma$, element $A \Colon \yo(\Gamma) \to \Ty_\CC$ and element $c : \isContr'_\Gamma(A)$, the restriction $c[f]$ sends a morphism $\rho : \Delta \to \Omega$ to a witness $c(f \circ \rho)$ of the contractibility of $A[f][\rho]$.

  The map $\isContr'_\Gamma(A) \to \isContr(A)$ is defined by evaluating $c$ at the identity morphism $\id : \Gamma \to \Gamma$.
  The map $\isContr(A) \to \isContr'_\Gamma(A)$ is defined using the fact that the weakly stable identity types are indeed weakly stable, providing maps $\isContr(A) \to \isContr(A[\rho])$ for any $\rho : \Delta \to \Gamma$.
\end{proof}

\begin{thm}
  If $\CC$ is equipped with weakly stable identity types, then the contractibility data constructed in~\cref{constr:stable_contractibility_data} satisfies the assumptions of~\cref{thm:identity_types_from_refleqv}.
\end{thm}
\begin{proof}
  All of the assumptions of~\cref{thm:identity_types_from_refleqv} are standard properties of identity types.
\end{proof}

\begin{rem}
  It should be possible to generalize this construction to CwFs without $\Sigma$-types by having the contractibility data quantify over telescopes of types, \ie{} $\isContr : \Ty_{\CC}^{\star} \to \UPsh$.
\end{rem}


%% file: inverse_diagram.tex
\section{Reflexive equivalences models}\label{sec:mgraph_model}
In this section, we construct $(\reppre)$-CwFs $\CPreReflGraph(\CC)$ of pre-reflexive graphs, $\CPreReflEqv(\CC)$ of equivalences with pre-reflexive equivalences and $\CReflEqv(\CC)$ of reflexive equivalences from a given $(\reppre)$-CwF $\CC$ equipped with suitable contractibility data.

By \emph{pre-reflexive} graph, we mean a graph together with a family of loops that should be thought of as reflexive loops, without conditions expressing the existence or the uniqueness of a reflexive loop yet.

For a SOGAT $\CT$, the model $\CPreReflEqv(\CT)$ will be used to prove external univalence for $\CT$ by constructing identity types over $\CT$ in~\cref{sec:univalence_proof}.

The $(\reppre)$-CwF $\CPreReflGraph(\CC)$ is an instance of an inverse diagram model~\parencite{HomotopicalInverseDiagrams}, indexed by the inverse category
\[ \PreReflGraph \triangleq \left\{ \begin{tikzcd}
    R
    \ar[r, "e"]
    \ar[rr, bend left=45, "c"]
    & E
    \ar[r, "l", shift left]
    \ar[r, "r"', shift right]
    & V
  \end{tikzcd} \right\}, \]
where $l\circ e = c = r \circ e$, although we present it syntactically, rather than diagrammatically.

The inverse diagram $\PreReflGraph$ is an inverse replacement~\parencite{SpaceValuedDiagramsTT} of the diagram
\[ \ReflGraph \triangleq \left\{ \begin{tikzcd}
      E
      \ar[r, "l", bend left]
      \ar[r, "r"', bend right]
      & V
      \ar[l, "e"]
    \end{tikzcd} \right\} \]
that indexes the presheaf category of reflexive graphs.

The $(\reppre)$-CwFs $\CPreReflEqv(\CC)$ of pre-reflexive equivalences and $\CReflEqv(\CC)$ of reflexive equivalences are homotopical inverse diagram models over the same base category $\PreReflGraph$, with different sets of morphisms marked as equivalences.
The types of $\CPreReflEqv(\CC)$ and $\CReflEqv(\CC)$ will be types of $\CPreReflGraph(\CC)$ along with some additional contractibility conditions for every marked arrow.
For $\CPreReflEqv(\CC)$, the arrows $l$ and $r$ are marked, while for $\CReflEqv(\CC)$, the arrows $l$, $r$ and $c$ are marked.

The CwF $\CReflEqv(\CC)$ of reflexive equivalences is essentially the same as the CwA of trivial auto-span-equivalences from \textcite{HomotopyTheoryTTs}.

We fix a $(\reppre)$-CwF $\CC$ for the whole section.

\subsection{The category of pre-reflexive graphs}

\begin{defi}
  We define a category $\CPreReflGraph(\CC)$ of \defemph{pre-reflexive graphs} in $\CC$.
  \begin{itemize}
  \item An object $\Gamma$ of $\CPreReflGraph(\CC)$ is a triple $(\Gamma_V,\Gamma_E,\Gamma_R)$ where
    \begin{alignat*}{3}
      & \Gamma_V && \Colon{ }
      && \Ob_\CC, \\
      & \Gamma_E && \Colon{ }
      && \forall (\gamma_l,\gamma_r : \yo(\Gamma_V)) \to \Ty_\CC, \\
      & \Gamma_R && \Colon{ }
      && \forall (\gamma : \Gamma)\ (\gamma_e : \Tm_\CC(\Gamma_E(\gamma,\gamma))) \to \Ty_\CC.
    \end{alignat*}

    We can see $\Gamma_V$ as a type of vertices, $\Gamma_E$ as a dependent type of edges and $\Gamma_R$ as a dependent type of marked loops, which should be thought of as the reflexivity edges.
  \item A morphism $f$ from $\Gamma$ to $\Delta$ is a triple $(f_V,f_E,f_R)$ where
    \begin{alignat*}{3}
      & f_V && \Colon{ }
      && \Gamma_V \to \Delta_V, \\
      & f_E && \Colon{ }
      && \forall (\gamma_l,\gamma_r : \yo(\Gamma_V))\ (\gamma_e : \Tm_\CC(\Gamma_E(\gamma_l,\gamma_r))) \to \Tm_\CC(\Delta_E(f_V(\gamma_l),f_V(\gamma_r))), \\
      & f_R && \Colon{ }
      && \forall (\gamma : \yo(\Gamma_V))\ (\gamma_e : \Tm_\CC(\Gamma_E(\gamma,\gamma)))\ (\gamma_r : \Tm_\CC(\Gamma_R(\gamma,\gamma_e))) \to \Tm_\CC(\Delta_R(f_V(\gamma),f_E(\gamma_e))).
    \end{alignat*}
  \item Identities and compositions are defined in the evident way.
    \defiEnd{}
  \end{itemize}
\end{defi}

We could instead define $\CPreReflGraph(\CC)$ as the (non-equivalent) diagram category $\CC^{\PreReflGraph}$.
Using the more syntactic definition helps with computations in our applications.

\subsection{Reedy types}

\begin{defi}\label{def:mgraph_reedy_types}
  A \defemph{Reedy type} $A$, or \defemph{dependent pre-reflexive graph} $A$, over an object $\Gamma : \CPreReflGraph(\CC)$ is a triple $(A_V,A_E,A_R)$ where
  \begin{alignat*}{3}
    & A_V && \Colon{ }
    && \forall (\gamma : \yo(\Gamma_V)) \to \Ty_\CC, \\
    & A_E && \Colon{ }
    && \forall \gamma_l\ \gamma_r\ (\gamma_e : \Tm_\CC(\Gamma_E(\gamma_l,\gamma_r)))\ (a_l : \Tm_\CC(A_V(\gamma_l)))\ (a_r : \Tm_\CC(A_V(\gamma_r))) \to \Ty_\CC, \\
    & A_R && \Colon{ }
    && \forall \gamma\ \gamma_e\ (\gamma_r : \Tm_\CC(\Gamma_R(\gamma,\gamma_e)))\ (a : \Tm_\CC(A_V(\gamma)))\ (a_e : \Tm_\CC(A_E(\gamma_e,a,a))) \to \Ty_\CC.
  \end{alignat*}

  The substitution of a Reedy-type $A$ along a morphism $f : \Delta \to \Gamma$ in $\CPreReflGraph(\CC)$ is defined by composition with the components of $f$:
  \begin{alignat*}{3}
    & {A[f]}_V && \triangleq{ }
    && \lambda \delta \mapsto A_V(f_V(\delta)), \\
    & {A[f]}_E && \triangleq{ }
    && \lambda \delta_e\ a_l\ a_r \mapsto A_E(f_E(\delta_e),a_l,a_r), \\
    & {A[f]}_R && \triangleq{ }
    && \lambda \delta_r\ a\ a_e \mapsto A_R(f_R(\delta_r,a,a_e)).
  \end{alignat*}

  The functoriality of this definition is easy to check; it follows from the associativity of function composition.
  \defiEnd{}
\end{defi}

\begin{defi}
  The representable Reedy types are defined in the same way:
  a \defemph{representable Reedy type} $A$, or \defemph{representable dependent pre-reflexive graph} $A$, over an object $\Gamma : \CPreReflGraph(\CC)$ is a triple $(A_V,A_E,A_R)$ where
  \begin{alignat*}{3}
    & A_V && \Colon{ }
    && \forall (\gamma : \yo(\Gamma_V)) \to \RepTy_\CC, \\
    & A_E && \Colon{ }
    && \forall \gamma_l\ \gamma_r\ (\gamma_e : \Tm_\CC(\Gamma_E(\gamma_l,\gamma_r)))\ (a_l : \Tm_\CC(A_V(\gamma_l)))\ (a_r : \Tm_\CC(A_V(\gamma_r))) \to \RepTy_\CC, \\
    & A_R && \Colon{ }
    && \forall \gamma\ \gamma_e\ (\gamma_r : \Gamma_R(\gamma,\gamma_e))\ (a : \Tm_\CC(A_V(\gamma)))\ (a_e : \Tm_\CC(A_E(\gamma_e,a,a))) \to \RepTy_\CC.
  \end{alignat*}
  In particular, any representable Reedy type can be seen as a Reedy type by applying the map $\RepTy_\CC \to \Ty_\CC$ to all components.
  \defiEnd{}
\end{defi}

\begin{defi}
  A \defemph{term} of a Reedy type $A$ over $\Gamma : \CPreReflGraph(\CC)$ is a triple $(a_V,a_E,a_R)$ where
  \begin{alignat*}{3}
    & a_V && \Colon{ }
    && \forall (\gamma : \yo(\Gamma_V)) \to \Tm_\CC(A_V(\gamma)), \\
    & a_E && \Colon{ }
    && \forall \gamma_l\ \gamma_r\ (\gamma_e : \Tm_\CC(\Gamma_E(\gamma_l,\gamma_r))) \to \Tm_\CC(A_E(\gamma_e,a_V(\gamma_l),a_V(\gamma_r))), \\
    & a_R && \Colon{ }
    && \forall \gamma\ \gamma_r\ (\gamma_r : \Tm_\CC(\Gamma_R(\gamma,\gamma_e))) \to \Tm_\CC(A_R(\gamma_r,a_V(\gamma),a_E(\gamma_e))).
  \end{alignat*}
  The substitution of a term of a Reedy type along a morphism in $\CPreReflGraph(\CC)$ is also defined by composition with the components of $f$.

  The extension of a context $\Gamma$ by a type $A$ is the context
  \begin{alignat*}{3}
    & {(\Gamma.A)}_V && \triangleq{ }
    && \Gamma_V.A_V, \\
    & {(\Gamma.A)}_E((\gamma_l,a_l),(\gamma_r,a_r)) && \triangleq{ }
    && (\gamma_e:\Gamma_E(\gamma_l,\gamma_r)) \times (A_E(\gamma_e,a_l,a_r)), \\
    & {(\Gamma.A)}_R((\gamma,a),(\gamma_e,a_e)) && \triangleq{ }
    && (\gamma_r:\Gamma_R(\gamma,\gamma_e)) \times (A_R(\gamma_r,a,a_e)).
  \end{alignat*}
  It can be checked that this definition satisfies the required universal property.
  \defiEnd{}
\end{defi}

\begin{constr}
  We equip the Reedy types and the representable Reedy types with $\Unit$- and $\Sigma$- types as follows:
  \begin{alignat*}{3}
    & \Unit_{V} && \triangleq{ }
    && \lambda \gamma \mapsto \Unit, \\
    & \Unit_{E} && \triangleq{ }
    && \lambda \gamma_e\ t_{l}\ t_{r} \mapsto \Unit, \\
    & \Unit_{R} && \triangleq{ }
    && \lambda \gamma_r\ t\ t_{e} \mapsto \Unit, \\
    & {(\Sigma(A,B))}_{V} && \triangleq{ }
    && \lambda \gamma \mapsto (a : A_{V}(\gamma)) \times (b : B_{V}(\gamma,a)) \\
    & {(\Sigma(A,B))}_{E} && \triangleq{ }
    && \lambda \gamma_e\ (a_{l},b_{l})\ (a_{r},b_{r}) \mapsto (a_{e} : A_{E}(\gamma_e,a_{l},a_{r})) \times (b_{e} : B_{E}((\gamma_e,a_{e}),b_{l},b_{r})) \\
    & {(\Sigma(A,B))}_{R} && \triangleq{ }
    && \lambda \gamma_r\ (a,b)\ (a_{e},b_{e}) \mapsto (a_{r} : A_{R}(\gamma_r,a,a_{e})) \times (b_{r} : B_{r}((\gamma_r,a_{r}),b,b_{e})).
  \end{alignat*}

  It is straightforward to check that these definitions are natural and satisfy the universal properties of $\Unit$- and $\Sigma$- types.
  \defiEnd{}
\end{constr}

\begin{constr}
  We equip the Reedy types with $\Pi$-types with arities in the representable Reedy types as follows:
  \begin{alignat*}{3}
    & {(\Pi(A,B))}_{V} && \triangleq{ }
    && \lambda \gamma \mapsto (a : A_{V}(\gamma)) \to B_{V}(\gamma,a) \\
    & {(\Pi(A,B))}_{E} && \triangleq{ }
    && \lambda \gamma_e\ f_l\ f_r \mapsto \forall a_l\ a_r\ (a_{e} : A_{E}(\gamma_e,a_{l},a_{r})) \to B_{E}((\gamma_e,a_{e}),f_{l}(a_l),f_{r}(a_r)), \\
    & {(\Pi(A,B))}_{R} && \triangleq{ }
    && \lambda \gamma_r\ f\ f_e \mapsto \forall a\ a_e\ (a_{r} : A_{R}(\gamma_r,a,a_{e})) \to B_{r}((\gamma_r,a_{r}),f(a),f_e(a_e)).
  \end{alignat*}
  Checking the naturality of this definition is straightfoward, and checking the universal property of the $\Pi$-types is a matter of unfolding the definitions.
  \defiEnd{}
\end{constr}

To summarize, we have described the following construction.
\begin{constr}\label{constr:model_rep_mgraph}
  If $\CC$ is a $(\reppre)$-CwF, then the category $\CPreReflGraph(\CC)$ is equipped with the structure of a $(\reppre)$-CwF whose types are the \emph{Reedy types} (\cref{def:mgraph_reedy_types}), and the projection functor $V : \CPreReflGraph(\CC) \to \CC$ extends to a morphism of $(\reppre)$-CwFs.
  \defiEnd{}
\end{constr}

\subsection{Homotopical Reedy types}

Now assume that the $(\reppre)$-CwF $\CC$ is equipped with contractibility data, \ie{} with families $\isContr$ and $\isRepContr$ over its types and representable types, along with a map $\isRepContr(A) \to \isContr(A)$ for any $A : \RepTy_\CC$.

\begin{defi}\label{def:homotopical_reedy_types}
  A \defemph{$\{l,r\}$-homotopical Reedy type}, or \defemph{dependent pre-reflexive equivalence} is a Reedy type $A$ over $\Gamma : \CPreReflGraph(\CC)$ that satisfies the following two contractibility conditions:
  \begin{alignat*}{1}
    & A .\funr : \forall \gamma_l\ \gamma_r\ (\gamma_e : \Gamma_E(\gamma_l,\gamma_r))\ (a_l : A_V(\gamma_l)) \to \isContr((a_r : A_V(\gamma_r)) \times A_E(\gamma_e,a_l,a_r)), \\
    & A .\funl : \forall \gamma_l\ \gamma_r\ (\gamma_e : \Gamma_E(\gamma_l,\gamma_r))\ (a_r : A_V(\gamma_r)) \to \isContr((a_l : A_V(\gamma_l)) \times A_E(\gamma_e,a_l,a_r)).
  \end{alignat*}

  A representable dependent pre-reflexive equivalence is a representable Reedy type that satisfies:
  \begin{alignat*}{1}
    & A .\funr : \forall \gamma_l\ \gamma_r\ (\gamma_e : \Gamma_E(\gamma_l,\gamma_r))\ (a_l : A_V(\gamma_l)) \to \isRepContr((a_r : A_V(\gamma_r)) \times A_E(\gamma_e,a_l,a_r)), \\
    & A .\funl : \forall \gamma_l\ \gamma_r\ (\gamma_e : \Gamma_E(\gamma_l,\gamma_r))\ (a_r : A_V(\gamma_r)) \to \isRepContr((a_l : A_V(\gamma_l)) \times A_E(\gamma_e,a_l,a_r)).
  \end{alignat*}

  The action of morphism $f : \Delta \to \Gamma$ in $\CPreReflGraph$ a on a dependent pre-reflexive equivalence $A$ is defined by composition with the components of $f$.
  \defiEnd{}
\end{defi}
In other words, a Reedy type is $\{l,r\}$-homotopical when it determines a dependent equivalence in the sense of~\cref{def:dependent_equivalence}.

\begin{defi}
  A \defemph{$\{l,r,c\}$-homotopical Reedy type}, or \defemph{dependent reflexive equivalence} is a dependent pre-reflexive equivalence $A$ over $\Gamma : \CPreReflGraph(\CC)$ that satisfies the following additional contractibility conditions:
  \begin{alignat*}{1}
    & A .\refl : \forall \gamma\ (\gamma_e : \Gamma_E(\gamma,\gamma))\ (\gamma_r : \Gamma_R(\gamma,\gamma_e))\ a \to \isContr((a_e : A_E(\gamma_e,a,a)) \times A_R(\gamma_r,a,a_e)).
  \end{alignat*}

  A representable dependent reflexive equivalence is a representable dependent pre-reflexive equivalence that satisfies:
  \begin{alignat*}{1}
    & A .\refl : \forall \gamma\ (\gamma_e : \Gamma_E(\gamma,\gamma))\ (\gamma_r : \Gamma_R(\gamma,\gamma_e))\ a \to \isRepContr((a_e : A_E(\gamma_e,a,a)) \times A_R(\gamma_r,a,a_e)).
  \end{alignat*}
\end{defi}

We identify a collection of closure conditions on $\isRepContr$ and $\isContr$ that ensure that the $\Unit$-, $\Sigma$- and $\Pi$- type formers of $\CPreReflGraph(\CC)$ lift from the dependent pre-reflexive graphs to the dependent pre-reflexive equivalences.
\begin{lem}\label{lem:contr_closure_conditions_mgraph}
  Assume that the contractibility data is closed under the following operations:
  \begingroup{}\allowdisplaybreaks{}
  \begin{alignat*}{3}
    & \contr^{\rep}_{\Unit} && :{ } && \isRepContr(\Unit \times \Unit), \\
    & \contr^{\rep}_{\Sigma} && :{ } && \forall (A : \RepTy_{\CC})\ (B : \forall a \to \RepTy_{\CC})\ (C : \forall a \to \RepTy_{\CC})\ (D : \forall a\ b\ c \to \RepTy_{\CC}) \\
    &&&&& \quad \to \isRepContr((a : A) \times B(a)) \\
    &&&&& \quad \to (\forall a\ b \to \isRepContr((c : C(a)) \times D(a,b,c))) \\
    &&&&& \quad \to \isRepContr(((a : A) \times (c : C(a))) \times ((b : B(a)) \times D(a,b,c))), \\
    & \contr^{\rep}_{\cong} && :{ } && \forall (A,B : \RepTy_\CC)\\
    &&&&& \quad \to (\Tm_\CC(A) \cong \Tm_\CC(B)) \to \isRepContr(A) \\
    &&&&& \quad \to \isRepContr(B), \\
    & \contr_{\Unit} && :{ } && \isContr(\Unit \times \Unit), \\
    & \contr_{\Sigma} && :{ } && \forall (A : \Ty_{\CC})\ (B : \forall a \to \Ty_{\CC})\ (C : \forall a \to \Ty_{\CC})\ (D : \forall a\ b\ c \to \Ty_{\CC}) \\
    &&&&& \quad \to \isContr((a : A) \times B(a)) \\
    &&&&& \quad \to (\forall a\ b \to \isContr((c : C(a)) \times D(a,b,c))) \\
    &&&&& \quad \to \isContr(((a : A) \times (c : C(a))) \times ((b : B(a)) \times D(a,b,c))), \\
    & \contr_{\cong} && :{ } && \forall (A,B : \Ty_\CC)\\
    &&&&& \quad \to (\Tm_\CC(A) \cong \Tm_\CC(B)) \to \isContr(A) \\
    &&&&& \quad \to \isContr(B), \\
    & \contr_{\Pi} && :{ } && \forall (X : \RepTy_{\CC})\ (Y : \Tm_\CC(X) \to \RepTy_{\CC}) \\
    &&&&& { } \to (\forall x \to \isRepContr(Y(x))) \\
    &&&&& { } \to (\underline{A} : \Tm_\CC(X) \to \Ty_{\CC}) \\
    &&&&& { } \to (\underline{B} : (x : \Tm_\CC(X)) \to \Tm_\CC(Y(x)) \to \Tm_\CC(\underline{A}(x)) \to \Ty_{\CC}) \\
    &&&&& { } \to (\forall x\ y \to \isContr((a : \underline{A}(x)) \times \underline{B}(x,y,a))) \\
    &&&&& { } \to \isContr((a : (x:X) \to \underline{A}(x)) \\
    &&&&& \phantom{{ }\to\isContr} \times ((x:X) \to (y : Y(x)) \to \underline{B}(x,y,a(x)))).
          \tag*{\defiEnd{}}
  \end{alignat*}\endgroup{}

  Then the $\Unit$-, $\Sigma$- and $\Pi$- type structures lift from $\CPreReflGraph(\CC)$ to $\CPreReflEqv(\CC)$, for both types and representable types.
\end{lem}
\begin{proof} We need to check two contractibility conditions for each type former.
  \begin{description}
  \item[Case ${\Unit}$]
    We have to check the following two contractibility conditions:
    \begin{alignat*}{3}
      & \forall \gamma_l\ \gamma_r\ \gamma_e\ t_{l} && \to{ }
      && \isContr(\Unit \times \Unit), \\
      & \forall \gamma_l\ \gamma_r\ \gamma_e\ t_{r} && \to{ }
      && \isContr(\Unit \times \Unit).
    \end{alignat*}
    They are both instances of $\contr_{\Unit}$, or $\contr^{\rep}_{\Unit}$ in the case of representable Reedy types.
  \item[Case ${\Sigma}$]
    We have a dependent pre-reflexive equivalence $A$ over $\Gamma$ and a dependent pre-reflexive equivalence $B$ over $(\Gamma.A)$.

    In order to check that the Reedy type $\Sigma(A,B)$ is a dependent pre-reflexive equivalence, we have to check the following two contractibility conditions:
    \begin{alignat*}{3}
      & \forall \gamma_l\ \gamma_r\ \gamma_e\ (a_l,b_l) && \to{ }
      && \isContr(((a_{r} : A_{V}(\gamma_r)) \times (b_{r} : B_{V}(\gamma_r,a_{r}))) \\
      &&&&& \phantom{\isContr}{ }\times ((a_{e} : A_{E}(\gamma_e,a_{l},a_{r})) \times (b_{e} : B_{E}((\gamma_e,a_{e}),b_{l},b_{r})))), \\
      & \forall \gamma_l\ \gamma_r\ \gamma_e\ (a_r,b_r) && \to{ }
      && \isContr(((a_{l} : A_{V}(\gamma_l)) \times (b_{l} : B_{V}(\gamma_l,a_{l}))) \\
      &&&&& \phantom{\isContr}{ }\times ((a_{e} : A_{E}(\gamma_e,a_{l},a_{r})) \times (b_{e} : B_{E}((\gamma_e,a_{e}),b_{l},b_{r})))).
    \end{alignat*}

    They both follow from $\isContr_{\Sigma}$ and from the contractibility conditions of $A$ and $B$.

    In the case of representable Reedy types, we use $\isContr^\rep_\Sigma$ instead.
  \item[Case ${\Pi}$]
    We have a representable dependent pre-reflexive equivalence $A$ over $\Gamma$ and a dependent pre-reflexive equivalence $B$ over $(\Gamma.A)$.

    In order to check that the Reedy type $\Pi(A,B)$ is a dependent pre-reflexive equivalence over $\Gamma$, we have to check the following two contractibility conditions:
    \begin{alignat*}{3}
      & \forall \gamma_l\ \gamma_r\ \gamma_e\ f_l && \to{ }
      && \isContr((f_r : \forall a_r \to B_V(\gamma_r,a_r)) \\
      &&&&& \phantom{\isContr}{ }\times (f_e : \forall a_l\ a_r\ a_e \to B_E((\gamma_e,a_e),f_l(a_l),f_r(a_r)))), \\
      & \forall \gamma_l\ \gamma_r\ \gamma_e\ f_r && \to{ }
      && \isContr((f_l : \forall a_l \to B_V(\gamma_l,a_l)) \\
      &&&&& \phantom{\isContr}{ }\times (f_e : \forall a_l\ a_r\ a_e \to B_E((\gamma_e,a_e),f_l(a_l),f_r(a_r)))).
    \end{alignat*}

    Up to type isomorphism, the first contractibility condition is an instance of $\contr_\Pi$, whose arguments $(X,Y,\underline{A},\underline{B})$ are instantiated to:
    \begin{alignat*}{3}
      & X && ={ } && A_{V}(\gamma_r), \\
      & Y(a_{r}) && ={ } && (a_{l} : A_{V}(\gamma_l)) \times (a_{e} : A_{E}(\gamma_e,a_{l},a_{r})), \\
      & \underline{A}(a_{r}) && ={ } && B_{V}(\gamma_r,a_{r}), \\
      & \underline{B}(a_r,(a_l,a_e),b_{r}) && ={ } && B_{E}((\gamma_e,a_e),f_{l}(a_{l}),f_r(a_r)).
    \end{alignat*}

    Relying on type isomorphisms is allowed thanks to the operations $\contr^{\rep}_{\cong}$ and $\contr_{\cong}$.

    Up to symmetry, the second contractibility condition is similar the first one.
    \qedhere{}
  \end{description}
\end{proof}

\begin{constr}\label{constr:homotopical_mgraph_model}
  Let $\CC$ be a $(\reppre)$-CwF equipped with operations satisfying the specification of~\cref{lem:contr_closure_conditions_mgraph}.
  Then there is a $(\reppre)$-CwF $\CPreReflEqv(\CC)$ whose types are the dependent pre-reflexive equivalences as defined in~\cref{def:homotopical_reedy_types}.
  There is a $(\reppre)$-CwF $\CPreReflEqv(\CC) \to \CPreReflGraph(\CC)$ lying over the identity functor that forgets the contractibility witnesses of the dependent pre-reflexive equivalences.
  \defiEnd{}
\end{constr}

\begin{lem}\label{lem:contr_closure_conditions_mgraph_2}
  Assume that the contractibility data is closed under the operations of~\cref{lem:contr_closure_conditions_mgraph} and the additional operation
  \begingroup{}\allowdisplaybreaks{}
  \begin{alignat*}{3}
    & \contr_{\Pi,\refl} && :{ }
    && \forall (X : \RepTy_{\CC})\ (Y : \Tm_\CC(X) \to \RepTy_{\CC})\ ((Z : \Tm_\CC(X) \to \RepTy_{\CC})) \\
    &&&&& { } \to (\forall x \to \isRepContr(Y(x))) \\
    &&&&& { } \to (\forall x \to \isRepContr(Z(x))) \\
    &&&&& { } \to (f : \forall x \to \Tm_\CC(Y(x)) \to \Tm_\CC(Z(x))) \\
    &&&&& { } \to (\underline{A} : (x : \Tm_\CC(X)) \to \Tm_\CC(Z(x)) \to \Ty_{\CC}) \\
    &&&&& { } \to (\underline{B} : (x : \Tm_\CC(X)) \to (y : \Tm_\CC(Y(x))) \to \Tm_\CC(\underline{A}(x,f(x,y))) \to \Ty_{\CC}) \\
    &&&&& { } \to (\forall x\ y \to \isContr((a : \underline{A}(x,f(x,y))) \times \underline{B}(x,y,a))) \\
    &&&&& { } \to \isContr((a : (x:X) \to (z : Z(x)) \to \underline{A}(x,z)) \\
    &&&&& \phantom{{ }\to\isContr} \times ((x:X) \to (y : Y(x)) \to \underline{B}(x,y,a(x,f(y))))).
          \tag*{\defiEnd{}}
  \end{alignat*}\endgroup{}

  Then the $\Unit$-, $\Sigma$- and $\Pi$- type structures lift from $\CPreReflEqv(\CC)$ to $\CReflEqv(\CC)$, for both types and representable types.
\end{lem}
\begin{proof} We need to check one contractibility condition for each type former.
  \begin{description}
  \item[Case ${\Unit}$]
    We have to check the following contractibility condition:
    \begin{alignat*}{3}
      & \forall \gamma\ \gamma_e\ \gamma_r\ t && \to{ }
      && \isContr(\Unit \times \Unit).
    \end{alignat*}
    This is an instance of $\contr_{\Unit}$, or $\contr^{\rep}_{\Unit}$ in the case of representable Reedy types.
  \item[Case ${\Sigma}$]
    We have a dependent reflexive equivalence $A$ over $\Gamma$ and a dependent reflexive equivalence $B$ over $(\Gamma.A)$.

    In order to check that the Reedy type $\Sigma(A,B)$ is a dependent reflexive equivalence, we have to check the following contractibility condition:
    \begin{alignat*}{3}
      & \forall \gamma\ \gamma_e\ \gamma_r\ (a,b) && \to{ }
      && \isContr(((a_{e} : A_{E}(\gamma_e,a,a)) \times (b_{e} : B_{E}((\gamma_e,a_{e}),b,b))) \\
      &&&&& \phantom{\isContr}{ }\times ((a_{r} : A_{R}(\gamma_r,a,a_{e})) \times (b_{r} : B_{R}((\gamma_r,a_{r}),b,b_{e})))).
    \end{alignat*}

    It follows from $\isContr_{\Sigma}$ (or $\isContr^\rep_\Sigma$) and from the contractibility conditions of $A$ and $B$.
  \item[Case ${\Pi}$]
    We have a representable dependent reflexive equivalence $A$ over $\Gamma$ and a dependent reflexive equivalence $B$ over $(\Gamma.A)$.

    In order to check that the Reedy type $\Pi(A,B)$ is a dependent reflexive equivalence, we have to check the following contractibility condition:
    \begin{alignat*}{3}
      & \forall \gamma\ \gamma_e\ \gamma_r\ f && \to{ }
      && \isContr((f_e : \forall a_l\ a_r\ (a_e : A_E(\gamma_e,a_l,a_r)) \to B_{E}((\gamma_e,a_e),f(a_l),f(a_r))) \\
      &&&&& \phantom{\isContr}{ }\times (f_r : \forall a\ a_e\ (a_r : A_R(\gamma_r,a,a_e)) \to B_R((\gamma_r,a_r),f(a),f_e(a_e)))).
    \end{alignat*}

    Up to type isomorphism, this contractibility condition is an instance of $\contr_{\Pi,\refl}$, whose arguments $(X,Y,Z,f,\underline{A},\underline{B})$ are instantiated to:
    \begin{alignat*}{3}
      & X && ={ }
      && A_{V}(\gamma), \\
      & Y(a) && ={ }
      && (a_{e} : A_{V}(\gamma_e,a,a)) \times (a_{r} : A_{R}(\gamma_r,a,a_{e})), \\
      & Z(a) && ={ }
      && (a_{r} : A_{V}(\gamma)) \times (a_{e} : A_{E}(\gamma_e,a,a_{r})), \\
      & f(a,(a_{e},a_{r})) && ={ }
      && (a,a_e), \\
      & \underline{A}(a,(a_r,a_e)) && ={ }
      && B_{E}((\gamma_e,a_e),f(a),f(a_r)), \\
      & \underline{B}(a,(a_e,a_r),b_{r}) && ={ }
      && B_{R}((\gamma_r,a_r),f(a),f_e(a_e)).
         \tag*{\qedhere{}}
    \end{alignat*}
  \end{description}
\end{proof}

\begin{constr}\label{constr:homotopical_mgraph_model_2}
  Let $\CC$ be a $(\reppre)$-CwF equipped with operations satisfying the specifications of~\cref{lem:contr_closure_conditions_mgraph} and~\cref{lem:contr_closure_conditions_mgraph_2}.
  Then there is a $(\reppre)$-CwF $\CReflEqv(\CC)$ whose types are the dependent reflexive equivalences as defined in~\cref{def:homotopical_reedy_types}.
  There is a $(\reppre)$-CwF $\CReflEqv(\CC) \to \CPreReflEqv(\CC)$ lying over the identity functor that forgets the additional contractibility conditions of dependent reflexive equivalences.
  \defiEnd{}
\end{constr}

\subsection{Parametricity structures}\label{ssec:parametricity_structures}

We now use the pre-reflexive graph and homotopical pre-reflexive graphs models to specify notions of parametricity structures over $(\reppre)$-CwFs.

\begin{defi}
  A \defemph{parametricity structure} for a $(\reppre)$-CwF $\CC$ is a section $\sem{-}$ of the projection morphism $V : \CPreReflGraph(\CC) \to \CC$.

  A \defemph{$\{l,r\}$-homotopical parametricity structure} for a $(\reppre)$-CwF $\CC$ that is equipped with contractibility data and satisfies the closure conditions of~\cref{lem:contr_closure_conditions_mgraph} is a section $\sem{-}$ of the projection morphism $V : \CPreReflEqv(\CC) \to \CC$.
  \defiEnd{}
\end{defi}

\begin{defi}\label{def:reflexivity_operation}
  Let $\CC$ be a $(\reppre)$-CwF equipped with a parametricity structure $\sem{-}$.

  A \defemph{reflexivity operation} for an object $\Gamma : \CC$ consists of:
  \begin{alignat*}{3}
    & \refl_\Gamma^E && \Colon{ }
    && \forall (\gamma : \yo(\Gamma)) \to \Tm_\CC(\sem{\Gamma}_E(\gamma,\gamma)), \\
    & \refl_\Gamma^R && \Colon{ }
    && \forall (\gamma : \yo(\Gamma)) \to \Tm_\CC(\sem{\Gamma}_R(\gamma,\refl_\Gamma^E(\gamma))).
  \end{alignat*}

  A \defemph{reflexivity operation} for a type $A \Colon \yo(\Gamma) \to \Ty_\CC$ consists of:
  \begin{alignat*}{3}
  & \refl_A^E && \Colon{ }
  && \forall (\gamma : \yo(\Gamma))\ (\gamma_e : \Tm_\CC(\sem{\Gamma}_E(\gamma,\gamma)))\ (\gamma_r : \Tm_\CC(\sem{\Gamma}_R(\gamma,\gamma_e))) \\
  &&&&& \phantom{\forall} (a : \Tm_\CC(A(\gamma))) \to \Tm_\CC(\sem{A}_E(\gamma_e,a,a)), \\
  & \refl_A^R && \Colon{ }
  && \forall (\gamma : \yo(\Gamma))\ (\gamma_e : \Tm_\CC(\sem{\Gamma}_E(\gamma,\gamma)))\ (\gamma_r : \Tm_\CC(\sem{\Gamma}_R(\gamma,\gamma_e))) \\
  &&&&& \phantom{\forall} (a : \Tm_\CC(A(\gamma))) \to \Tm_\CC(\sem{A}_R(\gamma_r,a,\refl_A^E(\gamma_r,a))).
  \end{alignat*}

  A \defemph{reflexivity structure} over $\CC$ consists of reflexivity operations for all objects and types of $\CC$.
  \defiEnd{}
\end{defi}

\begin{prop}\label{prop:contextual_reflexivity_structure}
  If $\CC$ is a contextual $(\reppre)$-CwF that is equipped with reflexivity operations for all types, then it is also equipped with reflexivity operations for all objects, and thus of a reflexivity structure.
\end{prop}
\begin{proof}
  By induction on the contexts of $\CC$, we pose
  \begin{alignat*}{3}
    & \refl_\diamond^E(\star) && \triangleq{ }
    && \star, \\
    & \refl_\diamond^R(\star) && \triangleq{ }
    && \star, \\
    & \refl_{\Gamma.A}^E(\gamma,a) && \triangleq{ }
    && (\refl_\Gamma^E(\gamma), \refl_A^E(\gamma,\refl_\Gamma^E(\gamma),\refl_\Gamma^R(\gamma),a)), \\
    & \refl_{\Gamma.A}^R(\gamma,a) && \triangleq{ }
    && (\refl_\Gamma^R(\gamma), \refl_A^R(\gamma,\refl_\Gamma^E(\gamma),\refl_\Gamma^R(\gamma),a)).
       \tag*{\qedhere}
  \end{alignat*}
\end{proof}

We now fix a $(\reppre)$-CwF $\CC$ equipped with contractibility data and with a homotopical parametricity structure $\sem{-}$, along with a reflexivity structure.

We can then equip it with reflexive equivalences (\cref{def:reflexive_equivalence}) and with dependent equivalences (\cref{def:dependent_equivalence}):
\begin{alignat*}{3}
  & \Id_A(\gamma,a_l,a_r) && \triangleq{ } && \sem{A}(\refl^E_\Gamma(\gamma),a_l,a_r), \\
  & \refl_A(\gamma,a) && \triangleq{ } && \refl_A^E(\refl^E_\Gamma(\gamma),\refl^R_\Gamma(\gamma),a), \\
  & \DId_{A.B}(\gamma,a_l,a_r,a_e,b_l,b_r) && \triangleq{ } && \sem{B}((\refl^E_\Gamma(\gamma),a_e),b_l,b_r).
\end{alignat*}

Our goal is now to investigate the remaining assumption of~\cref{thm:identity_types_from_refleqv}.
We prove some lemmata showing that center and homogeneous all-paths operations are preserved by the operations of~\cref{lem:contr_closure_conditions_mgraph}, under some additional hypothesis for some of the operations.
These lemmata will be needed in~\cref{sec:univalence_proof}.

\begin{constr}\label{constr:heterogeneous_all_paths}
  Let $\Ty'_\CC \hra \Ty_\CC$ be a subfamily of the family of types such that:
  \begin{itemize}
  \item for every $X \Colon \yo(\Gamma) \to \Ty'_\CC$, the contractibility witnesses of the homotomical Reedy type $\sem{X}$ are equipped with centers.
  \item for every $X \Colon \yo(\Gamma) \to \Ty'_\CC$, the dependent types $\sem{X}_E$ and $\sem{X}_R$ are dependent types in $\Ty'_\CC$, \ie{} they factor through $\Ty'_\CC \hra \Ty_\CC$.
  \end{itemize}

  Let $A \Colon \yo(\Gamma) \to \Ty'_\CC$ be a global type and $B \Colon (\gamma : \yo(\Gamma)) \to (a : \Tm_\CC(A(\gamma))) \to \Ty'_\CC$ be a global dependent types.
  If we have a homogeneous all-paths operation for $B$ over $\Gamma.A$, then we construct a heterogeneous all-paths operation for $B$.
\end{constr}
\begin{proof}
  We have to define
  \begin{alignat*}{3}
    & B.\cpath && \Colon{ }
    && \forall \gamma\ (a_{l},a_{r} : \Tm_\CC(A(\gamma)))\ (a_{e} : \Tm_\CC(\Id_{A}(\gamma,a_{l},a_{r}))) \\
    &&&&& \phantom{\forall} (b_{l} : \Tm_\CC(B(\gamma,a_{l})))\ (b_{r} : \Tm_\CC(B(\gamma,a_{r}))) \to \Tm_\CC(\DId_{A.B}(\gamma,a_{e},b_{l},b_{r})).
  \end{alignat*}

  We first transport $b_{r}$ through the equivalence $\DId_{A.B}(\gamma,a_{e}) : B(\gamma,a_{l}) \simeq B(\gamma,a_{r})$.
  \begin{alignat*}{3}
    & \overleftarrow{T_{B}}(\gamma,a_{e},b_{r}) && :{ } &&
    \isContr((b_{l}' : B(\gamma,a_{l})) \times (b_{e}' : \DId_{A.B}(\gamma,a_{e},b_{l}',b_{r}))), \\
    & \overleftarrow{T_{B}}(\gamma,a_{e},b_{r}) && \triangleq{ } &&
    \DId_{A.B}(\gamma,a_{e}).\funl(b_{r}), \\
    & (b_{l}'(\gamma,a_{e},b_{r}),b_{e}'(\gamma,a_{e},b_{r})) &&
    \triangleq{ } && \overleftarrow{T_{B}}.\ccenter(\gamma,a_{e},b_{r}),
  \end{alignat*}
  where the center can be obtained thanks to our first assumption about $\Ty'_\CC$.

  Now, using the homogeneous all-paths operation for $B$ at $a_l$, we obtain a homogeneous path between $b_{l}$ and $b_{l}'$.
  \begin{alignat*}{3}
    & b_{e}''(\gamma,a_{e},b_{l},b_{r}) && :{ } &&
    \Id_{B(\gamma,a_{l})}(\gamma,b_{l},b_{l}'), \\
    & b_{e}''(\gamma,a_{e},b_{l},b_{r}) && \triangleq{ } &&
    B.\chpath((\gamma,a_{l}),b_{l},b_{l}').
  \end{alignat*}

  It remains to compose the homogeneous path $b_{e}''$ with the homogeneous path $b_{e}'$.
  \begin{alignat*}{3}
    & I_{B}(\gamma,a_{e},b_{r})(b_{l}) && \triangleq{ } &&
    \DId_{A.B}(\gamma,a_{e},b_{l},b_{r}), \\
    & T_{I}(\gamma,a_{e},b_{l},b_{r}) && :{ } &&
    I_{B}(\gamma,a_{e},b_{r})(b_{l}) \simeq I_{B}(\gamma,a_{e},b_{r})(b_{l}'), \\
    & T_{I}(\gamma,a_{e},b_{l},b_{r}) && \triangleq{ } &&
    \DId_{B.I_{B}}((\gamma,a_{e},b_{r}),b_{e}''), \\
    & \overleftarrow{T_{I}}(\gamma,a_{e},b_{l},b_{r}) && :{ } &&
    \isContr((b_{e} : I_{B}(\gamma,a_{e},b_{r})(b_{l})) \times \cdots), \\
    & \overleftarrow{T_{I}}(\gamma,a_{e},b_{l},b_{r}) && \triangleq{ } &&
    T_{I}(\gamma,a_{e},b_{l},b_{r}).\funl(b_{e}''(\gamma,a_{e},b_{l},b_{r})), \\
    & (b_{e}(\gamma,a_{e},b_{l},b_{r}),\_{}) && \triangleq{ } &&
    \overleftarrow{T_{I}}.\ccenter(\gamma,a_{e},b_{l},b_{r}),
  \end{alignat*}
  where the center can be obtained thanks to our first assumption about $\Ty'_\CC$ and the fact that the dependent type $I_{B}$ lands in $\Ty'_\CC$.

  We can finally pose:
  \begin{alignat*}{1}
    & B.\cpath(\gamma,a_{e},b_{l},b_{r}) \triangleq b_{e}(\gamma,a_{e},b_{l},b_{r}).
    \tag*{\qedhere{}}
  \end{alignat*}
\end{proof}

\begin{constr}
  Let $A,B \Colon \yo(\Gamma) \to \Ty_\CC$ be two types related by an isomorphism $\alpha : \forall \gamma \to \Tm_\CC(A(\gamma)) \cong \Tm_\CC(B(\gamma))$.

  If $A$ is equipped with a center operation (\resp with a homogeneous all-paths operations), then we equip $B$ with a center operation (\resp with a homogeneous all-paths operations).
\end{constr}
\begin{proof}
  Equipping $B$ with a center operation is straightforward:
  \begin{alignat*}{3}
    & B.\ccenter(\gamma) && \triangleq{ } && \alpha(\gamma,A.\ccenter(\gamma));
  \end{alignat*}

  For the homogeneous all-paths operations, we show that the isomorphism $\alpha$ lifts an isomorphism between $\Id_A$ and $\Id_B$.

  We have $\sem{\alpha}_E : \forall \gamma_l\ \gamma_r\ \gamma_e\ a_l\ a_r \to \Tm_\CC(\sem{A}_E(\gamma_e,a_l,a_r)) \cong \Tm_\CC(\sem{B}_E(\gamma_e,\alpha(\gamma_l,a_l),\alpha(\gamma_r,a_r)))$.
  Thus $\sem{\alpha}_E(\refl^E_\Gamma(\gamma)) : \forall a_l\ a_r \to \Tm_\CC(\Id_A(\gamma,a_l,a_r)) \cong \Tm_\CC(\Id_B(\gamma,\alpha(\gamma,a_l),\alpha(\gamma,a_r)))$ is an isomorphism between $\Id_A$ and $\Id_B$.

  We can now pose
  \begin{alignat*}{3}
    & B.\chpath(\gamma,b_l,b_r) && \triangleq{ }
    && \sem{\alpha_E}(\refl^E_\Gamma(\gamma), A.\chpath(\gamma,\alpha^{-1}(\gamma,b_l),\alpha^{-1}(\gamma,b_r))).
       \tag*{\qedhere}
  \end{alignat*}
\end{proof}

\begin{constr}\label{constr:center_and_all_paths_unit}
  The type $\Unit \times \Unit$ has center and homogeneous all-paths operations over any context $\Gamma : \CC$.
\end{constr}
\begin{proof}
  The center is $(\tt,\tt)$ over any context.

  By definition of $\Unit$- and $\Sigma$- types in $\CPreReflGraph(\CC)$, we compute $\Id_{\Unit \times \Unit}(\gamma,-,-) = \sem{\Unit \times \Unit}_E(\cdots) = \Unit \times \Unit$.
  Thus the homogeneous all-paths operations can also be defined by $(\tt,\tt)$ over any context.
\end{proof}

\begin{constr}\label{constr:center_sigma}
  Assume given the data of:
  \begingroup{}\allowdisplaybreaks{}
  \begin{alignat*}{3}
    & \Gamma && :{ }
    && \CC, \\
    & A && \Colon{ }
    && \forall \gamma \to \Ty_\CC, \\
    & B && \Colon{ }
    && \forall \gamma\ a \to \Ty_\CC, \\
    & C && \Colon{ }
    && \forall \gamma\ a \to \Ty_\CC, \\
    & D && \Colon{ }
    && \forall \gamma\ a\ b\ c \to \Ty_\CC, \\
  \end{alignat*}\endgroup{}
  along with center operations for the types
  \[ (a : A(\gamma)) \times B(\gamma,a) \]
  over $(\gamma : \yo(\Gamma))$ and
  \[ (c : C(\gamma,a)) \times D(\gamma,a,b,c) \]
  over $((\gamma,a,b) : \yo(\Gamma.A.B))$.

  We construct an center over $(\gamma : \yo(\Gamma))$ for the type
  \[ ((a : A(\gamma)) \times (b : B(\gamma,a))) \times ((c : C(\gamma,a)) \times (d : D(\gamma,a,b,c))).
  \]
\end{constr}
\begin{proof}[Construction]
  We write $AB \triangleq (a : A(\gamma)) \times B(\gamma,a)$, \etc.

  From our hypotheses, we have
  \begin{alignat*}{3}
    & \angles{a,b}(\gamma) && \triangleq{ } && \ccenter_{AB}(\gamma), \\
    & \angles{c,d}(\gamma) && \triangleq{ } && \ccenter_{CD}(\gamma,ab(\gamma)).
  \end{alignat*}

  Thus we can pose
  \begin{alignat*}{3}
    & \ccenter_{ABCD}(\gamma) && \triangleq{ } && ((a(\gamma),c(\gamma)),(b(\gamma),d(\gamma))).
    \tag*{\qedhere{}}
  \end{alignat*}
\end{proof}

\begin{constr}\label{constr:all_paths_sigma}
  Assume given the data of:
  \begingroup{}\allowdisplaybreaks{}
  \begin{alignat*}{3}
    & \Gamma && :{ }
    && \CC, \\
    & A && \Colon{ }
    && \forall \gamma \to \Ty_\CC, \\
    & B && \Colon{ }
    && \forall \gamma\ a \to \Ty_\CC, \\
    & C && \Colon{ }
    && \forall \gamma\ a \to \Ty_\CC, \\
    & D && \Colon{ }
    && \forall \gamma\ a\ b\ c \to \Ty_\CC, \\
  \end{alignat*}\endgroup{}
  along with a homogeneous all-paths operation for the type
  \[ (a : A(\gamma)) \times B(\gamma,a) \]
  over $(\gamma : \yo(\Gamma))$ and a heterogeneous all-paths operation for the dependent type
  \[ (a,b) \mapsto (c : C(\gamma,a)) \times D(\gamma,a,b,c) \]
  over $(\gamma : \yo(\Gamma))$.

  We construct a homogeneous all-paths operation over $(\gamma : \yo(\Gamma))$ for the type
  \[ ((a : A(\gamma)) \times (b : B(\gamma,a))) \times ((c : C(\gamma,a)) \times (d : D(\gamma,a,b,c))).
  \]
\end{constr}
\begin{proof}[Construction]
  We pose $AB \triangleq (a : A(\gamma)) \times B(\gamma,a)$, \etc.
  We also write $ab$ instead of $(a,b)$, \etc.

  Our goal is to define
  \begin{alignat*}{3}
    & \chpath_{ABCD} && :{ }
    && \forall \gamma\ acbd_{l}\ acbd_{r} \to \Id_{ACBD}(acbd_{l},acbd_{r}).
  \end{alignat*}

  We pose
  \begin{alignat*}{3}
    & ab_{e}(\gamma,ab_{l},ab_{r}) && \triangleq{ }
    && \chpath_{AB}(\gamma,ab_{l},ab_{r}), \\
    & cd_{e}(\gamma,acbd_{l},acbd_{r}) && \triangleq{ }
    && \cpath_{CD}(\gamma,ab_{e},cd_{l},cd_{r}).
  \end{alignat*}

  Thus, we can define
  \begin{alignat*}{3}
    & \chpath_{ABCD}(\gamma,acbd_{l},acbd_{r}) && \triangleq{ }
    && acbd_{e}(\gamma,acbd_{l},acbd_{r}).
       \tag*{\qedhere{}}
  \end{alignat*}
\end{proof}

\begin{constr}\label{constr:center_pi}
  Also assume given the data of:
  \begingroup{}\allowdisplaybreaks{}
  \begin{alignat*}{3}
    & \Gamma && :{ }
    && \CC, \\
    & X && \Colon{ }
    && \yo(\Gamma) \to \RepTy_{\CC}, \\
    & Y && \Colon{ }
    && \forall \gamma \to \Tm_\CC(X(\gamma)) \to \RepTy_{\CC}, \\
    & c_Y && \Colon{ }
    && \forall \gamma\ x \to \isRepContr(Y(\gamma,x)), \\
    & A && \Colon{ }
    && \forall \gamma \to \Tm_\CC(X(\gamma)) \to \Ty_{\CC}, \\
    & B && \Colon{ }
    && \forall \gamma\ (x : \Tm_\CC(X(\gamma)))\ (y : \Tm_\CC(Y(\gamma,x))) \to \Tm_\CC(A(\gamma,x)) \to \Ty_{\CC},
  \end{alignat*}\endgroup{}
  along with a center operation for the type
  \[ AB(\gamma) \triangleq (a : A(\gamma,x)) \times B(\gamma,x,y,a) \]
  over $((\gamma,x,y) : \yo(\Gamma.X.Y))$.

  We construct an center over $(\gamma : \yo(\Gamma))$ for the type
  \[ (a : (x : X(\gamma)) \to A(\gamma,x)) \times (b : (x : X(\gamma)) \to (y : Y(\gamma,x)) \to B(\gamma,x,y,a(x))).
  \]
\end{constr}
\begin{proof}[Construction]

  Our goal is to define the following:
  \begingroup{}\allowdisplaybreaks{}
  \begin{alignat*}{3}
    & \ccenter_\Pi(\gamma) && :{ }
    && (a : (x : X(\gamma)) \to A(\gamma,x)) \\
    &&&&& { } \times (b : (x : X(\gamma))\ (y : Y(\gamma,x)) \to B(\gamma,x,y,a(x))).
  \end{alignat*}\endgroup{}
  Since $Y$ is a family of contractible representable types, we can find its centers of contraction.
  \begingroup{}\allowdisplaybreaks{}
  \begin{alignat*}{3}
    & y_{0}(\gamma,x) && \triangleq{ } && \ccenter_Y(\gamma,x), \\
  \end{alignat*}\endgroup{}
  We then obtain elements of $A$ from the centers of contraction of $AB$.
  \begingroup{}\allowdisplaybreaks{}
  \begin{alignat*}{3}
    & a(\gamma,x) && :{ }
    && A(\gamma,x), \\
    & b_{0}(\gamma,x) && :{ }
    && B(\gamma,x,y_{0}(\gamma,x),a(\gamma,x)), \\
    & (a(\gamma,x),b_{0}(\gamma,x)) && \triangleq{ }
    && \ccenter_{AB}(\gamma,x,y_{0}(\gamma,x)).
  \end{alignat*}\endgroup{}
  We then want to transport $b_{0}$ over paths in $Y$, using the fact that $Y$ has all paths.
  We start by computing paths from $y_{0}$ to any element in $Y$.
  \begingroup{}\allowdisplaybreaks{}
  \begin{alignat*}{3}
    & p_{Y}(\gamma,x,y) && :{ }
    && \Id_{Y[x]}(\gamma,y_{0}(x),y), \\
    & p_{Y}(\gamma,x,y) && \triangleq{ }
    && \chpath_Y((\gamma,x),y_{0}(\gamma,x),y).
  \end{alignat*}\endgroup{}
  We now consider the transport of $b_{0}$ over $p_{Y}$.
  \begingroup{}\allowdisplaybreaks{}
  \begin{alignat*}{3}
    & T_{B}(\gamma,x,y) && :{ }
    && B(\gamma,x,y_{0}(\gamma,x),a(\gamma,x)) \simeq B(\gamma,x,y,a(\gamma,x)), \\
    & T_{B}(\gamma,x,y) && \triangleq{ }
    && \DId_{(y:Y).B(\gamma,x,y,a(\gamma,x))}((\gamma,x),p_{Y}(\gamma,x,y)), \\
    & \overrightarrow{T_{B}}(\gamma,x,y) && :{ }
    && \isContr((b : B(\gamma,x,y,a(\gamma,x))) \times \cdots), \\
    & \overrightarrow{T_{B}}(\gamma,x,y) && \triangleq{ }
    && T_{B}(\gamma,x,y).\funr(b_{0}(\gamma,x)), \\
    & b(\gamma,x,y) && \triangleq{ }
    && \overrightarrow{T_{B}}.\ccenter(\gamma,x,y).
  \end{alignat*}\endgroup{}
  We can now conclude the definition:
  \begingroup{}\allowdisplaybreaks{}
  \begin{alignat*}{3}
    & \ccenter_\Pi(\gamma) && \triangleq{ }
    && (\lambda x \mapsto a(\gamma,x), \lambda x\ y \mapsto b(\gamma,x,y)).
       \tag*{\qedhere{}}
  \end{alignat*}\endgroup{}
\end{proof}

\begin{constr}\label{constr:all_paths_pi}
  Assume that the weakly stable identity type introduction structure $\Id_{-}$ on $\RepTy_\CC$ can be equipped with a weakly stable elimination structure.

  Also given the data of:
  \begingroup{}\allowdisplaybreaks{}
  \begin{alignat*}{3}
    & \Gamma && :{ }
    && \CC, \\
    & X && \Colon{ }
    && \yo(\Gamma) \to \RepTy_{\CC}, \\
    & Y && \Colon{ }
    && \forall \gamma \to \Tm_\CC(X(\gamma)) \to \RepTy_{\CC}, \\
    & A && \Colon{ }
    && \forall \gamma \to \Tm_\CC(X(\gamma)) \to \Ty_{\CC}, \\
    & B && \Colon{ }
    && \forall \gamma\ (x : \Tm_\CC(X(\gamma)))\ (y : \Tm_\CC(Y(\gamma,x))) \to \Tm_\CC(A(\gamma,x)) \to \Ty_{\CC},
  \end{alignat*}\endgroup{}
  along with a heterogeneous all-paths operation for the dependent type
  \[ AB \triangleq (x,y) \mapsto (a : A(\gamma,x)) \times B(\gamma,x,y,a) \]
  over $(\gamma : \yo(\Gamma))$.

  We construct a homogeneous all-paths operation over $(\gamma : \yo(\Gamma))$ for the type
  \[ (a : (x : X(\gamma)) \to A(\gamma,x)) \times (b : (x : X(\gamma)) \to (y : Y(\gamma,x)) \to B(\gamma,x,y,a(x))).
  \]
\end{constr}
\begin{proof}[Construction]
  We write $ab$ instead of $(a,b)$, $a_{lr}$ instead of $(a_l,a_r)$, \etc.

  Our goal is to define the following:
  \begingroup{}\allowdisplaybreaks{}
  \begin{alignat*}{3}
    & \chpath_\Pi(\gamma,ab_{l},ab_{r}) && :{ }
    && (a_{e} : \forall x_{lre} \to \sem{A}_E((\refl_{\Gamma}(\gamma),x_{e}), a_{l}(x_{l}), a_{r}(x_{r}))) \\
    &&&&& \times (b_{e} : \forall x_{lre}\ y_{lre} \to \sem{B}_E((\refl_{\Gamma}(\gamma),x_{e},y_{e},a_{e}(x_e)), b_{l}(x_{l},y_{l}), b_{r}(x_{r},y_{r}))).
  \end{alignat*}\endgroup{}
  Since $Y$ is a family of contractible representable types, we can find paths in $Y$ over any path in $X$.
  \begingroup{}\allowdisplaybreaks{}
  \begin{alignat*}{3}
    & y_{0}(\gamma,x) && \triangleq{ }
    && \ccenter_Y(\gamma,x), \\
    & y_{0e}(\gamma,x_{lre}) && :{ }
    && \sem{Y}_E((\refl_{\Gamma}(\gamma),x_{e}), y_{0}(\gamma,x_{l}), y_{0}(\gamma,x_{r})), \\
    & y_{0e}(\gamma,x_{lre}) && \triangleq{ }
    && \sem{y_{0}}_E(\refl_{\Gamma},x_{e}).
  \end{alignat*}\endgroup{}
  Now using the heterogeneous all-paths operation of $AB$, we obtain paths in $A$ over any path in $X$.
  \begingroup{}\allowdisplaybreaks{}
  \begin{alignat*}{3}
    & a_{e}(\gamma,ab_{lr},x_{lre}) && :{ }
    && \sem{A}_E((\refl_{\Gamma}(\gamma),x_{e}), a_{l}(x_{l}), a_{r}(x_{r})), \\
    & b_{0e}(\gamma,ab_{lr},x_{lre}) && :{ }
    && \sem{B}_{E}((\refl_{\Gamma}(\gamma),x_{e},y_{e},a_{e}(\gamma,a_{lr},x_{lre})),b_{l}(x_l,y_{0}(\gamma,x_{l})),b_{r}(x_r,y_{0}(\gamma,x_{r}))), \\
    & \angles{a_{e},b_{0e}}(\gamma,ab_{lr},x_{lre}) && \triangleq{ }
    && \cpath_{AB}(\gamma,x_e,y_{0e}(\gamma,x_{lre}),ab_{l}(x_l,y_{0}(\gamma,x_{l})),ab_{r}(x_r,y_{0}(\gamma,x_{r}))).
  \end{alignat*}\endgroup{}
  In order to define $b_{e}$, we transport $b_{0e}$ over squares in $Y$.
  These squares are constructed using the fact that $Y$ is contractible; since we already know that $\Id_{-}$ on representable types has a weakly stable elimination structure, we omit the precise construction of these squares.
  \begingroup{}\allowdisplaybreaks{}
  \begin{alignat*}{3}
    & s_{Y}(\gamma,x_{lre},y_{lre}) && :{ }
    && \Id_{\Id_{Y}}((\gamma,y_{l},y_{r}),y_{0e}(\gamma,x_{e}),y_{e}).
  \end{alignat*}\endgroup{}
  We can now transport $b_{0e}$ over $s_{Y}$.
  \begingroup{}\allowdisplaybreaks{}
  \begin{alignat*}{3}
    & I_{B}(\gamma,x_{lre},a_{lre},y_{lr},b_{lr},y_{e}) && \triangleq{ }
    && \sem{B}((\refl_{\Gamma}(\gamma),x_e,y_e,a_{e}),b_{l},b_{r}), \\
    & T_{I_{B}}(\gamma,ab_{lr},x_{lre},y_{lre}) && :{ }
    && I_{B}(\gamma,x_{lre},a_{e}(\gamma,ab_{lr},x_{lre}), b_{l}(xy_{l})), b_{r}(xy_{r}), y_{0e}(\gamma,x_{lre})) \\
    &&&&& \simeq I_{B}(\gamma,x_{lre},a_{e}(\gamma,ab_{lr},x_{lre}), b_{l}(xy_{l}), b_{r}(xy_{r}), y_{e}), \\
    & T_{I_{B}}(x,y) && \triangleq{ }
    && \DId_{I_{B}}(\dots, s_{Y}(\gamma,x_{lre},y_{lre})), \\
    & \overrightarrow{T_{I_{B}}}(\gamma,ab_{lr},x_{lre},y_{lre}) && :{ }
    && \isContr((b_{e} : I_{B}(\gamma,ab_{lr},x_{lre},y_{lre})) \times \cdots), \\
    & \overrightarrow{T_{I_{B}}}(\gamma,ab_{lr},x_{lre},y_{lre}) && \triangleq{ }
    && T_{I_{B}}(\gamma,ab_{lr},x_{lre},y_{lre}).\funr(b_{0e}(\gamma,ab_{lr},x_{lre})), \\
    & b_{e}(\gamma,ab_{lr},x_{lre},y_{lre}) && :{ }
    && \sem{B}((\refl_{\Gamma}(\gamma),x_{e},y_{e},a_{e}(\gamma,ab_{lr},x_{e})), b_{l}(xy_{l}), b_{r}(xy_{r})), \\
    & b_{e}(\gamma,ab_{lr},x_{lre},y_{lre}) && \triangleq{ }
    && \overrightarrow{T_{I_{B}}}.\ccenter(\gamma,ab_{lr},x_e,y_e).1.
  \end{alignat*}\endgroup{}
  Finally, we can pose
  \begingroup{}\allowdisplaybreaks{}
  \begin{alignat*}{3}
    & \chpath_\Pi(\gamma,ab_{lr}) && \triangleq{ }
    && (\lambda x_{lre} \mapsto a_{e}(\gamma,a_{lr},x_{lre}), \lambda x_{lre}\ y_{lre} \mapsto b_{e}(\gamma,ab_{lr},x_{lre},y_{lre})).
       \tag*{\qedhere{}}
  \end{alignat*}\endgroup{}
\end{proof}


%% file: univalence_proof.tex
\section{Proving external univalence}\label{sec:univalence_proof}

We fix a SOGAT $\Th$ equipped with homotopy relations.
Our goal is to prove that $\Th$ satisfies external univalence, that is to equip the $(\reppre)$-CwF $\Th$ with weakly stable identity types satisfying function extensionality and saturation with respect to the homotopy relations of $\Th$.
In this section we essentially give a construction of this data (the weakly stable identity types) from the facts that every operation of $\Th$ preserves the homotopy relations, and that the homotopy relations are equipped with some operations which essentially say that the homotopy relations should be reflexive and admit fillers of $1$- and $2$- dimensional cubes.
More precisely, we rely on~\cref{thm:identity_types_from_refleqv} to construct the identity types, and the constructions of~\cref{sec:mgraph_model} to satisfy the hypothesis of~\cref{thm:identity_types_from_refleqv}, that is to construct reflexive equivalences and dependent equivalences with respect to some contractibility data, such that the contractible types have centers and all-paths operations.

\subsection{Internal model in reflexive equivalences}

\begin{customasm}{(A1)}\label{asm:refleqv_internal_model}
  We assume given a parametricity structure on $\Th$, \ie{} a section
  \[ \sem{-} : \Th \to \CPreReflGraph(\Th) \]
  of the projection map $V$.
  \defiEnd{}
\end{customasm}
By the universal property of $\Th$, this amounts to equipping $\CPreReflGraph(\Th)$ with an internal model of $\Th$ that is displayed over $\Th$.

Concretely, we have to interpret the sorts, operations and equations of $\Th$ in $\CPreReflGraph(\Th)$.
The sorts have to be interpreted by Reedy types, which are typically given by the homotopy relations $(\sim)$, up to the fact that the Reedy types are more dependent than the homotopy relations.
The representable sorts have to be interpreted by representable Reedy types.

Then to give an interpretation of an operation, we exactly need to show that it preserves the homotopy relations.

Depending on the theory $\Th$, ensuring that the equations of $\Th$ are satisfied in this model may be quite tricky.
When $\Th$ is a type theory with the usual $\beta$- and $\eta$- equalities, we have to use the relational definition of equivalences (\cref{exa:homotopy_relations_th_id}).
While other definitions (such as half-adjoints equivalences, \etc{}) are equivalent up to homotopy, they do not seem to satisfy the necessary computational properties for this construction.

\subsection{Contractibility data}\label{ssec:proof_contractibility_data}

We define the following inductive families, internally to $\CPsh(\Th)$.
\begin{alignat*}{3}
  & \isBRepContr, \isBContr, \isRepContr, \isContr && :{ } && \RepTy_{\Th} \to \UPsh.
\end{alignat*}

The families $\isBRepContr$ and $\isBContr$ describe the \emph{basic} contractible types.
We introduce them so as to be able to state our last assumption later (\cref{asm:basic_center_and_all_paths}).
The family $\isBContr$ is generated by the following (non-recursive) constructors, for every generating type $\iS : \GenTy_{\Th}$:
\begin{alignat*}{3}
  & \contr_{\iS,l} && :{ } && \forall (\sigma_{l}, \sigma_{r} : \partial \iS)\ (\sigma_{e} : \sem{\partial \iS}_{E}(\sigma_{l},\sigma_{r}))\ (a_{l} : \iS(\sigma_{l})) \\
  &&&&& \quad \to \isContr((a_{r} : \iS(\sigma_{r})) \times \sem{\iS}_{E}(\sigma_{e},a_{l},a_{r})), \\
  & \contr_{\iS,r} && :{ } && \forall (\sigma_{l}, \sigma_{r} : \partial \iS)\ (\sigma_{e} : \sem{\partial \iS}_{E}(\sigma_{l},\sigma_{r}))\ (a_{r} : \iS(\sigma_{r})) \\
  &&&&& \quad \to \isContr((a_{l} : \iS(\sigma_{l})) \times \sem{\iS}_{E}(\sigma_{e},a_{l},a_{r})), \\
  & \contr_{\iS,{\sim}} && :{ } && \forall (\sigma : \partial \iS)\ (x : \iS(\sigma)) \\
  &&&&& \quad \to \isContr((y : \iS(\sigma)) \times (x \sim_{\iS(\sigma)} y)).
\end{alignat*}
The constructor $\contr_{\iS,{\sim}}$ is included to make sure that we construct identity types satisfying saturation with respect to the homotopy relations $(\sim)$.
The family $\isBRepContr$ is generated by the same constructors, but restricted to representable types.
This means that the constructors $\contr_{\iS,l}$, $\contr_{\iS,r}$ and $\contr_{\iS,{\sim}}$ are only included for $\iS : \GenRepTy_{\Th}$.
There is an evident map $\isBRepContr(A) \to \isBContr(A)$ for $A : \RepTy_{\Th}$.

The family $\isRepContr$ is inductively generated by the following constructors:
\begin{alignat*}{3}
  & \contr^{\rep}_{b} && :{ } && \forall A \to \isBRepContr(A) \to \isRepContr(A), \\
  & \contr^{\rep}_{\Unit} && :{ } && \isRepContr(\Unit \times \Unit), \\
  & \contr^{\rep}_{\Sigma} && :{ } && \forall (A : \RepTy_{\Th})\ (B : \forall a \to \RepTy_{\Th})\ (C : \forall a \to \RepTy_{\Th})\ (D : \forall a\ b\ c \to \RepTy_{\Th}) \\
  &&&&& \quad \to \isRepContr((a : A) \times B(a)) \\
  &&&&& \quad \to (\forall a\ b \to \isRepContr((c : C(a)) \times D(a,b,c))) \\
  &&&&& \quad \to \isRepContr(((a : A) \times (c : C(a))) \times ((b : B(a)) \times D(a,b,c))), \\
  & \contr^\rep_{\cong} && :{ } && \forall (A,B : \RepTy_\Th) \to (\Tm_\Th(A) \cong \Tm_\Th(B)) \to \isRepContr(A) \to \isRepContr(B), \\
\end{alignat*}

The family $\isContr : \Ty_{\Th} \to \UPsh$ is inductively generated by the following constructors:
\begingroup{}\allowdisplaybreaks{}
\begin{alignat*}{3}
  & \contr_{\mathsf{rep}} && :{ } && \forall A \to \isRepContr(A) \to \isContr(A), \\
  & \contr_{b} && :{ } && \forall A \to \isBContr(A) \to \isContr(A), \\
  & \contr_{\Unit} && :{ } && \isContr(\Unit \times \Unit), \\
  & \contr_{\Sigma} && :{ } && \forall (A : \Ty_{\Th})\ (B : \forall a \to \Ty_{\Th})\ (C : \forall a \to \Ty_{\Th})\ (D : \forall a\ b\ c \to \Ty_{\Th}) \\
  &&&&& \quad \to \isContr((a : A) \times B(a)) \\
  &&&&& \quad \to (\forall a\ b \to \isContr((c : C(a)) \times D(a,b,c))) \\
  &&&&& \quad \to \isContr(((a : A) \times (c : C(a))) \times ((b : B(a)) \times D(a,b,c))), \\
  & \contr_{\cong} && :{ } && \forall (A,B : \Ty_\Th) \to (\Tm_\Th(A) \cong \Tm_\Th(B)) \to \isContr(A) \to \isContr(B), \\
  & \contr_{\Pi} && :{ } && \forall (X : \RepTy_{\Th})\ (Y : X \to \RepTy_{\Th}) \\
  &&&&& \quad \to (\forall x \to \isRepContr(Y(x))) \\
  &&&&& \quad \to (A : X \to \Ty_{\Th}) \\
  &&&&& \quad \to (B : (x : X) \to (y : Y(x)) \to A(x) \to \Ty_{\Th}) \\
  &&&&& \quad \to (\forall x\ y \to \isContr((a : A(x)) \times B(x,y,a))) \\
  &&&&& \quad \to \isContr((a : (x:X) \to A(x)) \\
  &&&&& \phantom{{ }\to\isContr} \times ((x:X) \to (y:Y(x)) \to B(x,y,a(x)))).
\end{alignat*}\endgroup{}

Now that we have defined this contractibility data, we observe that it is equipped, by definition, with the operations of~\cref{lem:contr_closure_conditions_mgraph}.
Thus~\cref{constr:homotopical_mgraph_model} provides the pre-reflexive equivalences model $\CPreReflEqv(\Th)$ with respect to this contractibility data.

\begin{lem}
  The section $\sem{-} : \Th \to \CPreReflGraph(\Th)$ factors through $\CPreReflEqv(\Th) \to \CPreReflGraph(\Th)$.
\end{lem}
\begin{proof}
  Note that $\CPreReflEqv(\Th) \to \CPreReflGraph(\Th)$ is bijective on contexts and terms.
  Thus it suffices to consider the types.

  We prove by induction on the types of $\Th$ that for every type $A$ (\resp representable type $A$), the Reedy type $\sem{A}$ (\resp representable Reedy type $\sem{A}$) is homotopical.
  For any generating type $\iS$ (\resp generating representable type $\iS$), this is handled by the constructors $\contr_{\iS,l}$ and $\contr_{\iS,r}$ (\resp{} $\contr^\rep_{\iS,l}$ and $\contr^\rep_{\iS,r}$).
  All of the other cases of the induction (for the $\Unit$-, $\Sigma$- and $\Pi$- type formers) follow from the fact that $\CPreReflEqv(\Th)$ is a $(\reppre)$-CwF.
\end{proof}

\subsection{Reflexivity maps}

We now try to define a reflexivity structure (\cref{def:reflexivity_operation}) over $\Th$.

\begin{defi}[Congruence operation]\label{defi:congruence_operation}
  Let $A \Colon \partial A \to \RepTy_{\Th}$ and $B \Colon \partial B \to \Ty_{\Th}$ be two global dependent types, where $A$ is a representable dependent type.

  A \defemph{congruence operation} from $A$ to $B$ consists of:
  \begin{alignat*}{3}
    & \ap^{E}_{A,B} && \Colon{ } &&
    \forall (\sigma : \partial A)\ (\sigma_{e} : \sem{\partial A}_{E}(\sigma,\sigma))\ (\sigma_{r} : \sem{\partial A}_{R}(\sigma,\sigma_{e})) \\
    &&&&& \phantom{\forall} (\tau : A(\sigma) \to \partial B) \\
    &&&&& \phantom{\forall} (\tau_{e} : \forall a_{l}\ a_{r}\ a_{e} \to \sem{\partial B}_{E}(\tau(a_{l}),\tau(a_{r}))) \\
    &&&&& \phantom{\forall} (\tau_{r} : \forall a\ a_{e}\ a_{r} \to \sem{\partial B}_{R}(\tau(a),\tau_{e}(a_{e}))) \\
    &&&&& \phantom{\forall} (b : (a : A(\sigma)) \to B(\tau(a))) \\
    &&&&& \quad \to \forall a_{l}\ a_{r}\ (a_{e} : \sem{A}_{E}(\sigma_{e},a_{l},a_{r})) \to \sem{B}_{E}(\tau_{e}(a_{e}),b(a_{l}),b(a_{r})), \\
    & \ap^{R}_{A,B} && \Colon{ } &&
    \forall \sigma\ \sigma_{e}\ \sigma_{r}\ \tau\ \tau_{e}\ \tau_{r}\ b \\
    &&&&& \quad \to \forall a\ a_{e}\ a_{r} \to \sem{B}_{R}(\tau_{r}(a_{r}), \ap^{E}_{A,B}(\dotsc,b,a_{e})).
    \tag*{\defiEnd{}}
  \end{alignat*}
\end{defi}
The representablility of $A$ ensures that there is a type classifying the premises of a congruence operation, because $A(\sigma) \to \partial B$, $(a : A(\sigma)) \to B(\tau(a))$, \etc, are types.
Thus a congruence operation is fully determined by its evaluation at a suitable generic context.

\begin{customasm}{(A2)}\label{asm:congruence_operations}
  For every global dependent representable type $A$ and generating type $\iS : \GenTy_{\Th}$, we have a congruence operation from $A \Colon \partial A \to \RepTy_{\Th}$ to $\iS \Colon \partial \iS \to \Ty_{\Th}$.
  \defiEnd{}.
\end{customasm}
In particular, the congruence operation from $\Unit \Colon \Unit \to \RepTy_{\Th}$ to $\iS$ is exactly a reflexivity operation for $\iS$.
When $\Th$ is a first-order GAT, \cref{asm:congruence_operations} reduces to the existence of reflexivity operations for every generating type $\iS$.

Note that~\cref{asm:congruence_operations} is not ideal: it refers to an arbitrary representable type $A$, which can be seen as a telescope of basic representable types.
Thus it may require data for every possible telescope shape.
We would prefer to quantify over the generating representable types instead.

\begin{lem}\label{lem:refl_monomial_type}
  Every global dependent monomial type $A \Colon \partial A \to \MonoTy_{\Th}$ can be equipped with a reflexivity operation.
\end{lem}
\begin{proof}
  Since $A$ is a monomial type, we can write
  \[ A(\sigma) = (\delta : \Delta(\sigma)) \to \iT(\tau(\sigma,\delta)) \]
  for some $\Delta \Colon \partial A \to \RepTy_{\Th}$ and $\tau \Colon \forall \sigma \to \Delta(\sigma) \to \partial \iT$.

  By definition of the section $\sem{-}$, we know that:
  \begin{alignat*}{1}
    & \sem{A}_{E}(\sigma_{l},\sigma_{r},\sigma_{e},f_{l},f_{r}) = \forall \delta_{l}\ \delta_{r}\ \delta_{e} \to \sem{\iT}_{E}(\sem{\tau}_{E}(\sigma_{e},\delta_{e}),f_{l}(\delta_{l}),f_{r}(\delta_{r})), \\
    & \sem{A}_{R}(\sigma,\sigma_{e},\sigma_{r},f,f_{e}) = \forall \delta\ \delta_{e}\ \delta_{r} \to \sem{\iT}_{R}(\sem{\tau}_{R}(\sigma_{r},\delta_{r}),f_{e}(\delta_{e})).
  \end{alignat*}

  We first pose:
  \begin{alignat*}{3}
    & \tau_{e} && :{ } && \forall \sigma\ \sigma_{e}\ \delta_{l}\ \delta_{r}\ \delta_{e} \to \sem{\partial\iT}_{E}(\tau(\sigma_{l},\delta_{l})), \\
    & \tau_{e}(\dots) && \triangleq{ } && \sem{\tau}_{E}((\sigma,\delta_{l}),(\sigma,\delta_{r}),(\sigma_{e},\delta_{e})), \\
    & \tau_{r} && :{ } && \forall \sigma\ \sigma_{e}\ \sigma_{r}\ \delta\ \delta_{e}\ \delta_{r} \to \sem{\partial\iT}_{R}(\tau(\sigma,\delta),\tau_{e}(\sigma_{e},\delta_{e})), \\
    & \tau_{r}(\dots) && \triangleq{ } && \sem{\tau}_{R}((\sigma,\delta),(\sigma_{e},\delta_{e}),(\sigma_{r},\delta_{r})).
  \end{alignat*}

  By~\cref{asm:congruence_operations}, we have a congruence operation $\ap_{\Delta,\iT}$ from $\Delta$ to $\iT$.

  We can finally define the reflexivity operation for $A$:
  \begin{alignat*}{3}
    & \refl^{E}_{A}(\sigma,\sigma_{e},\sigma_{r},f) && \triangleq{ } &&
    \lambda \delta_{l}\ \delta_{r}\ \delta_{e} \mapsto \ap^{E}_{\Delta,\iT}(\sigma_{r},\tau(\sigma),\tau_{e}(\sigma_{e}),\tau_{r}(\sigma_{r}),f,\delta_{e}), \\
    & \refl^{R}_{A}(\sigma,\sigma_{e},\sigma_{r},f) && \triangleq{ } &&
    \lambda \delta\ \delta_{e}\ \delta_{r} \mapsto \ap^{R}_{\Delta,\iT}(\sigma_{r},\tau(\sigma),\tau_{e}(\sigma_{e}),\tau_{r}(\sigma_{r}),f,\delta_{r}).
    \tag*{\qedhere{}}
  \end{alignat*}
\end{proof}

\begin{lem}\label{lem:refl_polynomial_type}
  Every global dependent polynomial type $A \Colon \partial A \to \PolyTy_{\Th}$ can be equipped with a reflexivity operation.
\end{lem}
\begin{proof}
  We know that $A$ is a telescope of monomial types and we can perform induction on its length.
  The case of the empty telescope is trivial.

  If $A$ is a non-empty telescope, we have
  \[ A(\sigma) = (b : B(\sigma)) \times (c : C(\sigma,b)) \]
  for some $B \Colon \partial A \to \PolyTy_{\Th}$ and $C \Colon \forall \sigma \to B(\sigma) \to \MonoTy_{\Th}$.

  By the induction hypothesis, we have a reflexivity operation for $B$.
  By~\cref{lem:refl_monomial_type}, we have a reflexivity operation for $C$.

  We can thus define:
  \begin{alignat*}{3}
    & \refl^{E}_{A}(\sigma,\sigma_{e},\sigma_{r},(b,c)) && \triangleq{ } &&
    (\refl^{E}_{B}(\sigma,\sigma_{e},\sigma_{r},b), \refl^{E}_{C}((\sigma,b),(\sigma_{e},\refl^{E}_{B}(\dots,b)),(\sigma_{r},\refl^{R}_{B}(\dots,b)),c)), \\
    & \refl^{R}_{A}(\sigma,\sigma_{e},\sigma_{r},(b,c)) && \triangleq{ } &&
    (\refl^{R}_{B}(\sigma,\sigma_{e},\sigma_{r},b), \refl^{R}_{C}((\sigma,b),(\sigma_{e},\refl^{E}_{B}(\dots,b)),(\sigma_{r},\refl^{R}_{B}(\dots,b)),c)).
    \tag*{\qedhere{}}
  \end{alignat*}
\end{proof}

\begin{lem}\label{lem:refl_type}
  Every global dependent type $A \Colon \partial A \to \Ty_{\Th}$ can be equipped with a reflexivity operation.
\end{lem}
\begin{proof}
  By~\cref{prop:type_iso_polytype}, there is some polynomial type $A_{0} \Colon \partial A \to \PolyTy_{\Th}$ such that $\forall \sigma \to \Tm_{\Th}(A_0(\sigma)) \cong \Tm_{\Th}(A(\sigma))$.
  The results then follows from~\cref{lem:refl_polynomial_type}.
\end{proof}

Thus, by~\cref{prop:contextual_reflexivity_structure}, we have a reflexivity structure over $\Th$.

Now that we have a reflexivity structure, we obtain as in~\cref{ssec:parametricity_structures} reflexive equivalences and dependent equivalences over $\Th$.

\subsection{Centers and paths}

It remains to show that the contractible types have centers and all paths.

We assume that center and homogeneous all-paths operations are defined for the basic contractible types.
\begin{customasm}{(A3)}\label{asm:basic_center_and_all_paths}
  For every constructor $c \Colon (\gamma : \Gamma) \to \isBContr(A(\gamma))$ of $\isBContr$, we have the following data:
  \begin{alignat*}{3}
    & \ccenter_{c} && \Colon{ } && \forall (\gamma : \Gamma) \to A(\gamma), \\
    & \chpath_{c} && \Colon{ } && \forall \gamma\ (\gamma_{e} : \sem{\Gamma}_{E}(\gamma,\gamma))\ (\gamma_{r} : \sem{\Gamma}_{R}(\gamma,\gamma_e)) \\
    &&&&& \phantom{\forall} \to (x,y : A(\gamma))\ \to \sem{A}_{E}(\gamma_{e}, x,y).
  \end{alignat*}

  Note that every constructor of $\isBContr$ is of this form.
  \defiEnd{}
\end{customasm}
In~\cref{asm:basic_center_and_all_paths} we have been careful not to mention $\refl_{\Gamma}$; the assumption can be checked independently of the construction of the reflexivity maps, and independently of~\cref{asm:congruence_operations}.

\begin{lem}\label{lem:basic_center_and_all_paths}
  Every contractibility witness $c \Colon (\gamma : \Gamma) \to \isBContr(A(\gamma))$ admits a center operation and a homogeneous all-paths operation.
\end{lem}
\begin{proof}
  Fix a contractibility witness $c \Colon (\gamma : \Gamma) \to \isBContr(A(\gamma))$.
  By definition of $\isBContr$, there exists a constructor $d \Colon (\delta : \Delta) \to \isBContr(B(\delta))$ of $\isBContr$ such that $c = d[\delta]$ for some $\delta \Colon \Gamma \to \Delta$.
  In particular, $A = B[\delta]$.

  We can now pose
  \begin{alignat*}{3}
    & c.\ccenter(\gamma) && \triangleq{ } && \ccenter_{d}(\delta(\gamma)), \\
    & c.\chpath(\gamma,x,y) && \triangleq{ } && \chpath_{d}(\delta(\gamma), \sem{\delta}_{E}(\refl^{E}_{\Gamma}(\gamma)), \sem{\delta}_{R}(\refl^{R}_{\Gamma}(\gamma)), x,y).
    \tag*{\qedhere{}}
  \end{alignat*}
\end{proof}

We then proceed to extend this to arbitrary contractibility witnesses, relying on the constructions of~\cref{ssec:parametricity_structures}.

\begin{lem}\label{lem:repcontr_center}
  For every representable contractibility witness
  $ c \Colon (\gamma : \Gamma) \to \isRepContr(A(\gamma)), $
  there is a center operation for $c$.
\end{lem}
\begin{proof}
  By induction on $c$.
  \begin{description}
  \item[Constructor $\contr^{\rep}_{b}$] \hfill \\
    By~\cref{lem:basic_center_and_all_paths}.

  \item[Constructor $\contr^{\rep}_{\Unit}$] \hfill \\
    By~\cref{constr:center_and_all_paths_unit}.

  \item[Constructor $\contr^{\rep}_{\Sigma}$] \hfill \\
    By~\cref{constr:center_sigma}.
    \qedhere{}
  \end{description}
\end{proof}

\begin{lem}\label{lem:repcontr_heterogeneous_all_paths}
  Let $c \Colon (\gamma : \Gamma) \to A(\gamma) \to \isRepContr(B(\gamma,a))$ be a global dependent family of contractibility witnesses.
  If we have a homogeneous all-paths operation for $c$ over $\Gamma.A$, then we construct a heterogeneous all-paths operation for $c$.
\end{lem}
\begin{proof}
  By~\cref{constr:heterogeneous_all_paths}, applied to the family restriction $\RepTy_\Th \to \Ty_\Th$.
  Checking the assumption relies on~\cref{lem:repcontr_center}.
\end{proof}

\begin{lem}\label{lem:repcontr_homogeneous_all_paths}
  For every representable contractibility witness
  $ c \Colon (\gamma : \Gamma) \to \isRepContr(A(\gamma)), $
  there is a homogeneous all-paths operation for $c$.
\end{lem}
\begin{proof}
  By induction on $c$.
  \begin{description}
  \item[Constructor $\contr^{\rep}_{b}$] \hfill \\
    By~\cref{lem:basic_center_and_all_paths}.

  \item[Constructor $\contr^{\rep}_{\Unit}$] \hfill \\
    By~\cref{constr:center_and_all_paths_unit}.

  \item[Constructor $\contr^{\rep}_{\Sigma}$] \hfill \\
    By~\cref{constr:all_paths_sigma} and~\cref{lem:repcontr_heterogeneous_all_paths}.
    \qedhere{}
  \end{description}
\end{proof}

We have now defined center and all-paths operations for all contractibility witnesses for representable types.
We remark that this already equips the family of representable types with the structure of weakly stable identity types.
\begin{prop}\label{prop:repcontr_id_types}
  The family $\RepTy_\Th$ is equipped with weakly stable identity types, where the introduction structure is given by the reflexive equivalences constructed in~\cref{ssec:proof_contractibility_data}.
\end{prop}
\begin{proof}
  We use~\cref{thm:identity_types_from_refleqv}.
  We have defined contractibility data, reflexive equivalences and dependent equivalences in~\cref{ssec:proof_contractibility_data}, and centers and paths of contractible types in~\cref{lem:repcontr_center} and~\cref{lem:repcontr_homogeneous_all_paths}.
\end{proof}

\begin{lem}\label{lem:contr_center}
  For every contractibility witness
  $ c \Colon \isContr(A), $
  there is a center operation for $c$.
\end{lem}
\begin{proof}
  By induction on $c$.
  \begin{description}
  \item[Constructor $\contr_{b}$] \hfill \\
    By~\cref{lem:basic_center_and_all_paths}.

  \item[Constructor $\contr_{\Unit}$] \hfill \\
    By~\cref{constr:center_and_all_paths_unit}.

  \item[Constructor $\contr_{\Sigma}$] \hfill \\
    By~\cref{constr:all_paths_sigma}.

  \item[Constructor $\contr_{\Pi}$] \hfill \\
    By~\cref{constr:center_pi}.
    \qedhere{}
  \end{description}
\end{proof}

\begin{lem}\label{lem:contr_heterogeneous_all_paths}
  Let $c \Colon (\gamma : \Gamma) \to (a : A(\gamma)) \to \isContr(B(\gamma,a))$ be a global dependent family of contractibility witnesses.
  If we have a homogeneous all-paths operation for $c$ over $\Gamma.A$, then we construct a heterogeneous all-paths operation for $c$.
\end{lem}
\begin{proof}
  By~\cref{constr:heterogeneous_all_paths}, applied to the family restriction $\Ty_\Th \to \Ty_\Th$.
  Checking the assumption relies on~\cref{lem:contr_center}.
\end{proof}

\begin{lem}\label{lem:contr_homogeneous_all_paths}
  For every contractibility witness
  $ c \Colon (\gamma : \Gamma) \to \isContr(A(\gamma)), $
  there is a homogeneous all-paths operation for $c$.
\end{lem}
\begin{proof}
  By induction on $c$.
  \begin{description}
  \item[Constructor $\contr_{b}$] \hfill \\
    By~\cref{lem:basic_center_and_all_paths}.

  \item[Constructor $\contr_{\Unit}$] \hfill \\
    By~\cref{constr:center_and_all_paths_unit}.

  \item[Constructor $\contr_{\Sigma}$] \hfill \\
    By~\cref{constr:all_paths_sigma} and~\cref{lem:contr_heterogeneous_all_paths}.

  \item[Constructor $\contr_{\Pi}$] \hfill \\
    By~\cref{constr:all_paths_pi}, \cref{prop:repcontr_id_types} and \cref{lem:contr_heterogeneous_all_paths}.
    \qedhere{}
  \end{description}
\end{proof}

\subsection{Main theorem}

We record the results of this section in the following theorem.
\begin{thm}\label{thm:external_univalence_from_refleqv}
  Let $\Th$ be a SOGAT equipped with homotopy relations.
  If it satisfies~\cref{asm:refleqv_internal_model}, \cref{asm:congruence_operations} and~\cref{asm:basic_center_and_all_paths}, then it satisfies external univalence, \ie{} the $(\reppre)$-CwF $\Th$ can be equipped with weakly stable identity types satisfying function extensionality and saturation with respect to the homotopy relations.
\end{thm}
\begin{proof}
  The weakly stable identity types are constructed using~\cref{thm:identity_types_from_refleqv}.
  We have defined contractibility data, reflexive equivalences and dependent equivalences in~\cref{ssec:proof_contractibility_data}, and centers and paths of contractible types in~\cref{lem:contr_center} and~\cref{lem:contr_homogeneous_all_paths}.
  Saturation with respect to the homotopy relations follows from the constructors $\contr_{\iS,\sim}$ of $\isBContr$.
  Function extensionality follows from the definition of the $\Pi$-types in $\CPreReflGraph(\CC)$.
\end{proof}

\subsection{Necessary conditions}

The hypotheses of~\cref{thm:external_univalence_from_refleqv} are not necessary conditions in the general setting.
However, they become necessary for SOGATs without equations, and more generally for cofibrant SOGATs, \ie{} for SOGATs that are retracts (in $\CCwf_{\reppre}$) of SOGATs without equations.
\begin{thm}\label{thm:external_univalence_from_refleqv_converse}
  Let $\Th$ be a cofibrant SOGAT, \ie a $(\reppre)$-CwF that is in the left class of the weak factorization system generated by $\{I^\ty,I^{\repty},I^\tm\}$.

  Assume that $\Th$ is equipped with homotopy relations, such that for every generating representable sort $\iS : \GenRepTy_\Th$, the homotopy relation $(- \sim_{\iS(-)} -)$ is a family of representable types.

  If $\Th$ satisfies external univalence, then it satisfies the conditions of~\cref{thm:external_univalence_from_refleqv} (\cref{asm:refleqv_internal_model}, \cref{asm:congruence_operations} and~\cref{asm:basic_center_and_all_paths}).
\end{thm}
\begin{proof}
  The identity types $(\simeq)$ of $\Th$ induce contractibility data $\isContr_{\simeq}$ on $\Th$, as constructed in~\cref{constr:stable_contractibility_data}.

  The contractibility data $\isContr_{\simeq}$ satisfies the necessary closure conditions, so that the reflexive equivalences model $\CReflEqv(\Th)$ can be constructed by~\cref{constr:homotopical_mgraph_model_2}; we omit the details.

  We consider the contextual core $\CReflEqv^{\mathsf{cxl}}(\Th)$ of $\CReflEqv(\Th)$; its contexts can be obtained as iterated context extensions in $\CReflEqv(\Th)$.
  Since $\CReflEqv(\Th)$ has $\Sigma$-types, any context of $\CReflEqv^{\mathsf{cxl}}(\Th)$ is isomorphic to a closed type of $\CReflEqv(\Th)$.
  In particular, any context of $\CReflEqv^{\mathsf{cxl}}(\Th)$ satisfies the contractibility conditions of reflexive equivalences.

  We now show that the projection morphism $V : \CReflEqv^{\cxl}(\Th) \to \Th$ satisfies the right lifting property with respect to $\{I^\ty,I^{\repty},I^\tm\}$.
  \begin{description}
  \item[Case ${I^\ty : \Free_{\reppre}(\bm{\Gamma} \vdash) \to \Free_{\reppre}(\bm{\Gamma} \vdash \bm{A}\ \type)}$]
    Given an object $\Gamma : \CReflEqv^{\cxl}(\Th)$ and a type $A \Colon \yo(\Gamma_V) \to \Ty_\Th$, we have to extend $A$ to a dependent reflexive equivalence $(A,A_E,A_R)$ over $\Gamma$.

    Since $(\Gamma_E,\Gamma_R)$ satisfies the contractibility condition of a reflexive equivalence, it is equivalent to the reflexive equivalence $\Id_\Gamma$.
    Up to this equivalence, we can then compute $(A_E,A_R)$ as the dependent identity type $\DId_{\Gamma.A}$.

  \item[Case ${I^{\repty} : \Free_{\reppre}(\bm{\Gamma} \vdash) \to \Free_{\reppre}(\bm{\Gamma} \vdash \bm{A}\ \reptype)}$]
    Similar to the case of $I^\ty$.
    Given an object $\Gamma : \CReflEqv^{\cxl}(\Th)$ and a representable type $A \Colon \yo(\Gamma_V) \to \RepTy_\Th$, we have to extend $A$ to a representable dependent reflexive equivalence $(A,A_E,A_R)$ over $\Gamma$.
    We can compute $(A_E,A_R)$ as the dependent identity type $\DId_{\Gamma.A}$.
    The assumption that the homotopy relation $(- \sim_{\iS(-)} -)$ is a family of representable types whenever $\iS$ is representable ensures that we can choose representable families for $A_E$ and $A_R$.

  \item[Case ${I^{\tm} : \Free_{\reppre}(\bm{\Gamma} \vdash \bm{A}\ \type) \to \Free_{\reppre}(\bm{\Gamma} \vdash \bm{a} : \bm{A})}$]
    Given an object $\Gamma : \CReflEqv^{\cxl}(\Th)$, a dependent reflexive equivalence $A$ over $\Gamma$ and a term $a \Colon (\gamma : \yo(\Gamma_V)) \to \Tm_\Th(A_V(\gamma))$, we have to show that $a$ can be extended to a term of $\CReflEqv(\Th)$, that is we have to construct:
    \begin{alignat*}{3}
      & a_E && :{ }
      && \forall \gamma_l\ \gamma_r\ (\gamma_e : \Tm_\Th(\Gamma_E(\gamma_l,\gamma_r))) \to \Tm_\Th(A_E(\gamma_e,a(\gamma_l),a(\gamma_e))), \\
      & a_R && :{ }
      && \forall \gamma\ \gamma_e\ (\gamma_r : \Tm_\Th(\Gamma_R(\gamma,\gamma_e))) \to \Tm_\Th(A_R(\gamma_r,a(\gamma_l),a_E(\gamma_e))).
    \end{alignat*}

    Up to the identification of $(\Gamma_E,\Gamma_R)$ and $(A_E,A_R)$ with respectively $\Id_\Gamma$ and $\DId_{\Gamma.A}$, $a_E$ and $a_R$ are just given by the dependent action of $a$ on paths $\mathsf{apd}_a$.
  \end{description}
  Note that we could not show the lifting property with respect to $E^\tm$: there could be multiple ways to lift a same term from $\Th$ to $\CReflEqv^{\cxl}(\Th)$, \eg{} given by multiple definitions of the dependent action on paths.

  Since $\Th$ is $\{I^\ty,I^{\repty},I^\tm\}$-cellular, we obtain a section $\sem{-} : \Th \to \CReflEqv^{\cxl}(\Th)$ of $V : \CReflEqv^{\cxl}(\Th) \to \Th$.
  By composing this section with the $(\reppre)$-CwF morphism $\CReflEqv^{\cxl}(\Th) \to \CPreReflGraph(\Th)$, we equip $\CPreReflGraph(\Th)$ with an internal model of $\Th$, satisfying~\cref{asm:refleqv_internal_model}.

  It remains to check~\cref{asm:congruence_operations} and~\cref{asm:basic_center_and_all_paths}.

  The conditions of~\cref{asm:congruence_operations} are instances of the dependent action on paths.

  Finally, \cref{asm:basic_center_and_all_paths} is proven by showing that induction on the inductive families $\isRepContr$ and $\isContr$ that for every $A \Colon \yo(\Gamma) \to \Ty_\Th$, we have
  \[ \forall \gamma \to \isContr(A(\gamma)) \to \isContr_{\simeq}(A(\gamma)), \]
  since we already know that the types that are contractible with respect to $\isContr_{\simeq}$ have centers and all-paths.
  This amount to showing some closure conditions on $\isContr_{\simeq}$ that were already proven in the construction of $\CReflEqv(\Th)$.
\end{proof}


%% file: applications.tex
\section{Applications}\label{sec:applications}

In this section, we apply~\cref{thm:external_univalence_from_refleqv} to prove that some SOGATs satisfy external univalence.

\subsection{Categories}\label{ssec:application_cat}

We first consider the first-order generalized algebraic theory $\Th_{\CCat}$ of categories, equipped with the homotopy relations defined in~\cref{exa:homotopy_relations_th_cat}.
We apply~\cref{thm:external_univalence_from_refleqv} to show that it satisfies external univalence.

These constructions have been formalized in Agda~%
\footnote{The Agda files are available at~\url{https://rafaelbocquet.gitlab.io/Agda/20221114_ExternalUnivalence/Cat.html}.}%
, showing that for some concrete theory, \cref{asm:refleqv_internal_model}, \cref{asm:congruence_operations}, \cref{asm:basic_center_and_all_paths} are syntactic enough as to be checked in a proof assistant.

We first need to check~\cref{asm:refleqv_internal_model}, that is to equip $\CPreReflGraph(\Th_{\CCat})$ with an internal model of $\Th_{\CCat}$.

The sorts of objects, morphisms and equalities between morphisms are interpreted as follows:
\begingroup{}\allowdisplaybreaks{}
\begin{alignat*}{3}
  & \sem{\iob}_{E}(\gamma_e) && \triangleq{ } && \lambda x_{l}\ x_{r} \mapsto x_{l} \cong x_{r}, \\
  & \sem{\ihom(x,y)}_{E}(\gamma_e) && \triangleq{ } && \lambda f_{l}\ f_{r} \mapsto \iEqHom(\sem{y}_E(\gamma_e) \icirc f_{l}, f_{r} \icirc \sem{x}_E(\gamma_e)), \\
  & \sem{\ieqhom(f,g)}_{E}(\gamma_e) && \triangleq{ } && \lambda p_{l}\ p_{r} \mapsto \Unit, \\
  & \sem{\iob}_{R}(\gamma_r) && \triangleq{ } && \lambda x\ x_{e} \mapsto \iEqHom(x_{e},\iid(x)), \\
  & \sem{\ihom(x,y)}_{R}(\gamma_r) && \triangleq{ } && \lambda f\ f_{e} \mapsto \Unit, \\
  & \sem{\ieqhom(x,y)}_{R}(\gamma_r) && \triangleq{ } && \lambda p\ p_{e} \mapsto \Unit.
\end{alignat*}\endgroup{}

We then have to interpret the category operations, in a way that satisfies the category laws.
We have to define the following components:
\begin{alignat*}{3}
  & \sem{\iid(x)}_{E}(\gamma_e) && :{ } && \iEqHom(\sem{x}_E(\gamma_e) \icirc \iid(x(\gamma_l)), \iid(x(\gamma_r)) \icirc \sem{x}_{E}(\gamma_e)), \\
  & \sem{\icomp(f,g)}_{E}(\gamma_e) && :{ } && \iEqHom(\sem{z}_{E}(\gamma_e) \icirc \icomp(f(\gamma_l),g(\gamma_l)), \icomp(f(\gamma_r),g(\gamma_r)) \icirc \sem{x}_{E}(\gamma_e)).
\end{alignat*}
The components $\sem{\iid(x)}_{R}$ and $\sem{\icomp(f,g)}_{R}$ and are trivial, since $\sem{\ihom(-)}_{R}(-)$ is the unit type.

The component $\sem{\iid(x)}_{E}$ follows directly from the category laws.

For the remaining component $\sem{\icomp(f,g)}_{E}$, remember that we have assumptions
\begin{alignat*}{3}
  & \sem{f}_{E}(\gamma_e) && :{ } && \iEqHom(\sem{y}_{E}(\gamma_e) \icirc f(\gamma_l), f(\gamma_r) \icirc \sem{x}_{E}(\gamma_e)), \text{and} \\
  & \sem{g}_{E}(\gamma_e) && :{ } && \iEqHom(\sem{z}_{E}(\gamma_e) \icirc g(\gamma_l), g(\gamma_r) \icirc \sem{y}_{E}(\gamma_e)).
\end{alignat*}
We can then apply the category laws to derive $\sem{\icomp(f,g)}_{E}(\gamma_e)$.

Because the components ${-}_{E}$ are propositional, the category laws are trivially satisfied in $\CPreReflGraph(\Th_{\CCat})$.

This finishes the definition of an internal model of $\Th_{\CCat}$ in $\CPreReflGraph(\Th_{\CCat})$.
We then have to check~\cref{asm:congruence_operations}.
Since $\Th_{\CCat}$ is a first-order GAT, we just have to define reflexivity operations for every generating type.
The reflexivity operation of objects is given by $\iid$, the reflexivity operations for morphisms is given by $\irefl$ and the reflexivity operation for equalities between morphisms is trivial.

It remains to check~\cref{asm:basic_center_and_all_paths}.
We first unfold the definition of $\isContr^{b}$; it has the following constructors:
\begingroup{}\allowdisplaybreaks{}
\begin{alignat*}{3}
  & \contr_{\iOb,l} && :{ } &&
  (x : \iOb) \to \isContr^{b}((y : \iOb) \times (x \cong y)), \\
  & \contr_{\iOb,r} && :{ } &&
  (y : \iOb) \to \isContr^{b}((x : \iOb) \times (x \cong y)), \\
  & \contr_{\iOb,{\sim}} && :{ } &&
  (x : \iOb) \to \isContr^{b}((y : \iOb) \times (x \cong y)), \\
  & \contr_{\iHom,l} && :{ } &&
  \forall x_{l}\ x_{r}\ x_{e}\ y_{l}\ y_{r}\ y_{e}\ (f : \iHom(x_{l},y_{l})) \\
  &&&&& \quad \to \isContr^{b}((g : \iHom(x_{r},y_{r})) \times \iEqHom(y_{e} \icirc f \icirc x_{e}^{-1}, g)), \\
  & \contr_{\iHom,r} && :{ } &&
  \forall x_{l}\ x_{r}\ x_{e}\ y_{l}\ y_{r}\ y_{e}\ (f : \iHom(x_{r},y_{r})) \\
  &&&&& \quad \to \isContr^{b}((g : \iHom(x_{l},y_{l})) \times \iEqHom(y_{e} \icirc g \icirc x_{e}^{-1}, f)), \\
  & \contr_{\iHom,{\sim}} && :{ } &&
  \forall x\ y\ (f : \iHom(x,y)) \\
  &&&&& \quad \to \isContr^{b}((g : \iHom(x,y)) \times \iEqHom(f, g)), \\
  & \contr_{\iEqHom,l} && :{ } &&
  \forall x_{l}\ x_{r}\ x_{e}\ y_{l}\ y_{r}\ y_{e}\ f_{l}\ f_{r}\ f_{e}\ g_{l}\ g_{r}\ g_{e}\ (p : \iEqHom(f_{l},g_{l})) \\
  &&&&& \quad \to \isContr^{b}(\iEqHom(f_{r},g_{r}) \times \Unit), \\
  & \contr_{\iEqHom,r} && :{ } &&
  \forall x_{l}\ x_{r}\ x_{e}\ y_{l}\ y_{r}\ y_{e}\ f_{l}\ f_{r}\ f_{e}\ g_{l}\ g_{r}\ g_{e}\ (p : \iEqHom(f_{r},g_{r})) \\
  &&&&& \quad \to \isContr^{b}(\iEqHom(f_{l},g_{l}) \times \Unit), \\
  & \contr_{\iEqHom,{\sim}} && :{ } &&
  \forall x\ y\ f\ g\ (p : \iEqHom(f,g)) \\
  &&&&& \quad \to \isContr^{b}(\iEqHom(f,g) \times \Unit).
\end{alignat*}\endgroup{}
However only $\contr_{\iOb,l}$, $\contr_{\iOb,r}$ and $\contr_{\iOb,{\sim}}$ are interesting.
In all of the other constructors, the type is isomorphic to $\Unit$, and verifying the conditions of~\cref{asm:basic_center_and_all_paths} is then trivial.
Furthermore, the constructors $\contr_{\iOb,l}$, $\contr_{\iOb,r}$ and $\contr_{\iOb,{\sim}}$ are equivalent, and it suffices to consider $\contr_{\iOb,l}$.

We have to construct the following terms:
\begin{alignat*}{3}
  & \ccenter_{\iOb,l} && :{ } &&
  \forall (x : \iOb) \to (y : \iOb) \times (f : x \cong y), \\
  & \chpath_{\iOb,l} && :{ } &&
  \forall x\ (x_{e} : x \cong x) \\
  &&&&& \phantom{\forall} (y_{l} : \iOb)\ (f_{l} : x_{l} \cong y_{l})\ (y_{r} : \iOb)\ (f_{r} : x_{r} \cong y_{r}) \\
  &&&&& \quad\to (y_{e} : y_{l} \cong y_{r}) \\
  &&&&& \quad\phantom{\to}\times \iEqHom(y_{e} \icirc f_{l} \icirc x_{e}^{-1}, f_{r}) \\
  &&&&& \quad\phantom{\to}\times \iEqHom(x_{e} \icirc f_{l}^{-1} \icirc y_{e}^{-1}, f_{r}^{-1}).
\end{alignat*}

The input data for $\chpath_{\iOb,l}$ can be described in the following diagram:
\[ \begin{tikzcd}
    x
    \ar[r, "{x_{e}}", "{\cong}"']
    \ar[d, "{\cong}", "{f_{l}}"']
    &
    x
    \ar[d, "{\cong}"', "{f_{r}}"]
    \\
    y_{l}
    &
    y_{r}\rlap{\ .}
  \end{tikzcd} \]
Then $\chpath_{\iOb,l}$ should be a filler of that open square.

They can be constructed as follows:
\begin{alignat*}{3}
  & \ccenter_{\iOb,l}(x) && \triangleq{ } && (x,\iid(x)), \\
  & \chpath_{\iOb,l}(x,x_{e},y_{l},f_{l},y_{r},f_{r}) && \triangleq{ } && (f_{r} \icirc x_{e} \icirc f_{l}^{-1}, \irefl, \irefl).
\end{alignat*}

Thus we have checked~\cref{asm:basic_center_and_all_paths}.
By~\cref{thm:external_univalence_from_refleqv}, the theory $\Th_{\CCat}$ satisfies external univalence.

\subsection{Type theory with identity types}\label{ssec:application_type_theory}

We now show external univalence for the SOGAT $\Th_{\Id}$ of a representable family equipped with weak identity types, equipped with the homotopy relations of~\cref{exa:homotopy_relations_th_id}.

We first equip the inverse diagram model $\CPreReflGraph(\Th_{\Id})$ with an internal model of $\Th_{\Id}$.

The sorts of types and terms are interpreted as follows:
\begin{alignat*}{3}
  & \sem{\ity}_{E}(\gamma_e) && \triangleq{ } && \lambda A\ B \mapsto \iEquiv(A,B), \\
  & \sem{\itm(A)}_{E}(\gamma_e) && \triangleq{ } && \lambda x\ y \mapsto \iTm(\sem{A}_{E}(\gamma_e,x,y)), \\
  & \sem{\ity}_{R}(\gamma_r) && \triangleq{ } && \lambda A\ E \mapsto \isRefl(E), \\
  & \sem{\itm(A)}_{R}(\gamma_r) && \triangleq{ } && \lambda x\ p \mapsto \iTm(\sem{A}_{R}(\gamma_r,x,p)),
\end{alignat*}
where $\iEquiv(A,B)$ is the sort of relational equivalences between $A$ and $B$, and
\begin{alignat*}{3}
  & \isRefl(E) && \triangleq{ } && (P : \forall a \to \iTm(E(a,a)) \to \iTy) \\
  &&&&& \times (\forall a \to \isContr((p : E(a,a)) \times P(p)))
\end{alignat*}
is the sort of reflexivity structures over an equivalence $E : \iEquiv(A,A)$.

We also have to interpret the operations $\iId$, $\irefl$, $\iJ$ and $\iJb$ in this model.
For the operations $\iId$ and $\irefl$, we have to define the following components:
\begingroup{}\allowdisplaybreaks{}
\begin{alignat*}{3}
  & \sem{\iId(A,x,y)}_{E}(\gamma_e) && :{ } && \iEquiv(\iId(A(\gamma_l),x(\gamma_l),y(\gamma_l)), \iId(A(\gamma_r),x(\gamma_r),y(\gamma_r))), \\
  & \sem{\iId(A,x,y)}_{R}(\gamma_r) && :{ } && \isRefl(\sem{\iId(A,x,y)}_{E}(\gamma_e)), \\
  & \sem{\irefl(A,x)}_{E}(\gamma_e) && :{ } && \iTm(\sem{\iId(A,x,x)}_E(\gamma_e,\irefl(A(\gamma_l),x(\gamma_l)),\irefl(A(\gamma_r),x(\gamma_r)))), \\
  & \sem{\irefl(A,x)}_{R}(\gamma_r) && :{ } && \iTm(\sem{\iId(A,x,x)}_{R}(\gamma_r,\irefl(A(\gamma),x(\gamma)),\sem{\irefl(A,x)}_E(\gamma_e))).
\end{alignat*}\endgroup{}
In other words, we have to show that $\iId$ preserves relational equivalences in a way that preserves reflexive equivalences, and that $\irefl$ preserves elements of these relations.
We omit this standard proof.

This completes the definition of the internal model of $\Th_{\Id}$ in $\CPreReflGraph(\Th_\Id)$.
We now have to check~\cref{asm:congruence_operations}.

Let $A \Colon \partial A \to \RepTy_{\Th_{\Id}}$ be a global dependent representable type.
We have to construct congruence operations from $A$ to $\ity$ and $\itm$:
\begingroup{}\allowdisplaybreaks{}
\begin{alignat*}{3}
  & \ap^{E}_{A,\ity} && \Colon{ }
  && \forall (\sigma : \partial A) (\sigma_{e} : \sem{\partial A}_{E}(\sigma,\sigma)) (\sigma_{r} : \sem{\partial A}_{R}(\sigma,\sigma_{e})) \\
  &&&&& \phantom{\forall} (B : A(\sigma) \to \iTy) \\
  &&&&& \to \forall a_l\ a_r\ (a_e : \sem{A}_E(\sigma_e,a_l,a_r)) \to \iEquiv(B(a_l),B(a_r)), \\
  & \ap^{R}_{A,\ity} && \Colon{ }
  && \forall (\sigma : \partial A) (\sigma_{e} : \sem{\partial A}_{E}(\sigma,\sigma)) (\sigma_{r} : \sem{\partial A}_{R}(\sigma,\sigma_{e})) \\
  &&&&& \phantom{\forall} (B : A(\sigma) \to \iTy) \\
  &&&&& \to \forall a\ a_e\ (a_r : \sem{A}_R(\sigma_r,a,a_e)) \to \isRefl(\ap^{E}_{A,\ity}(\sigma,\sigma_e,\sigma_r,B,a,a,a_e)), \\
  & \ap^{E}_{A,\itm} && \Colon{ }
  && \forall (\sigma : \partial A) (\sigma_{e} : \sem{\partial A}_{E}(\sigma,\sigma)) (\sigma_{r} : \sem{\partial A}_{R}(\sigma,\sigma_{e})) \\
  &&&&& \phantom{\forall} (B : A(\sigma) \to \ity) \\
  &&&&& \phantom{\forall} (B_{e} : \forall a_l\ a_r\ a_e \to \iEquiv(B(a_l),B(a_r))) \\
  &&&&& \phantom{\forall} (B_{r} : \forall a\ a_e\ a_r \to \isRefl(B_e(a,a,a_e))) \\
  &&&&& \phantom{\forall} (b : (a : A(\sigma)) \to \iTm(B(a))) \\
  &&&&& \to \forall a_l\ a_r\ (a_e : \sem{A}_E(\sigma_e,a_l,a_r)) \to B_e(a_l,a_r,a_e,b(a_l),b(a_r)), \\
  & \ap^{R}_{A,\itm} && \Colon{ }
  && \forall (\sigma : \partial A) (\sigma_{e} : \sem{\partial A}_{E}(\sigma,\sigma)) (\sigma_{r} : \sem{\partial A}_{R}(\sigma,\sigma_{e})) \\
  &&&&& \phantom{\forall} (B : A(\sigma) \to \ity) \\
  &&&&& \phantom{\forall} (B_{e} : \forall a_l\ a_r\ a_e \to \iEquiv(B(a_l),B(a_r))) \\
  &&&&& \phantom{\forall} (B_{r} : \forall a\ a_e\ a_r \to \isRefl(B_e(a,a,a_e))) \\
  &&&&& \phantom{\forall} (b : (a : A(\sigma)) \to \iTm(B(a))) \\
  &&&&& \to \forall a\ a_e\ (a_r : \sem{A}_R(\sigma_r,a,a_e)) \to B_r(a,a_e,a_r,b(a),\ap^{E}_{A,\itm}(\sigma,\sigma_e,\sigma_r,B,B_e,B_r,b,a,a,a_e)).
\end{alignat*}\endgroup{}
By~\cref{prop:type_iso_polytype}, we can assume without loss of generality that $A$ is a telescope of basic representable sorts.
Since the only representable sort of $\Th_\Id$ is $\itm$, this means that we have a dependent telescope $\Delta \Colon \partial A \to \iTy^\star$, with $A(\sigma) = \iTm^\star(\Delta(\sigma))$.
We can then compute $\sem{A}_E$ and $\sem{A}_R$ and prove by induction on the telescope $\Delta$ that $(a_r : A(\sigma)) \times (a_e : \sem{A}_E(\sigma_e,a_l,a_r))$ and $(a_e : \sem{A}_E(\sigma_e,a,a)) \times (a_r : \sem{A}_R(\sigma_r,a,a_e))$ are contractible in $\iTy^{\star}$ (with respect to the inner identity types), which is sufficient to derive $\ap_{A,\ity}$ and $\ap_{A,\itm}$.

It remains to check~\cref{asm:basic_center_and_all_paths}.
We first unfold the definition of $\isContr^{b}$ for this theory; it is given by the following constructors.
\begingroup{}\allowdisplaybreaks{}
\begin{alignat*}{3}
  & \contr_{\iTy,l} && :{ } && (A : \iTy) \to \isContr^{b}((B : \iTy) \times \iEquiv(A,B)) \\
  & \contr_{\iTy,r} && :{ } && (B : \iTy) \to \isContr^{b}((A : \iTy) \times \iEquiv(A,B)) \\
  & \contr_{\iTy,{\sim}} && :{ } && (A : \iTy) \to \isContr^{b}((B : \iTy) \times \iEquiv(A,B)) \\
  & \contr_{\iTm,l} && :{ } && \forall A\ B\ (E : \iEquiv(A,B))\ (a : \iTm(A)) \\
  &&&&& \quad \to \isContr^{b}((b : \iTm(B)) \times \iTm(E(a,b))), \\
  & \contr_{\iTm,r} && :{ } && \forall A\ B\ (E : \iEquiv(A,B))\ (b : \iTm(B)) \\
  &&&&& \quad \to \isContr^{b}((a : \iTm(A)) \times \iTm(E(a,b))), \\
  & \contr_{\iTm,{\sim}} && :{ } && \forall A\ (x : \iTm(A)) \\
  &&&&& \quad \to \isContr^{b}((y : \iTm(A)) \times \iTm(\Id(A,x,y))).
\end{alignat*}\endgroup{}

We now check~\cref{asm:basic_center_and_all_paths} for every constructor.
\begin{description}
  \item[Constructor ${\contr_{\iTy,l}}$] \hfill \\
        We have to construct the following terms:
        \begin{alignat*}{3}
          & \ccenter_{\iTy,l} && :{ } &&
          \forall (A : \iTy) \to (B : \iTy) \times \iEquiv(A,B), \\
          & \chpath_{\iTy,l} && :{ } &&
          \forall A\ (A_{e} : \iEquiv(A,A)) \\
          &&&&& \phantom{\forall} (B_{l} : \iTy)\ (E_{l} : \iEquiv(A,B_{l}))\ (B_{r} : \iTy)\ (E_{r} : \iEquiv(A,B_{r})) \\
          &&&&& \quad\to \sem{A \mapsto (B : \iTy) \times (E : \iEquiv(A,B))}_E(\dots,(B_{l},E_{l}),(B_{r},E_{r})).
        \end{alignat*}

        The center is given by the identity equivalence on $A$.
        \begin{alignat*}{3}
          & \ccenter_{\iTy,l}(A) && \triangleq{ } && (A,\id).
        \end{alignat*}

        Defining $\chpath$ amounts to giving a composite and filler for the following open square of equivalences:
        \[ \begin{tikzcd}
            A
            \ar[r, "\simeq"']
            \ar[r, "A_{e}"]
            \ar[d, "\simeq"]
            \ar[d, "E_{l}"']
            &
            A
            \ar[d, "\simeq"']
            \ar[d, "E_{r}"]
            \\
            B_{l}
            &
            B_{r}\rlap{\ .}
          \end{tikzcd} \]

        We omit this construction.
  \item[Constructor ${\contr_{\iTy,r}}$] \hfill \\
        Up to symmetry, this constructor is equivalent to ${\contr_{\iTy,l}}$.

  \item[Constructor ${\contr_{\iTy,{\sim}}}$] \hfill \\
        This constructor is identical to ${\contr_{\iTy,l}}$.

  \item[Constructor ${\contr_{\iTm,l}}$] \hfill \\
        We have to construct the following:
        \begingroup{}\allowdisplaybreaks{}
        \begin{alignat*}{3}
          & \ccenter_{\iTm,l} && :{ }
          && \forall A\ B\ (E : \iEquiv(A,B))\ a \to (b : \iTm(B)) \times (p : \iTm(E(a,b))), \\
          & \chpath_{\iTm,l} && :{ }
          && \forall (A : \iTy)\ (A_{e} : \iEquiv(A,A))\ (A_{r} : \isRefl(A_{e})) \\
          &&&&& \phantom{\forall} (B : \iTy)\ (B_{e} : \iEquiv(B,B))\ (B_{r} : \isRefl(B_{e})) \\
          &&&&& \phantom{\forall} (E : \iEquiv(A,B))\ E_e\ E_r \\
          &&&&& \phantom{\forall} (a : \iTm(A))\ (a_{e} : A_{e}(a,a))\ (a_r : A_r(a,a_e)) \\
          &&&&& \phantom{\forall} b_{l}\ (p_{l} : \iTm(E(a,b_{l})))\ b_{r}\ (p_{r} : \iTm(E(a,b_{r}))) \\
          &&&&& \quad \to (b_{e} : \iTm(B_{e}(b_{l},b_{r}))) \\
          &&&&& \quad \times (p_{e} : \iTm(E_{e}(a,a,a_{e},b_{l},b_{r},b_{e},p_{l},p_{r}))).
        \end{alignat*}\endgroup{}

        The center is given by transporting $a$ along the equivalence $E$:
        \[ \ccenter_{\iTm,l}(A,B,E,a) \triangleq E.\funr(a).1. \]

        Defining $\chpath$ amounts to giving a composite and filler for the following open dependent square
         \[ \begin{tikzcd}
            a
            \ar[r, no head, "a_{e}"]
            \ar[d, no head, "p_{l}"']
            &
            a
            \ar[d, no head, "p_{r}"]
            \\
            b_{l}
            &
            b_{r}\rlap{\ ,}
          \end{tikzcd} \]
        over the following square of equivalences
        \[ \begin{tikzcd}
            A
            \ar[r, "\simeq"']
            \ar[r, "A_{e}"]
            \ar[d, "\simeq"]
            \ar[d, "E"']
            \ar[rd, phantom, "(E_{e})"]
            &
            A
            \ar[d, "\simeq"']
            \ar[d, "E"]
            \\
            B
            \ar[r, "\simeq"]
            \ar[r, "B_{e}"']
            &
            B\rlap{\ .}
          \end{tikzcd} \]

        We also omit this construction.
  \item[Constructor ${\contr_{\iTm,r}}$] \hfill \\
        Up to symmetry, this constructor is equivalent to ${\contr_{\iTm,l}}$.

  \item[Constructor ${\contr_{\iTm,{\sim}}}$] \hfill \\
        This is a version of ${\contr_{\iTm,l}}$ that is specialized to the identity equivalence.
\end{description}

By~\cref{thm:external_univalence_from_refleqv}, the theory $\Th_{\Id}$ satisfies external univalence.

\subsection{Other type formers}
If we extend $\Th_{\Id}$ with additional type formers, the proofs of~\cref{asm:congruence_operations} and~\cref{asm:basic_center_and_all_paths} remain the same.
Only~\cref{asm:refleqv_internal_model} needs to be checked; that is we have to give an interpretation of the additional operations in $\CPreReflGraph(\Th)$.

Giving an interpretation of type and term formers in $\CPreReflGraph(\Th)$ amounts to a more or less standard parametricity translation of the type and terms formers, similar to other parametricity translations found in the literature~\parencite{ProofForFreeParametricityForDependentTypes,MarriageUnivalenceParametricity}.
As in the case of identity types, we have to show that the type and term formers preserve equivalences and identification in a way that preserves identity equivalences and reflexivity identifications.
Furthermore, these constructions have to be compatible with the various definitional equalities, such as the $\beta$- and $\eta$- rules for $\Sigma$- and $\Pi$- types.

For example, in the case of $\Unit$-, $\Sigma$- and $\Pi$- types, this translation mirrors at the inner level the outer level constructions of~\cref{sec:mgraph_model}:
\begingroup{}\allowdisplaybreaks{}
\begin{alignat*}{3}
  & \sem{\iUnit}_E(\gamma_e) && \triangleq{ }
  && \lambda t_l\ t_r \mapsto \iUnit, \\
  & \sem{\iUnit}_R(\gamma_r) && \triangleq{ }
  && \lambda t\ t_e \mapsto \iUnit, \\
  & \sem{\iSigma(A,B)}_E(\gamma_e) && \triangleq{ }
  && \lambda p_l\ p_r \mapsto (a_e : \sem{A}_E(\gamma_e,\ipi_1(p_l),\ipi_1(p_r))) \times (b_e : \sem{B}_E((\gamma_e,a_e),\ipi_2(p_l),\ipi_2(p_r))), \\
  & \sem{\iSigma(A,B)}_R(\gamma_r) && \triangleq{ }
  && \lambda p\ p_e \mapsto (a_r : \sem{A}_R(\gamma_r,\ipi_1(p),\ipi_1(p_e))) \times (b_r : \sem{B}_R((\gamma_r,a_r),\ipi_2(p),\ipi_2(p_e))), \\
  & \sem{\iPi(A,B)}_E(\gamma_e) && \triangleq{ }
  && \lambda f_l\ f_r \mapsto (\forall a_l\ a_r\ (a_e : \sem{A}_E(\gamma_e,a_l,a_r)) \to \sem{B}_E((\gamma_e,a_e),\iapp(f_l,a_l),\iapp(f_r,a_r)), \\
  & \sem{\iPi(A,B)}_R(\gamma_r) && \triangleq{ }
  && \lambda f\ f_e \mapsto (\forall a\ a_e\ (a_r : \sem{A}_R(\gamma_r,a,a_e)) \to \sem{B}_R((\gamma_r,a_r),\iapp(f,a),f_e(a_e)).
\end{alignat*}\endgroup{}

An empty type $\innerSym{\Init}$ can be interpreted in the presence of a $\Unit$-type.
\begin{alignat*}{3}
  & \sem{\innerSym{\Init}}_E(\gamma_e) && \triangleq{ }
  && \lambda z_l\ z_r \mapsto \iUnit, \\
  & \sem{\innerSym{\Init}}_R(\gamma_r) && \triangleq{ }
  && \lambda z\ z_e \mapsto \iUnit.
\end{alignat*}

Inductive types such as booleans or natural numbers can be interpreted provided that they support large elimination (or alternatively that sufficiently many indexed inductive types exist).
\begin{alignat*}{3}
  & \sem{\inner{\BoolTy}}_E(\gamma_e)(\inner{true},\inner{true}) && \triangleq{ }
  && \iUnit, \\
  & \sem{\inner{\BoolTy}}_E(\gamma_e)(\inner{true},\inner{false}) && \triangleq{ }
  && \innerSym{\Init}, \\
  & \sem{\inner{\BoolTy}}_E(\gamma_e)(\inner{false},\inner{true}) && \triangleq{ }
  && \innerSym{\Init}, \\
  & \sem{\inner{\BoolTy}}_E(\gamma_e)(\inner{false},\inner{false}) && \triangleq{ }
  && \iUnit.
\end{alignat*}

If $a \Colon \Tm_\Th(A)$ is any axiom of the theory, that is a generating element of a closed sort of $\Th$, then its interpretation in $\CPreReflGraph(\Th)$ can be derived automatically from the reflexivity structure of $\Th$:
\begin{alignat*}{1}
  & \sem{a}_{E}(\star) \triangleq
  \refl_{A}^{E}(a), \\
  & \sem{a}_{R}(\star) \triangleq
  \refl_{A}^{R}(a),
\end{alignat*}
provided that $\refl_{A}^{E}$ and $\refl_{A}^{R}$ are well-defined (to be more precise we should consider extensions of the theory $\Th$ by additional axioms).

The Uniqueness of Identity Proofs principle
\[ \inner{uip} : \forall (A : \iTy) \to \isSet(A), \]
cannot be seen as an axiom in the previous sense, because $\forall (A : \iTy) \to \isSet(A)$ is not a sort.
However it can still be interpreted in the pre-reflexive graphs model.
It suffices to define
\begin{alignat*}{3}
  & \sem{\inner{uip}(A)}_E(\gamma_e) && :{ }
  && \iEquiv(\isSet(A(\gamma_l)), \isSet(A(\gamma_r))), \\
  & \sem{\inner{uip}(A)}_R(\gamma_r) && :{ }
  && \isRefl(\sem{\inner{uip}(A)}_E(\gamma_e)).
\end{alignat*}
When defining $\sem{\inner{uip}(A)}_E(\gamma_e)$, we have an equivalence $\sem{A}_E(\gamma_e)$ between $A(\gamma_l)$ and $A(\gamma_r)$, from which we can derive an equivalence between $\isSet(A(\gamma_l))$ and $\isSet(A(\gamma_r))$.
Defining $\sem{\inner{uip}(A)}_R(\gamma_r)$ is straightforward using the fact that $\isSet(-)$ is a propositional type.

\begin{thm}
  The SOGATs of type theories with weak identity types and any selection of type structures among:
  \begin{itemize}
  \item a $\Unit$-type;
  \item $\Sigma$-types;
  \item $\Pi$-types;
  \item an empty type $\Init$, in the presence of $\Unit$;
  \item a boolean type with large elimination, in the presence of $\Unit$ and $\Init$;
  \item a natural number type with large elimination, in the presence of $\Unit$ and $\Init$;
  \item either weak or strict computation rules for any of the above type structures;
  \item any number of axioms, \ie{} generating elements of closed representable sorts;
  \item the Uniqueness of Identity Proofs principle;
  \end{itemize}
  all satisfy external univalence.
  \qed{}
\end{thm}

\subsection{Type theories with universes}

Finally, we discuss the situation of universes.
Tabareau \etal{} have achieved results that are similar to ours, but their work seemingly require the univalence axiom~\parencite[§6.4]{MarriageUnivalenceParametricity}
We claim that by using universes à la Russell, they implicitly assume the existence of a coding function, that turns types into elements of the universe.
By considering universes à la Tarski and without coding functions, there is no obstacle to the proof of external univalence, even in the presence of non-univalent universes.

We consider two SOGATs $\Th_{1}$ and $\Th_{2}$ that correspond to type theories with either a hierarchy of universes à la Tarski or a hierarchy of universes à la Coquand.
\begin{defi}
  The SOGAT $\Th_{1}$ is presented by the following signature
  \begingroup{}\allowdisplaybreaks{}
  \begin{alignat*}{3}
    & \ity && :{ } && (i : \Nat) \to \Ty, \\
    & \iTy_i && \triangleq{ } && \Tm(\ity_i), \\
    & \itm && :{ } && (i : \Nat) \to \iTy \to \RepTy, \\
    & \iTm_i(A) && \triangleq{ } && \Tm(\itm_i(A)), \\
    & \inner{\bm{U}} && :{ } && (i : \Nat) \to \iTy_{i+1}, \\
    & \inner{El} && :{ } && (i : \Nat) \to \iTm(\inner{\bm{U}}_i) \to \iTy_i, \\
    & \inner{\bm{u}} && :{ } && (i : \Nat) \to \iTm_{i+2}(\inner{\bm{U}}_{i+1}), \\
    & - && :{ } && \inner{El}(\inner{\bm{u}}_i) = \inner{\bm{U}}_i, \\
    & \inner{Lift} && :{ } && (i : \Nat) \to \iTy_i \to \iTy_{i+1}, \\
    & \inner{lift} && :{ } && (i : \Nat) (A  : \iTy_i) \to \iTm(A) \cong \iTm(\inner{Lift}(A)), \\
    & \iId && :{ } && (i : \Nat) (A : \iTy_i) (x,y : \iTm(A)) \to \iTy_i, \\
    & \irefl && :{ } && (i : \Nat) (A : \iTy_i) (x : \iTm(A)) \to \iTm_i(\iId(A,x,x)), \\
    & \iJ && :{ } && \dots, \\
    & \iJb && :{ } && \dots
  \end{alignat*}\endgroup{}

  The theory $\Th_2$ is the extension of $\Th_{1}$ with additional maps
  \[ \inner{c} : \iTy_i \to \iTm(\inner{\bm{U}}_i) \]
  that are inverses to the maps $\inner{El}$.
  \defiEnd{}
\end{defi}

The two theories $\Th_1$ and $\Th_2$ may seem equivalent, since they their syntaxes can be identified.
\begin{prop}
  The map $\Init_{\Th_1} \to \Init_{\Th_2}$ is an isomorphism (of models of $\Th_1$).
\end{prop}
\begin{proof}
  It suffices to show that the maps $\inner{El}$ have inverses in $\Init_{\Th_1}$, which follows from normalization for the type theory $\Th_1$.
\end{proof}

However, while $\Init_{\Th_1} \to \Init_{\Th_2}$ is an isomorphism, this is not the case for the map $\Th_1 \to \Th_2$ between their coclassifying $(\reppre)$-CwFs.
We will show that $\Th_1$ satisfies external univalence while $\Th_2$ does not.

\begin{defi}\label{def:homotopy_relations_univ}
  We define homotopy relations on both $\Th_1$ and $\Th_2$, analogously to the homotopy relations defined in~\cref{exa:homotopy_relations_th_id}.
  \begin{alignat*}{3}
    & A \sim_{\ity_i} B && \triangleq{ } && \iEquiv(A,B), \\
    & x \sim_{\itm_i(A)} y && \triangleq{ } && \iTm(\iId(A,x,y)).
  \end{alignat*}
\end{defi}

\begin{lem}
  The theory $\Th_1$ satisfies external univalence with respect to the homotopy relations defined in~\cref{def:homotopy_relations_univ}.
\end{lem}
\begin{proof}
  We use~\cref{thm:external_univalence_from_refleqv}.
  \Cref{asm:congruence_operations} and~\cref{asm:basic_center_and_all_paths} can be checked in the same way as in~\cref{ssec:application_type_theory}; we omit the proof.

  We equip $\CPreReflGraph(\Th_1)$ with the structure of an internal model of $\Th_1$.

  The sorts of types and terms are interpreted as follows:
  \begingroup{}\allowdisplaybreaks{}
  \begin{alignat*}{3}
    & \sem{\ity_i}_{E}(\gamma_e) && \triangleq{ }
    && \lambda A\ B \mapsto \iEquiv_i(A,B), \\
    & \sem{\itm_i(A)}_{E}(\gamma_e) && \triangleq{ }
    && \lambda x\ y \mapsto \iTm_i(\sem{A}_{E}(\gamma_e,x,y)), \\
    & \sem{\ity_i}_{R}(\gamma_r) && \triangleq{ }
    && \lambda A\ E \mapsto \isRefl_i(E), \\
    & \sem{\itm_i(A)}_{R}(\gamma_r) && \triangleq{ }
    && \lambda x\ p \mapsto \iTm_i(\sem{A}_{R}(\gamma_r,x,p)),
  \end{alignat*}\endgroup{}
  where $\iEquiv_i(A,B)$ is the sort of relational equivalences between $A$ and $B$ in $\Ty_i$, and
  \begin{alignat*}{3}
    & \isRefl(E) && \triangleq{ } && (P : \forall a \to \iTm_i(E(a,a)) \to \iTy_i) \\
    &&&&& \times (\forall a \to \isContr((p : E(a,a)) \times P(p)))
  \end{alignat*}
  is the sort of reflexivity structures in $\Ty_i$ over an equivalence $E : \iEquiv_i(A,A)$.

  The universes are interpreted as follows; note that the relation between elements of the universe is not equivalence, but identification of the codes:
  \begingroup{}\allowdisplaybreaks{}
  \begin{alignat*}{3}
    & \sem{\inner{\bm{U}}_i}_{E}(\gamma_e) && :{ }
    && \iEquiv_{i+1}(\inner{\bm{U}}_i, \inner{\bm{U}}_i), \\
    & \sem{\inner{\bm{U}}_i}_{E}(\gamma_e) && \triangleq{ }
    && \lambda A\ B \mapsto \iId_{\inner{\bm{U}}_i}(A,B), \\
    & \sem{\inner{\bm{U}}_i}_{R}(\gamma_r) && :{ }
    && \isRefl_{i+1}(\sem{\inner{\bm{U}}_i}_{E}(\gamma_e)), \\
    & \sem{\inner{\bm{U}}_i}_{R}(\gamma_r) && \triangleq{ }
    && \lambda A\ e \to \iId(e, \irefl_{\inner{\bm{U}}_i}(A)), \\
    & \sem{\inner{El}_i(A)}_{E}(\gamma_e) && :{ }
    && \iEquiv_{i}(\inner{El}(A(\gamma_l)), \inner{El}(A(\gamma_r))), \\
    & \sem{\inner{El}_i(A)}_{E}(\gamma_e) && \triangleq{ }
    && \mathsf{idtoeqv}(\sem{A}_{E}(\gamma_e)), \\
    & \sem{\inner{El}_i(A)}_{R}(\gamma_r) && :{ }
    && \isRefl(\mathsf{idtoeqv}(\sem{A}_{E}(\gamma_e))), \\
    & \sem{\inner{El}_i(A)}_{R}(\gamma_r) && \triangleq{ }
    && \text{(follows from ${\sem{A}_{R}(\gamma_{r}) : \iId(\sem{A}_{E}(\gamma_{e}),\irefl(A(\gamma)))})$}.
  \end{alignat*}\endgroup{}

  Interpreting the lifting operations is straightforward, and the identity types can be interpreted as in~\cref{ssec:application_type_theory}.

  This completes the construction of the internal model of $\Th_1$ in $\CPreReflGraph(\Th_1)$.
  By~\cref{thm:external_univalence_from_refleqv}, the theory $\Th_1$ satisfies external univalence.
\end{proof}

\begin{lem}\label{lem:coding_refutes_external_univalence}
  The theory $\Th_2$ does not satisfies external univalence with respect to the homotopy relations defined in~\cref{def:homotopy_relations_univ}.
\end{lem}
\begin{proof}
  Assuming that $\Th_2$ satisfies external univalence, the action of $\inner{c}$ on paths determines a map
  \[ (A,B : \iTy_0) \to \iEquiv(A,B) \to \iTm(\iId(\inner{\bm{U}}_0,\inner{c}(A),\inner{c}(B))) \]
  in the $(\reppre)$-CwF $\Th_2$.

  This map can be interpreted into any model of $\Th_2$; in particular it can be interpreted in the standard model $\CSet$.
  Since we can assume without loss of generality that there exists two sets $A$ and $B$ that are equivalent but not equal, $\Th_2$ cannot satisfy external univalence.
  (In a set-theoretic metatheory, we can choose $A = \{1\}$ and $B = \{2\}$.
  In a type-theoretic metatheory, we can add new redundant codes to the universes of $\CSet$, \eg{} we define $\UU'_0 \triangleq \UU_0 + \{A,B\}$ with $\El(A) = \El(B) = \Unit$.)
\end{proof}

\begin{thm}
  The SOGATs of type theories with a $\Nat$-indexed hierarchy of universes à la Tarski, weak identity types and any selection of type structures among:
  \begin{itemize}
  \item $\Unit$-types;
  \item $\Sigma$-types;
  \item $\Pi$-types;
  \item empty types $\Init$, in the presence of $\Unit$;
  \item boolean types, in the presence of $\Unit$ and $\Init$;
  \item natural number types, in the presence of $\Unit$ and $\Init$;
  \item either weak or strict computation rules for the above type structures;
  \item any number of axioms, \ie{} postulated elements of closed representable sorts types, such as the univalence axiom;
  \item the Uniqueness of Identity Proofs principle;
  \end{itemize}
  all satisfy external univalence.
  \qed{}
\end{thm}


%% file: conclusion.tex
\section{Future work}\label{sec:conclusion}

\subsection{Semantic study of external univalence}

In this paper, we have defined external univalence as a property of the $(\reppre)$-CwF $\Th$.
The advantage of this approach is that we did not have to consider the semantics of $\Th$ at all.

In future work, we plan to study how external univalence for a SOGAT $\Th$ is related to properties of the category $\CMod_\Th$ of models of $\Th$.
In particular, we plan to show~\cref{clm:external_univalence_equiv_modelcat}, which says that $\Th$ satisfies external univalence exactly when the category $\CMod_\Th^\cxl$ of contextual models of $\Th$, equipped with suitable classes of maps, is a left semi-model category.
We also plan to show that the notion of Morita equivalence between type theories, which was introduced by \textcite{IsaevMoritaEquivalences}, can be captured at the level of the $(\reppre)$-CwFs: given a morphism $\Th_1 \to \Th_2$ of SOGATs that satisfy external univalence, the adjunction between the left semi-model categories $\CMod_{\Th_1}^\cxl$ and $\CMod_{\Th_2}^\cxl$ is a Quillen adjunction if and only if $\Th_1 \to \Th_2$ preserves the identity types, and a Quillen equivalence if and only if $\Th_1 \to \Th_2$ is additionally a weak equivalence in $\CCwf_{\repinfty}$.

\subsection{Strictification}

It is rather inconvenient that our methods only equip $\Th$ with weakly stable identity types.
It does not seem possible to equip $\Th$ with strictly stable identity types in general.
Instead, we may want to strictify the identity types, \ie faithfully embed $\Th$ into a $(\reppre)$-CwF with strictly stable identity types.

\begin{conj}
  Let $\Th$ be a SOGAT equipped with homotopy relations.
  If $\Th$ satisfies external univalence, then there exists a $(\reppre,\Id)$-CwF $\CC$ with strictly stable identity types and a $(\reppre)$-CwF morphism $\Th \to \CC$ that weakly preserves identity types and is essentially surjective on types and terms.
  \defiEnd{}
\end{conj}
The known strictification methods cannot be applied to this situation.
For example, the local universes method~\parencite{LocalUniversesModel} requires more $\Pi$-types than available in $\Th$.

It is however possible to strictify the identity types in the special case of first-order generalized algebraic theories without equations.
\begin{thm}
  Let $\Th$ be a first-order generalized algebraic theory without equations, \ie{} an $\{I^\ty,I^\tm\}$-cellular $\lexpre$-CwF.

  If $\Th$ satisfies external univalence with respect to a choice of homotopy relations, then $\Th$ can be equipped with strictly stable identity types satisfying saturation with respect to the homotopy relations.
\end{thm}
\begin{proof}
  This is proven for $\{I^\ty,I^\tm\}$-cellular CwFs (without $\Sigma$) in~[\cite[Theorems 1 and 2]{StrictificationWeaklyStable}], but the methods can be generalized to $\{I^\ty,I^\tm\}$-cellular $\lexpre$-CwFs.
\end{proof}

\subsection{Embedding theories into richer models}
While we have shown that it is possible to transport structures over homotopies for any SOGAT that satisfies external univalence, this only holds for structures that are expressible in the language of the $(\reppre)$-CwF $\Th$.
This language is not sufficiently expressive for all applications.

It would be desirable to conservatively embed $\Th$ into a richer language that allows for the specification of additional structures and properties.
A good candidate for this richer language is (any variant of) Homotopy Type Theory.

\begin{conj}[Weak embedding into HoTT]\label{conj:weak_emb_hott}
  Let $\Th$ be a SOGAT equipped with homotopy relations satisfying external univalence.
  Then there exists a model $\CC$ of (some variant of) HoTT equipped with an univalent internal model of $\Th$ such that the induced $(\reppre)$-morphism $\Th \to \CC$ is essentially surjective on terms.
  \defiEnd{}
\end{conj}

\begin{conj}[Strict embedding into HoTT]\label{conj:str_emb_hott}
  Let $\Th$ be a SOGAT equipped with homotopy relations satisfying external univalence.
  Then there exists a model $\CC$ of (some variant of) HoTT equipped with an univalent internal model of $\Th$ such that the induced $(\reppre)$-morphism $\Th \to \CC$ is bijective on terms.
  \defiEnd{}
\end{conj}

These conjectures should be seen as $\infty$-categorical variants of the following $1$-categorical theorem:
\begin{thm}
  Let $\Th$ be any SOGAT.
  Then there exists a model $\CC$ of extensional type theory equipped with an internal model of $\Th$ such that the induced $(\reppre)$-morphism $\Th \to \CC$ is bijective on terms.
\end{thm}
\begin{proof}[Proof sketch]
  We define $\CC$ as the presheaf topos $\psh{\Th}$.
  Then the Yoneda embedding $\yo : \Th \to \psh{\Th}$ is a pseudo-morphism of $(\reppre)$-CwFs, and bijective on terms.
  Relying on the fact that $\Th$ is $\{I^\ty,I^{\repty},I^\tm,E^\tm\}$-cellular, we can construct a strict replacement $y : \Th \to \psh{\Th}$ of the Yoneda embedding, along with a $2$-cell $\yo \cong y$.
  This strict replacement is also bijective on terms.
\end{proof}

An $\infty$-categorical version of this argument gives intuition for why~\cref{conj:weak_emb_hott} should hold.
Indeed, when a SOGAT $\Th$ satisfies external univalence, the $(\reppre)$-CwF $\Th$ is a $(\repinfty)$-CwF, which should correspond to some $\infty$-category with representable maps.
By the $\infty$-categorical Yoneda lemma, we can faithfully embed this $\infty$-category into an $\infty$-topos of $\infty$-categorical presheaves.
We can finally interpret HoTT into this $\infty$-topos.
Unfortunately, turning this this informal proof idea into a proper proof is not straightforward.

When studying the computational properties of type theories, such as canonicity and normalization properties, it is typical to rely on the interpretation of extensional type theory into some presheaf categories.
A solution of these conjectures would provide a good setting for the study of some homotopical properties of type theories, such as homotopy canonicity and normalization up to homotopy.
